\documentclass[11pt]{article}
\usepackage{rotating,amsmath,amssymb,amsfonts,amsthm,longtable}
\usepackage{epsfig,mathrsfs,natbib,color,graphics,bm,dsfont,subfigure}
\usepackage{multirow}
\def\boxit#1{\vbox{\hrule\hbox{\vrule\kern6pt \vbox{\kern6pt#1\kern5pt}
\kern6pt\vrule}\hrule}}

\newcommand{\bE}{{\boldsymbol E}}

\newcommand{\be}{{\boldsymbol e}}

\newcommand{\bV}{{\boldsymbol V}}
\newcommand{\bG}{{\boldsymbol G}}

\newcommand{\bq}{{\boldsymbol q}}

\newcommand{\bX}{{\boldsymbol X}}

\newcommand{\mN}{\mathbb{N}}
\newcommand{\mR}{\mathbb{R}}
\newcommand{\mE}{\mathbb{E}}

\newcommand{\mG}{{\cal G}}

\newcommand{\bmu}{{\boldsymbol \mu}}

\newcommand{\bpsi}{{\boldsymbol \psi}}
\newcommand{\bSigma}{{\boldsymbol \Sigma}}

\newcommand{\tmE}{\widetilde{\cal E}}
\newcommand{\tbE}{\widetilde{\boldsymbol E}}
\newtheorem{theorem}{Theorem}[]
\newtheorem{lemma}{Lemma}[]

\def\ANNALS{{\it Annals of Statistics}}

\def\JASA{{\it Journal of the American Statistical Association}}

\def\JCGS{{\it Journal of Computational and Graphical Statistics}}

\def\JCGS{{\it Journal of Computational and Graphical Statistics}}

\def\JASA{{\it Journal of the American Statistical Association}}
\def\ANNALS{{\it Annals of Statistics}}

\def\ANNALS{{\it Annals of Statistics}}

\def\JASA{{\it Journal of the American Statistical Association}}

\def\JMLR{{\it Journal of Machine Learning Research}}
\begin{document}
\bibliographystyle{asa}

\title{Fast Bayesian Integrative Learning of Multiple Gene Regulatory Networks for Type 1 Diabetes}

\author{Bochao Jia, Faming Liang and the TEDDY Study Group
\thanks{
 Faming Liang is Professor,
 Department of Statistics, Purdue University, West Lafayette, IN 47907;
 Email: fmliang@purdue.edu.
 Bochao Jia is Research Scientist, Eli Lilly and Company, Lilly Corporate Center, 
 Indianapolis, IN 46285; Email: jia\_bochao@lilly.com.  
 }
 }

\maketitle

\begin{abstract}
 Motivated by the need to study the molecular mechanism underlying Type 1 Diabetes (T1D) with 
 the gene expression data collected from both the patients and healthy controls at multiple time points, 
  we propose an innovative method for jointly estimating multiple dependent Gaussian graphical models. 
  Compared to the existing methods, the proposed method has
a few significant advantages. First, it includes a meta-analysis procedure to explicitly integrate
 information across distinct conditions. In contrast, the existing methods often integrate
 information through prior distributions or penalty function, which is usually less efficient.
Second, instead of working on original data, the Bayesian step of the proposed method works on
edge-wise scores, through which the proposed method avoids to invert high-dimensional covariance matrices 
 and thus can perform very fast. The edge-wise score forms an equivalent measure
of the partial correlation coefficient and thus provides a good summary for the graph structure
information contained in the data under each condition. Third, the proposed method can provide 
 an overall uncertainty measure for the edges detected in multiple graphical models, while
the existing methods only produce a point estimate or are feasible for very small size problems.
We prove consistency of the proposed method under mild conditions and illustrate its performance 
 using simulated and real data examples. The numerical results indicate the superiority
of the proposed method over the existing ones in both estimation accuracy and computational
efficiency. Extension of the proposed method to joint estimation of multiple mixed graphical
 models is straightforward.
  
\vspace{1mm}


\end{abstract}

 \newpage 

{\centering \section{Introduction}}

 Type 1 diabetes (T1D) is one of the most common autoimmune diseases. 
 The Environmental Determinants of Diabetes in the Young (TEDDY) study is designed to 
 identify environmental exposures triggering islet autoimmunity and T1D in genetically high-risk children. 
 A large dataset has been collected through the study, including clinical data, genetic data and 
 demographical data. 
 While great efforts have been made for identifying the genetic and environmental factors 
 that contribute to the etiology of the disease, the molecular mechanism underlying the disease  
 is still far from understanding. To enhance our understanding to the molecular mechanism, 
 this work aims to learn a gene regulatory network (GRN) by integrating the gene expression data measured  
 from both the patients and healthy controls at multiple time points. 
 Figure \ref{dis1} shows the structure of the data, where the gene expression was  
 measured for each of the case and control children at nine time points within    
 four years of age. How to integrate the data collected under the 18 distinct conditions 
 has posed a great challenge on the current statistical methods.   

\begin{figure}[htbp]
\centering
\includegraphics[scale=0.5]{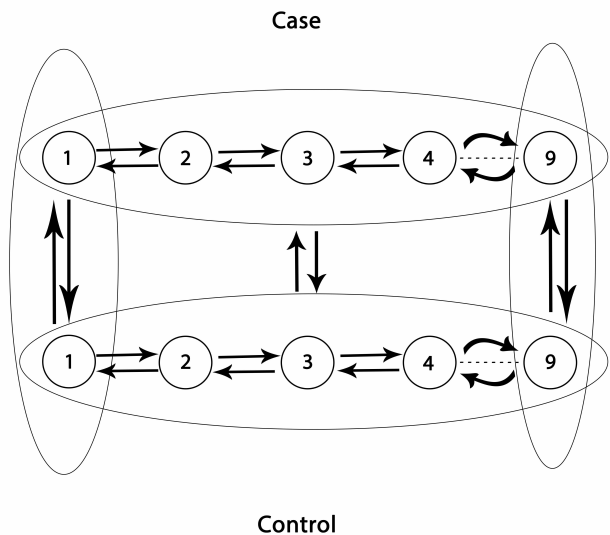}
\caption{Structure of the T1D data considered in the paper, where the numbers represent 9 time points 
  at which gene expression data were collected and the arrows represent joint estimation of 
  Gaussian graphical models by integrating the data across different time points 
 and case-control groups.}\label{dis1}
\end{figure}

 During the past decade, a variety of approaches have been proposed for estimating GRNs 
 with the data collected under both scenarios, single condition or multiple distinct conditions. 
 For the former, Gaussian graphical models (GGMs) have become widely used,  
 see e.g., Meinshausen and B\"uhlmann (2006), 
 Yuan and Lin (2007), Friedman, Hastie and Tibshirani (2008), Wang (2015), and Liang, Song and Qiu (2015). 
 For the latter, the existing approaches can be roughly grouped into two categories, namely, 
 regularization and Bayesian. 

 The regularization approaches works with some specific penalty functions that enhance 
 the shared structure of the graphical models. For example, 
 Guo et al. (2011) employed a hierarchical penalty that targets the removal of common zeros 
 in the precision matrices across conditions.  
 Danaher, Wang and Witten (2014) employed penalized fused lasso or group lasso penalties that 
 encourage shared elements of the precision matrices.   
 A shortcoming of these approaches is that they assume the observations under different conditions 
 are independent. This is hard to be satisfied for the temporal data, where the observations 
 were taken from the same subject at multiple time points. 
 For example, under either the case or control, the T1D data are temporal.  
 To address this issue, Zhou, Lafferty and Wasserman (2010)
 and Qiu et al. (2015) proposed to model the temporal data in a high-dimensional time series and 
 then estimate the time varying graphical structure using a nonparametric method by assuming that the 
 covariance changes smoothly over time. These approaches usually require the time series to be 
 fairly long, say, 50 or longer. 
 Another approach that allows for the data dependence is 
 proposed by Cai et al. (2011), which is based on the constrained $l_{\infty}/l_1$ minimization
 of the precision matrices by abandoning the use of the likelihood function.  

 As an analog to regularization approaches, Bayesian approaches enhance the shared structure of 
 multiple graphical models by employing some specific priors. For example, 
 Peterson, Stingo and Vannucci (2015) and Shaddox et al. (2016) link the estimation of graph structures via 
 a Markov random field prior which encourages common edges. However, since this method involves repeated 
 calculations of concentration matrices (i.e., inverse of covariance matrices), 
 it is only applicable when the graph is not very large.   
 To accelerate computation, Lin et al. (2017) proposed a Bayesian analog of 
 the neighborhood selection method (Meinshausen and B\"uhlmann, 2006) to learn 
 the structure of multiple graphical models with the Markov random field prior.  
 For the single graph case, Wang (2015) proposed a method to scale up the simulation based on 
 continuous spike and slab priors and provided timing results for $p$, the number of variables,
 up to 250.  However, since the method still involves repeated calculations of concentration matrices, 
 the computational cost is still not acceptable when $p$ is very large. 
 
 There is a major shortcoming with the existing methods: they try to integrate the data collected 
 under distinct conditions through penalty function or prior distributions. 
 It is hard to justify that the information has been integrated efficiently in this way. 
 Another shortcoming with the existing methods is the lack of uncertainty measure
 for the resulting graph estimator.
 The regularization methods produce only a point estimator for the graphical models
 and fail to provide an uncertainty measure for the point estimator. The Bayesian method by
 Peterson, Stingo and Vannucci (2015) is able to provide an uncertainty measure for its estimator, but it
 works only for small graphs. The method by Lin et al. (2017) is also Bayesian, but it is hard
 to provide a global uncertainty measure for the resulting graph, as the method works node-wisely.
 From our point of view, estimation of graphical models is essentially a multiple hypothesis testing
 problem, i.e., simultaneously testing the existence of a large number of candidate edges.
 An uncertainty measure, e.g., false discovery rate, can be naturally provided for
 the resulting graph estimate.

 In this paper, we propose a fast Bayesian integrative analysis (FBIA) method for 
 jointly estimating multiple Gaussian graphical models.
 The FBIA method consists of a few steps,
 including $\psi$-score calculation, Bayesian clustering, data integration,
 and multiple hypothesis tests. The $\psi$-score calculation step is to calculate a
 $\psi$-score for each edge of the multiple graphs using
 the $\psi$-learning algorithm developed by Liang et al. (2015). The $\psi$-score, which forms an
 equivalent measure of the partial correlation coefficient, provides a good summary for the graph
 structure information contained in the data under each condition.
 The Bayesian clustering step identifies possible status changes of each edge across distinct
 conditions. Based on the possible changes identified in the Bayesian clustering step, a meta-analysis
 method is applied to integrate data information across distinct conditions. Finally,
 a multiple hypothesis test is applied to classify the integrated $\psi$-scores to two groups, which correspond
 to existence and non-existence of edges, respectively. 

 Compared to the existing methods, FBIA has a few significant advantages.  
 First, FBIA includes an meta-analysis procedure to explicitly integrate information across distinct classes.
 In contrast, the existing methods often integrate information through prior distributions or penalty functions,
 which is usually less efficient.
 Second, unlike the traditional Bayesian methods which attempt to model the original data, 
 the proposed method models $\psi$-scores, which avoids to invert high-dimensional covariance matrices and
 thus can perform very fast. 
 Third, the proposed method can provide an uncertainty measure for the edges detected in the multiple graphical models
 and the difference of edges detected in the graphical models under any two distinct conditions,
 while the existing methods only produce a point estimate or are feasible for very small size problems.
 We illustrate the performance of the proposed method using simulated and T1D data examples.
 The numerical results indicate the superiority of the proposed method over the existing ones.

 The rest of the paper is organized as follows. 
 In Section 2, we describe the FBIA method and establish its consistency.
 In Section 3, we illustrate the FBIA method using simulated data along with comparisons 
 with some existing methods. 
 In Section 4, we apply the proposed method to T1D data. 
 In Section 5, we conclude the paper with a brief discussion.

\section{A Fast Bayesian Integrative Analysis Method}

 The FBIA method consists of a few steps, including $\psi$-score calculation, 
 Bayesian clustering and meta-analysis, and joint edge detection, which are described 
 in sequel as follows. At the end of this section, we discuss the consistency and 
 parameter setting of the method. 

\subsection{$\psi$-Score Calculation} 

 This step is to convert original data information into edge-wise scores, which 
 facilitates the followed Bayesian clustering and meta-analysis. 
 Suppose that we have a dataset of $p$ variables observed under $K$ distinct conditions. 
 Let $\bX^{(k)}=(\bX_1^{(k)},\ldots, \bX_{n_k}^{(k)})^T$ denote the dataset 
 observed under condition $k$, where $n_k$ denotes the sample size  
 under condition $k$; and $\bX_i^{(k)}=(X_{i1}^{(k)},\ldots, X_{ip}^{(k)})^T$ 
 is a $p$-dimensional random vector distributed according to 
 the multivariate normal distribution $N(\bmu_k, \bSigma_k)$, and $\bmu_k$ and $\bSigma_k$ are the mean and 
 covariance matrix of the distribution, respectively. The sample size $n_k$ is not necessarily 
 the same for all conditions. Without loss of generality, we assume that  
 $\bmu_k$ is a zero vector for all $k$. With a slight abuse of notation, we let $X_1,\ldots, X_p$ denote the 
 $p$ variables that are common for all $K$ datasets. Let $V=\{1,2,\ldots,p\}$ denote
 the index set of the variables. 
 
 In this paper, we adopt the $\psi$-learning algorithm (Liang et al., 2015) to convert 
 each dataset $\bX^{(k)}$ to edge-wise scores independently. 
  Since the essence of learning the structure of Gaussian graphical models (GGMs)
  is to find the pairs of the 
  variables for which the partial correlation coefficient is equal to zero, 
  a correlation screening procedure can be applied to reduce the size of 
  conditioning set used for calculating the partial correlation coefficient.
  Let $\tilde{\psi}_{ij}$ denote the partial correlation coefficient calculated with the 
  reduced conditioning set $S_{ij}$, i.e., $\tilde{\psi}_{ij}=\rho_{ij|S_{ij}}$.  
  Under the faithfulness condition (see e.g., B\"uhlmann and van de Geer, 2011), 
  i.e., assuming that all the conditional independence 
  among the variables  $X_1,\ldots, X_p$ can be read off from the graphical concept of separations,  
  Liang et al. (2015) showed that $\tilde{\psi}_{ij}$ is equivalent to 
  $\rho_{ij|V\setminus \{i,j\}}$ in learning the structure of GGMs
  in the sense that  
  \[
  \tilde{\psi}_{ij}=0 \Longleftrightarrow \rho_{ij|V\setminus \{i,j\}}=0.
  \]  
  Further, under mild conditions for the sparsity of the underlying GGM, 
   Liang et al. (2015) showed that the size of $S_{ij}$ can be bounded by 
   $n/\log(n)$. Therefore, the $\psi$-learning algorithm has 
   successfully reduced the problem of partial correlation coefficient 
   calculation from a high-dimensional setting to a low-dimensional one. 
   Note that $\rho_{ij|V\setminus \{i,j\}}$ is even not calculable 
   when $p$ is greater than $n$. 
   In summary, the $\psi$-learning algorithm consists of the following two steps to 
   calculate the $\psi$-partial correlation coefficients for each dataset $\bX^{(k)}$: 

\begin{itemize}
 \item[(a)]  (Correlation screening) Determine the reduced neighborhood for each variable $X_i$. 
\begin{itemize}
\item[(i)] Conduct a multiple hypothesis test to identify the pairs of variables 
 for which the empirical correlation coefficient 
 is significantly different from zero. This step results in a so-called empirical correlation network.

\item[(ii)] For each variable $X_i$, identify its neighborhood in the empirical correlation network, 
 and reduce the size of the neighborhood to 
 $O(n/\log(n))$ by removing the variables having lower correlation (in absolute value) with $X_i$. 
 This step results in a so-called reduced correlation network.
\end{itemize}

\item[(b)] ($\psi$-calculation) For each pair of variables $X_i$ and $X_j$, 
 identify the separator $S_{ij}^{(k)}$ based on the reduced correlation network 
 and calculate $\tilde{\psi}_{ij}^{(k)}=\rho_{ij|S_{ij}^{(k)}}$, 
 where $\rho_{ij|S_{ij}^{(k)}}$ denotes the partial correlation coefficient 
 of $X_i$ and $X_j$ calculated for the dataset $\bX^{(k)}$ 
 conditional on the variables $\{X_l: l \in S_{ij}^{(k)} \}$.   
\end{itemize}
 
To facilitate followed analysis, we further convert the $\psi$-partial correlation coefficients to 
$\psi$-scores via the Fisher's transformation 
\begin{equation}\label{score}
\psi_{ij}^{(k)}=\frac{\sqrt{n-|S_{ij}^{(k)}|-3}}{2}\log\left[\frac{1+\tilde{\psi}_{ij}^{(k)}}{1-\tilde{\psi}_{ij}^{(k)}}\right],
\end{equation}
which approximately follows the standard normal distribution 
under the null hypothesis $H_0: \rho_{ij|V\setminus\{i,j\}} =0$. 
Therefore, the $\psi$-score can be used as a test statistic for identifying non-zero partial correlation 
coefficients and thus the structure of Gaussian graphical models, and 
 $n-|S_{ij}^{(k)}|-3$ can be viewed as the effective sample size of the test statistic.     
Compared to sure independence screening (Luo, Song and Witten, 2015), the correlation screening step 
 often leads to a smaller neighborhood for each variable and thus, as implied by (\ref{score}), 
 helps to improve the power of the proposed method. 

 Since the Gaussian graphical model is symmetric, we have a total of $p(p-1)/2$ $\psi$-scores to calculate 
 for each dataset $X^{(k)}$. For convenience, we re-arrange all the $\psi$-scores for the $K$ datasets 
 into a $N\times K$ matrix $(\psi_l^{(k)})$ with $l=1,2,\ldots, N$, $k=1,2,\ldots,K$, and $N=p(p-1)/2$.

\subsection{Bayesian Clustering and Meta-Analysis}

 Consider the $\psi$-scores $(\psi_l^{(k)})$, where each pair $(l,k)$ corresponds to one 
 candidate edge in the graph $k$. 
 Let $e_l^{(k)}$ be the indicator for the status of the 
 edge $l$ in the underlying graph $k$; $e_l^{(k)}=1$ if the edge exists and 0 otherwise. 
 The $e_l^{(k)}$'s work as the latent variables in FBIA.  
 Conditioned on $e_l^{(k)}$, we assume that $\psi_l^{(k)}$'s are mutually independent and follow  
 a two-component mixture Gaussian distribution,   
 \begin{equation}\label{plugin}
     p(\psi_{l}^{(k)}|e_{l}^{(k)})=\left\{\begin{array}{ll}
                  N(\mu_{l0},\sigma_{l0}^2),&\textrm{if $e_{l}^{(k)}=0$},\\
                  N(\mu_{l1},\sigma_{l1}^2),&\textrm{if $e_{l}^{(k)}=1$},
                \end{array}\right.
\end{equation}
 for $l=1,2,\ldots, N$ and $k=1,2,\ldots, K$.  
  When $e_{l}^{(k)}=0$, $\psi_{l}^{(k)}$'s have a value close to 0,
  otherwise, $\psi_{l}^{(k)}$'s might have a large negative or positive value depending on 
  the sign of the partial correlation coefficient. Under the assumption that the structure of 
  the GGM changes only slightly under adjacent conditions, 
  it is reasonable to assume that for each $l$, the sign of $\psi_{l}^{(k)}$'s are not changed when the 
  edge exists; therefore, $\psi_{l}^{(k)}$'s can be modeled by  
  a two-component mixture Gaussian distribution. In some cases, e.g., when $K$ grows, 
  a three-component mixture Gaussian distribution might be needed, which allows us 
  to handle the scenario when an edge is included in multiple graphs, but with a sign 
  difference in the partial correlation. 
  The derivation under this scenario is given in Appendix A, which is just a simple 
  extension of the deviation presented below.  
  
  Regarding the 2-component mixture distribution (\ref{plugin}), we further note that   
  $\mu_{l0}$ can be simply set to 0 considering the physical mean of $\psi$-scores. 
  However, as shown below, this general setup does not cause any computational difficulty. 
  Essentially, we have formulated the problem as a clustering problem, grouping 
  $\psi_l^{(k)}$ to up to two different clusters. For the case of 3-component mixture distribution,
  this is similar.

 Let $\bpsi_l=(\psi_{l}^{(1)},...,\psi_{l}^{(K)})$ and 
 $\be_{l}=(e_{l}^{(1)},...,e_{l}^{(K)})$. Conditioned on $\be_l$, the joint likelihood function of 
 $\bpsi_l$ is given by  
\begin{equation} \label{jointeq}
 p(\bpsi_{l}|\be_{l},\mu_{l0},\sigma_{l0}^2,\mu_{l1},\sigma_{l1}^2) 
  = \prod_{\{k:e_{l}^{(k)}=0\}}\phi(\psi_{l}^{(k)}|\mu_{l0},\sigma_{l0}^2)\prod_{\{k:e_{l}^{(k)}=1\}}
  \phi(\psi_{l}^{(k)}|\mu_{l1},\sigma_{l1}^2), 
\end{equation}
 where $\phi(.|\mu,\sigma^2)$ is the density function of the Gaussian distribution 
 with mean $\mu$ and variance $\sigma^2$.
 Taking a product of (\ref{jointeq}) over $l=1,2,\ldots, N$, we will have the joint distribution of all  
 $\psi$-scores $(\bpsi_l^{(k)})$ conditioned on $\be_l$'s and other parameters. 
 Then, using the Bayes theorem, $e_l^{(k)}$'s can be inferred  
 with an appropriate priors of $\be_l$'s and other parameters. 
 For example, the Markov random field prior used in Peterson, Stingo and Vannucci (2015), Shaddox et al. (2016)
  and Lin et al. (2017) can again be used here as the prior of $\be_l$'s. 
 In this case, the posterior distribution can 
 be sampled from using a Markov chain Monte Carlo algorithm (see e.g., Liang et al., 2010).  
 
 Instead of specifying a joint prior distribution for all $\be_l$'s, we assume that 
 $\be_l$'s are {\it a priori} independent for different $l$'s, as we believe that the 
 neighboring dependence of the Gaussian graphical network can be accounted for  
 in calculation of the $\psi$-scores. 
 To enhance shared edges among distinct conditions, we consider two types 
 of priors for $\be_l$'s, namely, temporal prior and spatial prior, with borrowed terms from  
 geostatistics. The former is suitable for the scenario that 
 the networks or precision matrices $\Omega^{(k)}$, $k=1,2,\ldots, K$, 
 evolve sequentially along with the index $k$. In this scenario, it is quite common to consider 
 the index $k$ as the time of experiments.    
 The latter is suitable for the scenario that the networks or 
 precision matrices  $\Omega^{(k)}$, $k=1,2,\ldots, K$, evolve independently from a 
 common structure. For example, the genetic networks constructed using the gene expression data 
 collected from different tissues are more likely developed from a common structure.   
 Xie, Liu and Valdar (2016) have developed a graphical EM algorithm to deal with the data 
 under this scenario. 
 
 \subsubsection{Temporal Prior}

 To enhance the similarity of the networks between adjacent conditions, we let $\be_l$ be subject to 
 the following prior distribution 
\begin{equation}\label{prior}
  p(\be_l|q) =q^{\sum_{i=1}^{K-1} c_l^{(i)} } (1-q)^{\sum_{i=1}^{K-1} (1-c_l^{(i)}) },
\end{equation}
 where $c_l^{(i)}=|e_l^{(i+1)}-e_l^{(i)}|$ indicates the change of the status of the edge $l$ from 
 condition $i$ to condition $i+1$, and $q$ is a prior hyperparameter representing the prior probability 
 of edge status changes. In this paper, we assume that $q$ follows a beta distribution $Beta(a_1,b_1)$, 
 where $a_1$ and $b_1$ are pre-specified parameters.  
 Further, we let $\mu_{l0}$ and $\mu_{l1}$ be subject to an improper uniform distribution, i.e., 
 $\pi(\mu_{l0}) \propto 1$ and $\pi(\mu_{l1}) \propto 1$, and let $\sigma_{l0}^2$ and $\sigma_{l1}^2$  
 be subject to an inverted-gamma distribution, i.e., 
 $\sigma_{l0}^2, \sigma_{l1}^2  \sim IG(a_2,b_2)$, where $a_2$ and $b_2$ are pre-specified constants. 
 Then the joint posterior distribution of $(\be_l,\mu_{l0},\sigma_{l0}^2, \mu_{l1}, \sigma_{l1}^2, q)$ 
 is given by 
 \[
 \pi(\be_l,\mu_{l0},\sigma_{l0}^2, \mu_{l1}, \sigma_{l1}^2, q|\bpsi_l)\propto 
 p(\bpsi_l|\be_l,\mu_{l0},\mu_{l1},\sigma_{l0}^2,\sigma_{l1}^2) \pi(\mu_{l0},\sigma_{l0}^2, \mu_{l1}, \sigma_{l1}^2) 
 \pi(\be_l|q) \pi(q),
 \]
 where $\pi(\cdot)$'s denote the respective prior distributions.  
 After integrating out the parameters $\mu_{l0}$, $\sigma_{l0}^2$, $\mu_{l1}$, $\sigma_{l1}^2$ and $q$, we have 
 the marginal posterior distribution of $\be_l$ given by 
\begin{equation} \label{posteq1}
\begin{split}
\pi(\be_l|\bpsi_l) & \propto \frac{\Gamma(a_1+k_1)\Gamma(b_1+k_2)}{\Gamma(a_1+k_1+b_1+k_2)} \\
&\times \frac{1}{\sqrt{n_0}}(\frac{1}{\sqrt{2\pi}})^{n_0}\Gamma(\frac{n_0-1}{2}+a_2)
\left[\frac{1}{2}\sum_{\{k:e_{l}^{(k)}=0\}}(\psi_l^{(k)})^2-\frac{(\sum_{\{k:e_{l}^{(k)}=0\}}\psi_l^{(k)})^2}{2n_0}+b_2\right]^{-(\frac{n_0-1}{2}+a_2)}\\
&\times \frac{1}{\sqrt{n_1}}(\frac{1}{\sqrt{2\pi}})^{n_1}\Gamma(\frac{n_1-1}{2}+a_2)
\left[\frac{1}{2}\sum_{\{k:e_{l}^{(k)}=1\}}(\psi_l^{(k)})^2-\frac{(\sum_{\{k:e_{l}^{(k)}=1\}}\psi_l^{(k)})^2}{2n_1}+b_2\right]^{-(\frac{n_1-1}{2}+a_2)}\\
& = (H) \times (I) \times (J), \\
\end{split}
\end{equation}
when $n_0>0$ and $n_1>0$ hold, where $n_0=\#\{k:e_{l}^{(k)}=0\}$, 
 $n_1=\#\{k:e_{l}^{(k)}=1\}$, $k_1=\sum_{d=1}^{K-1}c_{l}^{(d)}$ and $k_2=K-1-k_1$. 
When $n_0=0$ and $n_1>0$, we have 
\begin{equation} \label{posteq2}
\pi(\be_l|\bpsi_l)\propto (H) \times (J).
\end{equation}
When $n_0>0$ and $n_1=0$,  we have 
\begin{equation} \label{posteq3}
\pi(\be_l|\bpsi_l)\propto (H) \times (I).
\end{equation}

Given $K$ distinct conditions, the total number of possible configurations 
 of $\be_l$ is $2^K$. For each possible configuration of $\be_l$, we can 
 calculate its posterior probability and integrated $\psi$-scores. 
 We denote the corresponding posterior probability by $\pi_{ld}$, 
 and denote the corresponding integrated $\psi$-scores by 
 $\bar{\bpsi}_{ld}=(\bar{\psi}_{ld}^{(1)},...,\bar{\psi}_{ld}^{(K)})$ for $d=1,2,\ldots, 2^K$. 
 Here, according to Stouffer's meta-analysis method (Stouffer et al., 1949; Mosteller and Bush, 1954), we define 
 \begin{equation} \label{inteq}
\bar{\psi}_{ld}^{(k)}=\begin{cases}
\sum_{\{i:e_{ld}^{(i)}=0\}} w_i \psi_l^{(i)}/\sqrt{\sum_{\{i: e_{ld}^{(i)}=0\}} w_i^2},  & \mbox{if $e_{ld}^{(k)}=0$}, \\
\sum_{\{i:e_{ld}^{(i)}=1\}}w_i \psi_l^{(i)}/\sqrt{\sum_{\{i: e_{ld}^{(i)}=1\}} w_i^2},  & \mbox{if $e_{ld}^{(k)}=1$}, \\
\end{cases}
\end{equation} 
for $k=1,\ldots, K$, 
where the weight $w_i$ might account for the size or quality of the samples collected under each condition. 
In this paper, we set $w_i=1$ for all $i=1,\ldots, K$.  
Then the Bayesian integrated $\psi$-scores are given by  
\begin{equation} \label{aveeq}
\hat{\psi}_l^{(k)}=\sum_{d=1}^{2^K} \pi_{ld} \bar{\psi}_{ld}^{(k)}, \quad l=1,2,\ldots,N; \ k=1,2,\ldots,K, 
\end{equation}
which has integrated information across all conditions. 
When $K$ is reasonably large, the posterior probabilities 
 $\pi_{ld}$'s can be estimated with a short MCMC run. Since the MCMC can be run in parallel for different $l$'s, 
 the computation is not a big burden in this case.   

 Finally, we note that Stouffer's integrated score (\ref{inteq}) can be viewed as a boosted version 
 of the posterior mean of  $(\mu_{l0},\mu_{l1})$, which amplifies 
 the posterior mean by a factor between 1 and $\sqrt{K}$. 
 As indicated by our proofs [see inequality (\ref{Jan25eq21.2}) 
 in the proof of Lemma \ref{lem4}], 
 such amplification helps to improve the power of the proposed method by 
 reducing the false negative error. 


\subsubsection{Spatial Prior}

 To enhance our prior knowledge that there exits a common structure 
 for all the networks from which they evolve independently, we let $\be_l$'s 
 be subject to the following prior distribution  
 \begin{equation}\label{equ14}
p(\be_l|q)=q^{\sum_{i=1}^{K}c_{l}^{*(i)}}(1-q)^{\sum_{i=1}^{K}(1-c_{l}^{*(i)})}, 
\end{equation}
where $c_{l}^{*(i)}=|e_l^{(i)}-e_l^{mod}|$ indicates the status change of 
 the edge $l$ at condition $i$ from $e_l^{mod}$, and $e_l^{mod}$ is the mode of $\be_l$ and  
 represents the common status of the edge $l$ across all networks. 
 With this prior distribution, the posterior distribution $\pi(\be_l|\bpsi_l)$ can also 
 be expressed in the form of (\ref{posteq1}) but with 
 $k_1=\sum_{d=1}^{K}c_{l}^{*(d)}$ and $k_2=K-k_1$.

\subsection{Joint Edge Detection}

To jointly estimate the structure of multiple GGMs based on the Bayesian integrated $\psi$-scores (\ref{aveeq}), 
a multiple hypothesis test can be applied. The multiple hypothesis test 
classifies the  integrated $\psi$-scores into two classes, one class for the presence of edges and 
 the other class for the absence of edges. In this paper, we adopt the empirical Bayesian method
 developed by Liang and Zhang (2008) for the multiple hypothesis test. 
 A significant advantage of this method is that it allows for 
 the dependence between test statistics. Other multiple hypothesis tests which accounts for the dependence 
 between test statistics, e.g., Benjamini and Yekutieli (2001), can also be applied here.

\subsection{Parameter Setting} 

 FBIA contains two free parameters, i.e., $\alpha_1$ and $\alpha_2$, which
 refer to the significance levels of the
 multiple hypothesis tests conducted in correlation screening and joint edge detection, respectively.
 Following the suggestion of Liang et al. (2015), we set $\alpha_1=0.2$ and
 $\alpha_2=0.05$ as the default values. Otherwise, their values will be stated in the context.
 In general, a high significance level of correlation screening will lead to
 a slightly large conditioning set $S_{ij}$, which reduces the risk of missing some important variables in
 the conditioning set. Including a few false variables in the conditioning set will not hurt much the
 accuracy of the $\psi$-partial correlation coefficient. However, the setting of $\alpha_2$ is quite free,
 which determines the sparsity of the resulting graphs.  A smaller value of $\alpha_2$ might be used
 if sparse graphs are preferred.

 In addition to the two free parameters, FBIA contains four prior-hyperparameters, i.e.,
 $a_1$, $b_1$, $a_2$ and $b_2$. Since the probability $q$ usually takes a small value,
 we set $(a_1,b_1)=(1,10)$ for its prior distribution Beta($a_1$,$b_1$).
 Since the variance of the $\psi$-scores is approximately equal to 1 under the
 null hypothesis that the true partial correlation coefficient is equal to 0, we
 set $(a_2,b_2)=(1,1)$ for its prior distribution IG($a_2$, $b_2$).
 The prior hyperparameter settings have been used in all examples of this paper.

\subsection{Consistency} 
  
 Under the faithfulness assumption and other regularity conditions for the joint Gaussian distribution,  
 e.g., the dimension $p=O(\exp(n^{\delta}))$ is allowed to grow exponentially with the sample size $n$ 
  for some constant $0 \leq \delta <1$ and the largest eigenvalue of the covariance matrix 
 can grow with $n$ at 
 a restricted rate, Liang et al. (2015) showed that the multiple hypothesis test based on the  
 $\psi$-scores produces a consistent estimate for the GGM under single condition.  
 Essentially, Liang et al. (2015) showed that the $\psi$-partial correlation coefficients  
 are separable in probability for the linked and non-linked pairs of nodes. 
  
 To accommodate the change from single condition to multiple conditions, we modified 
 the assumptions of Liang et al. (2015) and added an assumption for $K$.
 Under the new set of assumptions, we proved that the FBIA method is consistent; that is,   

 \begin{theorem} \label{conthem}
  Assume $(A_1)$--$(A_6)$ (see Appendix B) hold.  Then
  \[
  P[\hat{\bE}_{\zeta_n}^{(k)} =\tbE_n^{(k)}, k=1,2,\ldots,K]  \geq 1-o(1), \quad 
   \mbox{as $n \to \infty$}. 
  \]
  where $\tbE_n^{(k)}=\{ (i,j): \rho_{ij|\bV\setminus \{i,j\}^{(k)} } \ne 0, \ i, j=1,\ldots,p\}$ 
  denotes the true network under condition $k$, $\hat{\bE}_{\zeta_n}^{(k)}$ denotes the 
  FBIA estimator of $\tbE_n^{(k)}$, and $\zeta_n$ denotes a threshold value of  
  Bayesian integrated $\psi$-scores based on which the edges are determined for all $K$ 
  graphs. 
 \end{theorem}

 The proof of the theorem is presented in Appendix B. 
 Theorem \ref{conthem} implies that for all graphs there exists a common threshold  
 with respect to which the Bayesian integrated $\psi$-scores are separable in probability for 
 the linked and non-linked pairs of nodes. 
 Here we would like to highlight three points. 
 First, as indicated by our proofs [see the inequality (\ref{Jan25eq21.2}) 
 in the proof of Lemma \ref{lem4}], 
 the data integration step can indeed improve the power of proposed method.  
 Second, following from the inequalities (\ref{Jan25eq21.2}) and (\ref{Jan25eq4.2}) 
 and the condition $(A_5)$, we can conclude the sign consistency of the estimator 
 $\hat{\bE}_{\zeta_n}^{(k)}$; i.e., for any edge of the graph, the sign of the 
 Bayesian integrated $\psi$-score has the same sign as the true partial 
 correlation coefficient when the sample size $n$ becomes large. 
 Third, the assumption imposed on $K$, i.e., $K=O(n^{\delta+2d+\epsilon-1})$, is rather weak, 
 where $\delta$, $d$ and $\epsilon$ are all some positive constants as defined 
 in other assumptions and $\delta+2d <1$ (see Appendix B). 
 For example, we can choose $\epsilon=1-\delta-2d$ and thus $K=O(1)$. This is consistent  
 with our numerical results; the method can perform very well even with a small value of $K$.

\section{Simulation Studies}

\subsection{Scenario with Temporal Priors}

 To illustrate the performance of the proposed method under the scenario with temporal priors, 
 we consider three types of 
 network structures, namely, autoregressive (AR), scale-free and hub, which are all allowed to change slightly 
 with the evolvement of conditions. For all types of structures, we fix $K=4$ and $p=200$, and 
 varied the sample size $n=100$ and 500. We let $\Omega^{(k)}$ denote the precision matrix at 
 condition $k$ for $k=1, \ldots, K$. At each condition $k$, we generated 10 independent datasets of size $n$ 
 by drawing from the multivariate Gaussian distribution $N\left(0, \left(\Omega^{(k)}\right)^{-1}\right)$. 
 
 For the autoregressive network structure, we let the precision matrix at 
 condition 1 be given by 
\begin{equation}\label{plugin2}
  	\Omega^{(1)}_{i,j}=\left\{\begin{array}{ll}
 				  	0.5,&\textrm{if $\left| j-i \right|=1, i=2,...,(p-1),$}\\
  					0.25,&\textrm{if $\left| j-i \right|=2, i=3,...,(p-2),$}\\
					1,&\textrm{if $i=j, i=1,...,p,$}\\
					0,&\textrm{otherwise,}
  				\end{array}\right. 
\end{equation}
 which represents an AR(2) graphical model.  
 To construct $\Omega^{(2)}$, we employed the following random edge deleting-adding procedure: 
 we first randomly removed 5\% edges in $\Omega^{(1)}$ 
 by setting the corresponding non-zero elements to 0, and then added the same number of edges at random 
 by replacing zeros in $\Omega^{(1)}$ with the values drawn from the uniform distribution 
 defined on $[-0.1,-0.3]\cup [0.1,0.3]$; to ensure $\Omega^{(2)}$ to be positive definite, we set 
 the diagonal elements of $\Omega^{(2)}$ to be the smallest absolute eigenvalue of  
 $\tilde{\Omega}^{(2)}$ plus a small positive number, where $\tilde{\Omega}^{(2)}$ is 
 obtained from $\Omega^{(2)}$ by setting the diagonal elements to zero.  
 In the same procedure, we generated $\Omega^{(3)}$ conditioned on $\Omega^{(2)}$
 and then generated $\Omega^{(4)}$ conditioned on $\Omega^{(3)}$.
 We note that similar procedures have been used in Peterson, Stingo and Vannucci (2015) 
 and Lin et al. (2017) to generate multiple precision matrices. 
 For the scale-free and hub structures, we first generated the precision matrix $\Omega^{(1)}$ using  
 the R package ``{\it huge}'', then applied the random edge deleting-adding procedure to 
 generate $\Omega^{(k)}$'s for $k=2,3,4$ in a sequential manner.

 The proposed FBIA method was first applied to this example. 
 To access the performance of the method, we plot the precision-recall curves in Figure \ref{fig_1}. 
 The precision and recall are defined by 
\begin{equation*}
 \mbox{precision}=\frac{TP}{TP+FP},  \qquad \mbox{recall}=\frac{TP}{TP+FN},
\end{equation*}
where $TP$, $FP$ and $FN$ denote true positives, false positives and false negatives, respectively, 
 as defined in Table \ref{Binarytab}.  
 To draw the precision-recall curves shown in Figure \ref{fig_1}, 
 we fix the significance level of correlation screening 
 to $\alpha_1=0.2$ and varied the value of $\alpha_2$, the significance level of joint edge detection. 
 Note that the precision and recall values shown in Figure \ref{fig_1}  
 were calculated by cumulating the TP, FP, FN and TN values across all $K$ conditions.
 In this paper, we employ the precision-recall curve instead of the ROC curve 
 as the classification problem involved in recovering the network structure is severely imbalanced, 
 which contains a large number of negative cases due to the network sparsity. 
 As pointed out by Saito and Rehmsmeier (2015) and Davis and Goadrich (2006), the precision-Recall curve can be more
 informative than the ROC curve in the imbalanced classification scenario.

\begin{table}[htbp]
\begin{center}
\caption{Outcomes of binary decision.} 
\label{Binarytab}
\vspace{2mm}
\begin{tabular}{c|cc} \hline
&   True & False\\\hline
Predicted Positive & True Positive(TP)& False Positive(FP) \\\hline
Predicted Negative & False Negative(FN) & True Negative(TN) \\\hline
\end{tabular}
\end{center}
\end{table}

For comparison, we also applied the fused graphical Lasso(FGL) 
 and group graphical Lasso(GGL) to this example, 
 which are available in the R package {\it JGL} (Danaher, 2012). 
 The FGL employed the fused Lasso penalty  
 \begin{equation} \label{peneq1}
 P(\{\Omega^{(1)}, \ldots, \Omega^{(K)}\})= \lambda_1 \sum_{k=1}^K \sum_{i\ne j} |\omega_{ij}^{(k)}| + 
    \lambda_2 \sum_{k<k'} \sum_{i,j} |\omega_{ij}^{(k)}-\omega_{ij}^{(k')}|, 
 \end{equation}
 where $\lambda_1$ and $\lambda_2$ are regularization parameters, and $\omega_{ij}^{(k)}$ denotes the 
 $(i,j)$-th element of the precision matrix $\Omega^{(k)}$.  
 The GGL employed the following penalty,   
 \begin{equation} \label{peneq2}
 P(\{\Omega^{(1)}, \ldots, \Omega^{(K)}\})= \lambda_1 \sum_{k=1}^K \sum_{i\ne j} |\omega_{ij}^{(k)}| + 
    \lambda_2 \sum_{i\ne j} \left( \sum_{k=1}^K {\theta_{ij}^{(k)}}^2 \right)^{1/2} , 
 \end{equation}
 which is a combination of Lasso and group Lasso penalties.
 For both penalties (\ref{peneq1}) and (\ref{peneq2}), 
 the first term enhances the sparsity of each precision matrix, and the second 
 term enhances a similar pattern across all precision matrices. 
 To determine the values of $\lambda_1$ and $\lambda_2$, we follow the procedure 
 recommended by Danaher, Wang and Witten (2014) to search 
 over a grid of possible values for a combination that 
 minimizes the Akaike information criterion (AIC).    
 To draw the precision-recall curve shown in Figure \ref{fig_1}, we fix the value of 
 $\lambda_2$ to its optimal value at which the minimum AIC is attained, and varied the 
 value of $\lambda_1$; that is, we fix the level of similarity and varied the level of 
 sparsity of the graphs.  
 For a thorough comparison, we also applied the original $\psi$-learning algorithm 
 to this example, for which the models under each condition were estimated separately. 
 As indicated by Figure \ref{fig_1}, the FBIA method significantly outperforms 
 the existing methods, especially when the sample size is small.   
 When the sample size is large,  
 FBIA, FGL and GGL tend to perform similarly for the scale-free and hub networks; however,  
 FBIA still outperforms FGL and GGL for the AR(2) network. It is not surprising that FBIA always outperforms 
 the separated $\psi$-learning algorithm, which implies the importance of data integration 
 for such high-dimensional problems.

 \begin{figure}[htbp]
\centering
\subfigure[AR(2) with $(n,p)=(100,200)$.]{
\label{fig1} 
\includegraphics[width=2.2in,angle=270]{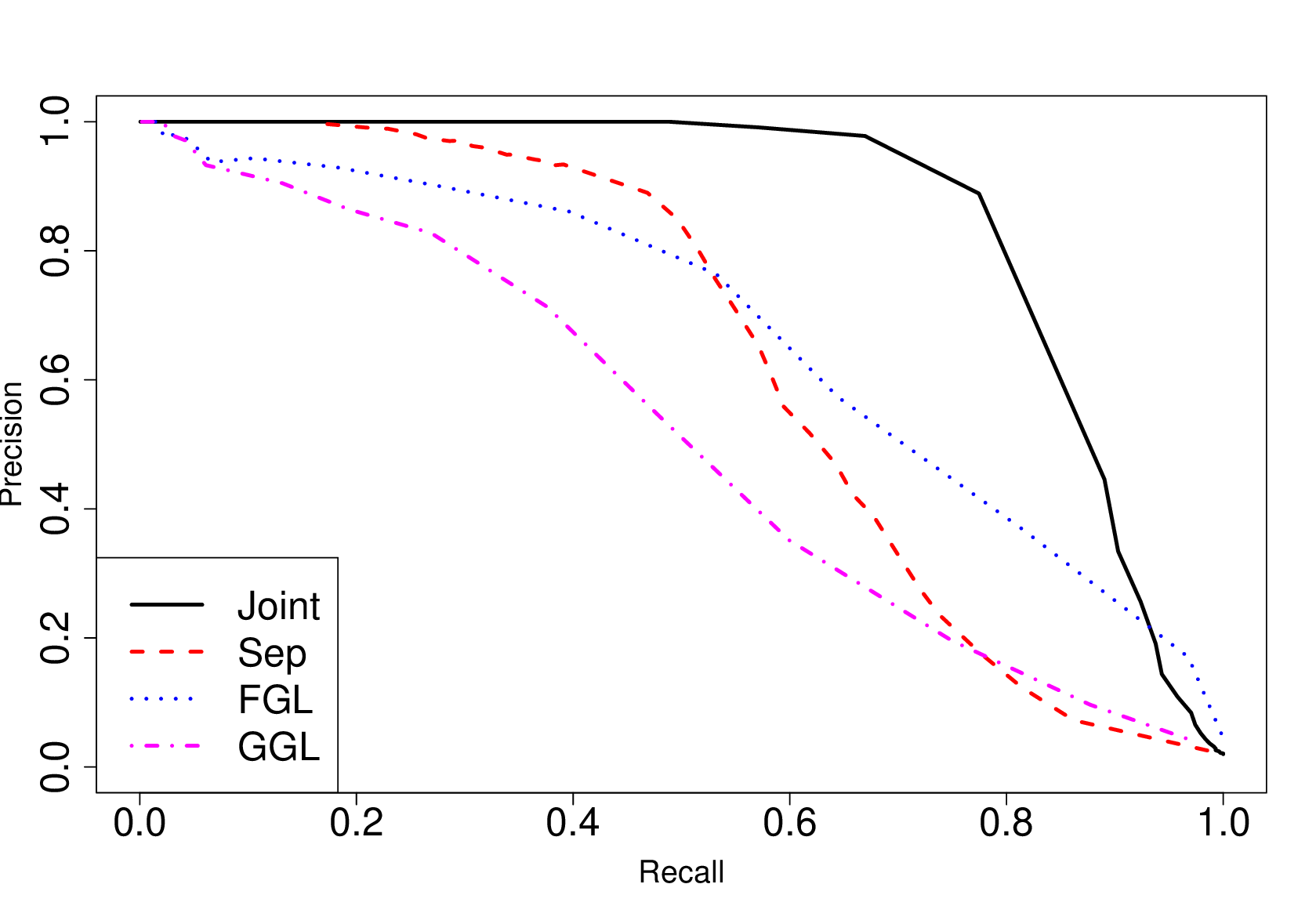}}
\hspace{0.5in}
\subfigure[AR(2) with $(n,p)=(500,200)$.]{
\label{fig2}
\includegraphics[width=2.2in,angle=270]{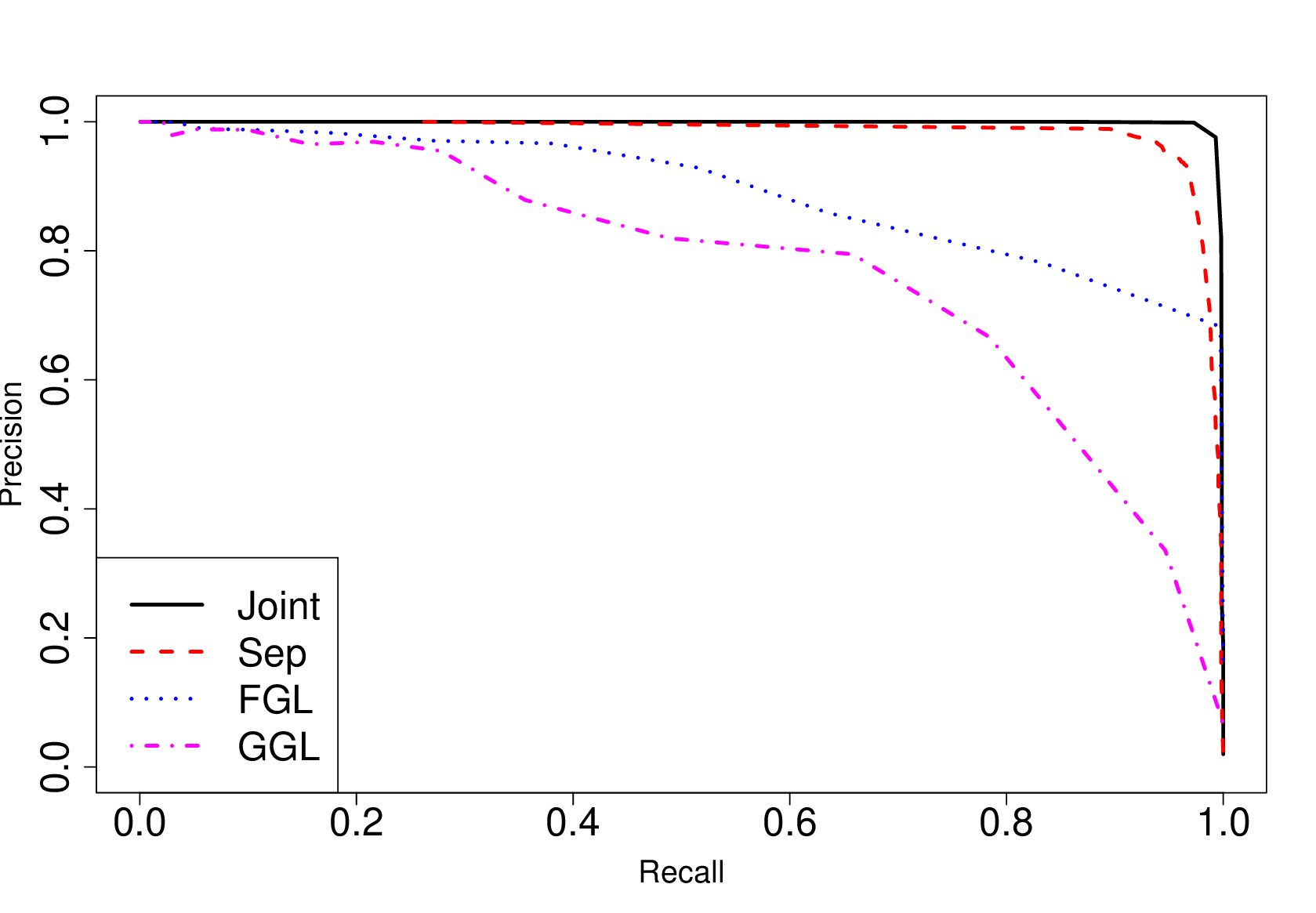}}
\\
\noindent
\subfigure[Scale-free with $(n,p)=(100,200)$.]{
\label{fig3} 
\includegraphics[width=2.2in,angle=270]{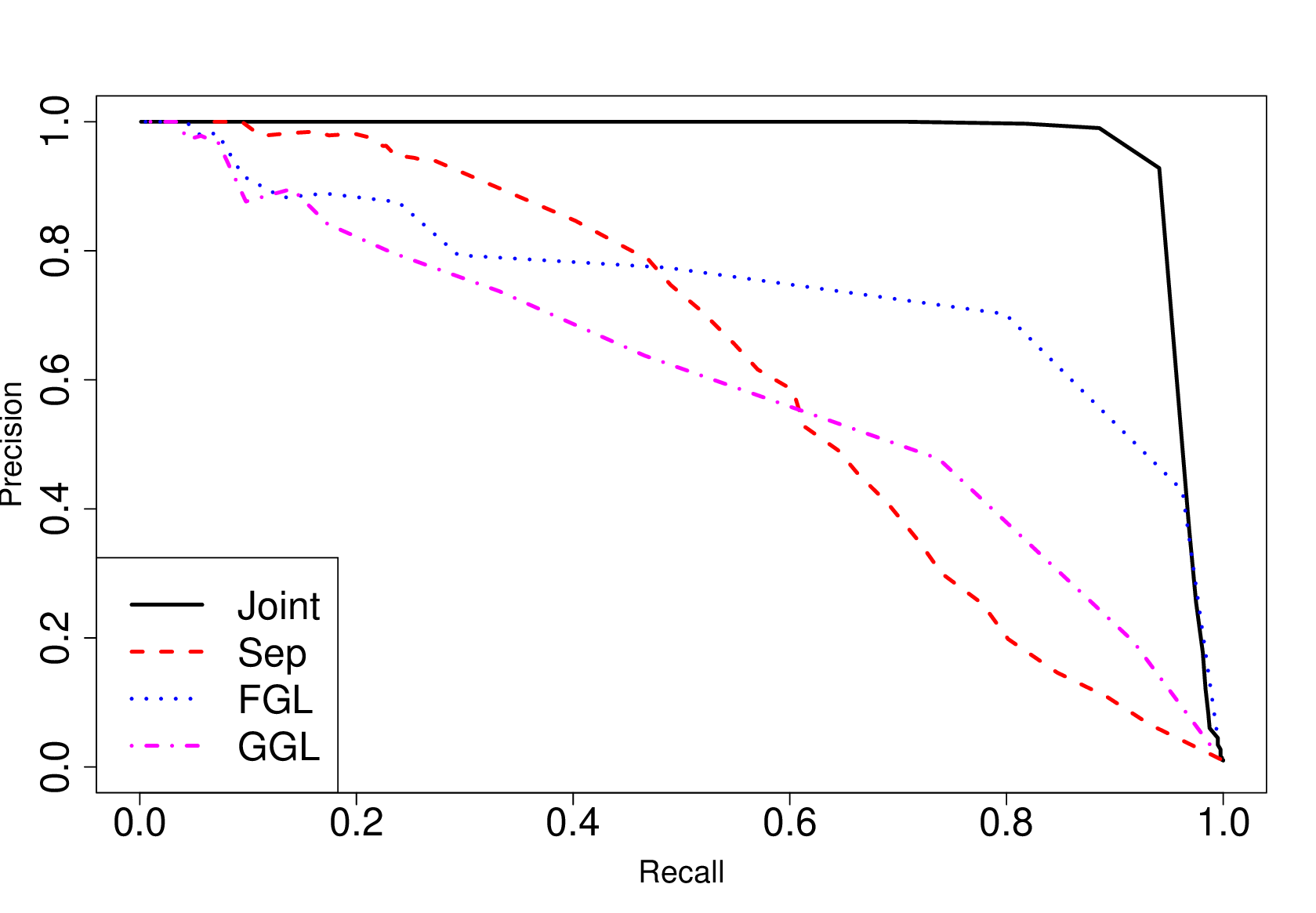}}
\hspace{0.5in}
\subfigure[Scale-free with $(n,p)=(500,200)$.]{
\label{fig4} 
\includegraphics[width=2.2in,angle=270]{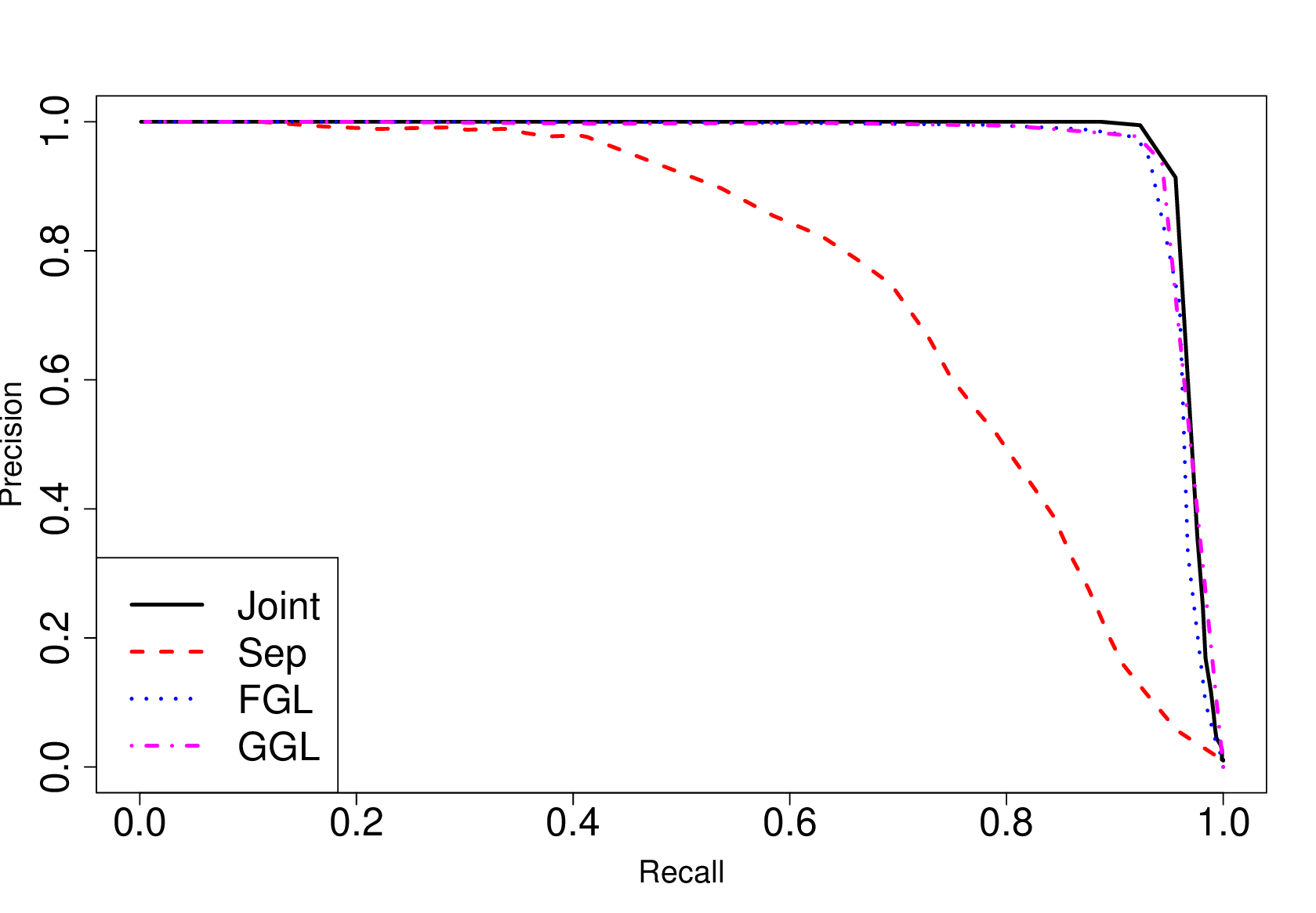}}
\noindent
\subfigure[Hub with $(n,p)=(100,200)$.]{
\label{fig5} 
\includegraphics[width=2.2in,angle=270]{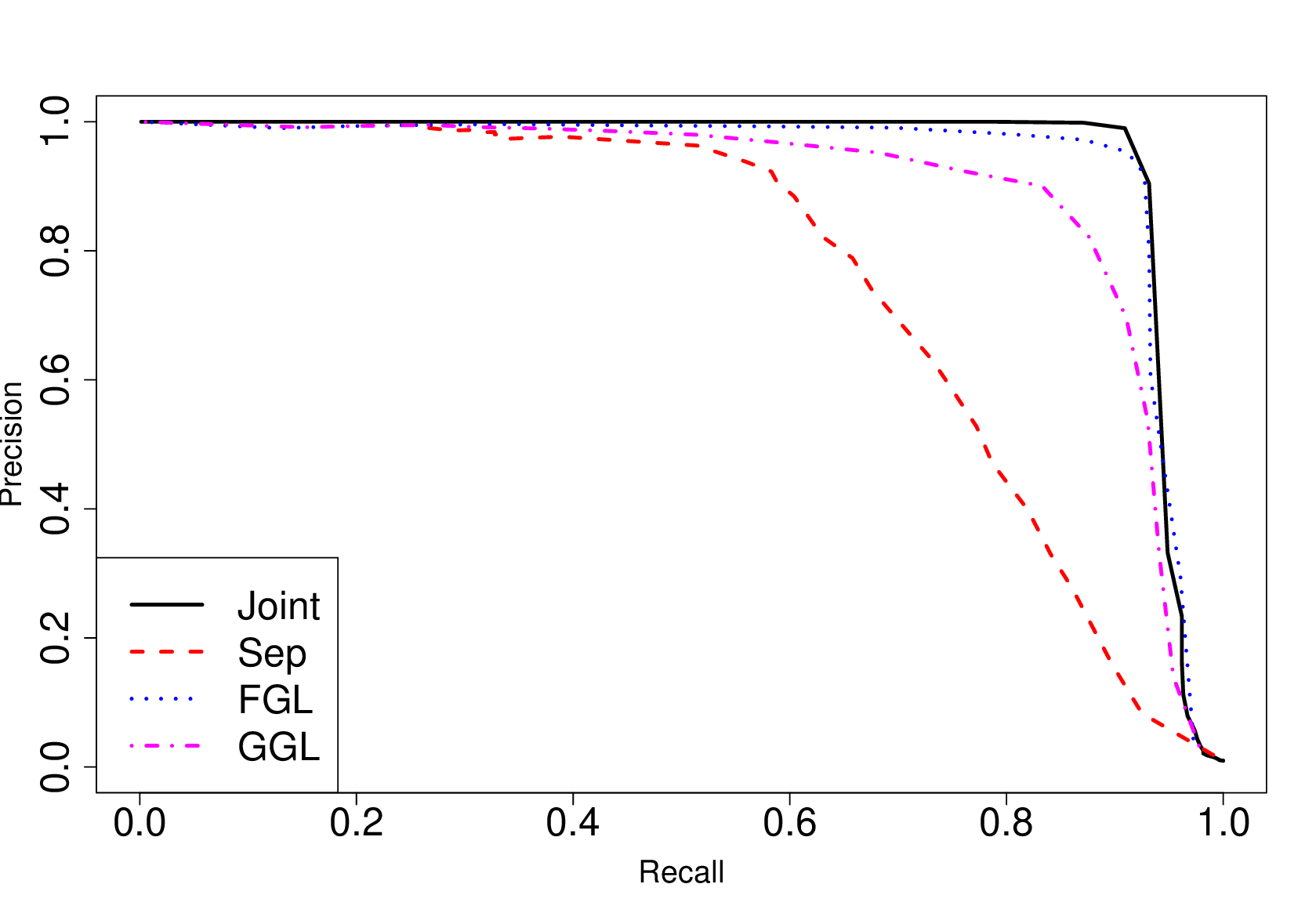}}
\hspace{0.5in}
\subfigure[Hub with $(n,p)=(500,200)$.]{
\label{fig6} 
\includegraphics[width=2.2in,angle=270]{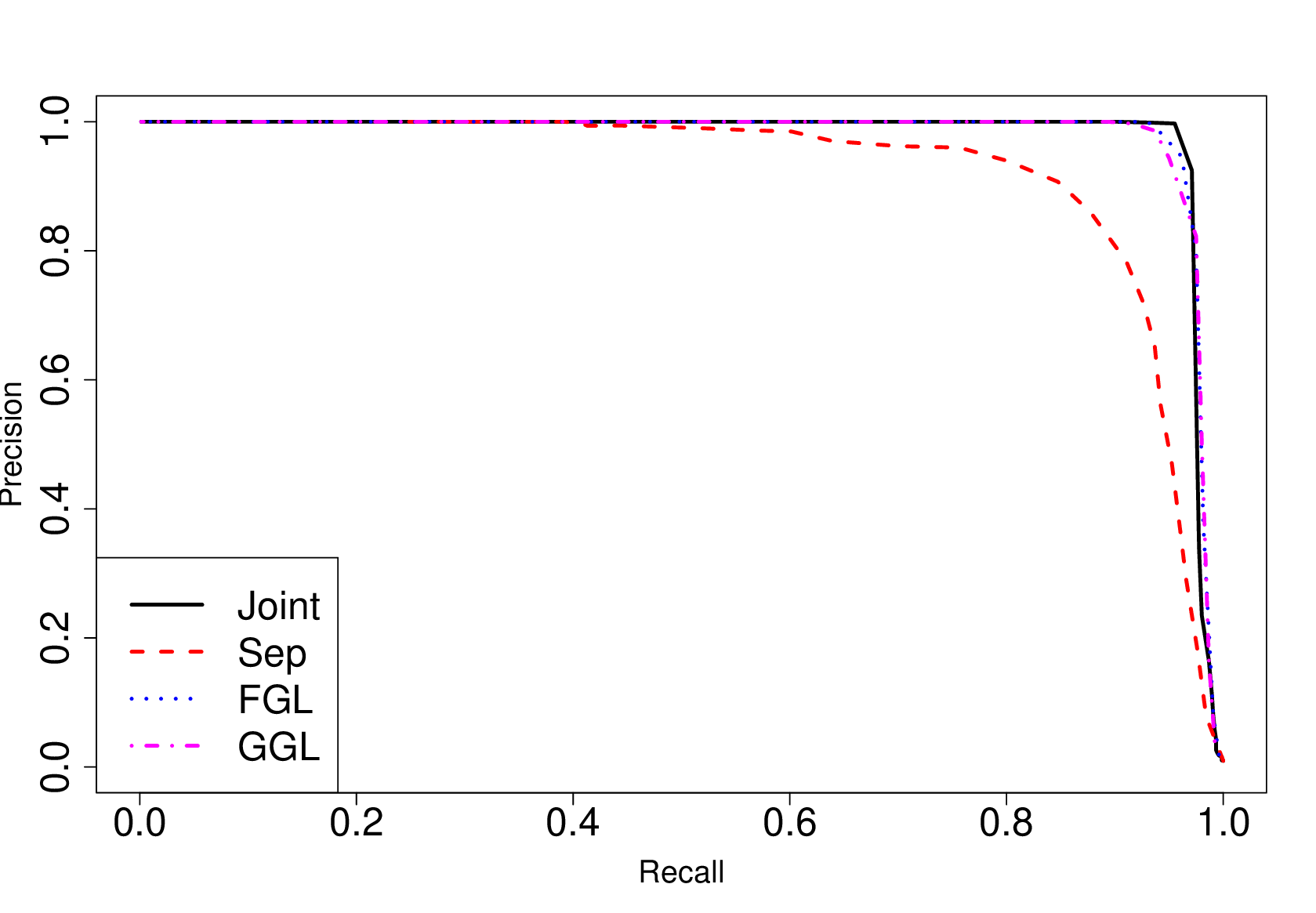}}
\\
\caption{Comparison of the FBIA method (labeled as ''Joint'' in the plots) with 
 FGL, GGL and separated $\psi$-learning (labeled as ''Sep'' in the plots) for the 
 simulated temporal data with the underlying structures: AR(2) (top row), 
 scale free (middle row), and hub (bottom row).  }
\label{fig_1} 
\end{figure}

Table \ref{simTab1} summarizes the performance of the FBIA, FGL, GGL and separated $\psi$-learning methods 
on 10 datasets by reporting the averaged areas under the precision-recall curves.  
 The comparison indicates that when $n=100$, 
 FBIA significantly outperforms all other three methods; and when $n=500$,  
 FBIA still significantly outperforms all other three methods for the AR(2) network, but 
 tends to have the same performance as FGL and GGL for the scale-free and hub networks.

\begin{table}[htbp]
\begin{center}
\caption{Averaged areas under the Precision-Recall curves produced by 
 FBIA, FGL, GGL, and separated $\psi$-learning for 10 datasets simulated under the scenario of 
 temporal priors, where the 
 number in the parentheses represents the standard deviation of the averaged area. }
\label{simTab1}
\vspace{2mm}
\begin{tabular}{cccccc} \hline
 n& Structure & FGL & GGL & $\psi$-Learning & FBIA\\ \hline
  \multirow{3}{2cm}{\centering $100$} & AR(2) & 0.655(0.008) & 0.494(0.009) & 0.628(0.005) & {\bf 0.863}(0.004)  \\
  &scale-free &  0.750(0.009) & 0.609(0.008) & 0.664(0.005) & {\bf 0.965}(0.001)  \\
   &hub &  0.937(0.002) & 0.899(0.002) & 0.750(0.005) & {\bf 0.950}(0.002)   \\ \hline
  \multirow{3}{2cm}{\centering $500$} & AR(2)&  0.882(0.003)& 0.770(0.006) & 0.985(0.010) & {\bf 0.999}(0.004) \\
 & scale-free & 0.950(0.002)& 0.949(0.002) & 0.728(0.003) & {\bf 0.970}(0.005) \\
   &hub &  0.972(0.001)& 0.970(0.001) & 0.902(0.001) & {\bf 0.977}(0.001)   \\\hline
\end{tabular}
\end{center}
\end{table}

Table \ref{Time1} reports the CPU times cost by FGL, GGL, separated $\psi$-learning and FBIA for 
one dataset of AR(2) structure, where the CPU time was measured on a Linux desktop 
with Inter Core i7-4790 CPU\@3.6Ghz. All computations reported in this paper were done on 
 the same computer.  
 The CPU times of these methods for the other two graph structures are about the same.  
 FGL is extremely slow for this example, as it needs to search over a grid of possible 
 values for an optimal setting of $(\lambda_1,\lambda_2)$. The grid we used consists 
 of 100 different pairs of $(\lambda_1,\lambda_2)$.  Moreover, 
 for each pair of $(\lambda_1,\lambda_2)$, it needs to solve a generalized fused Lasso problem 
 for which a closed-form solution does not exist when $K$ is greater than 2.  
 Solving the generalized fused Lasso problem is time consuming and 
 has a computational complexity of $O(p^2 K \log K)$. 
 The GGL is better as for which there exists a closed-form solution to the regularized 
 parameter optimization problem under each setting of $(\lambda_1,\lambda_2)$, although the optimal 
 setting of $(\lambda_1,\lambda_2)$ also needs to be searched over a grid of 100 points.
 The computational complexity of FBIA is of $O(p^2 2^K)$, which can be pretty fast for a small 
 value of $K$. The separated $\psi$-learning is  a little more time consuming 
 than FBIA because it needs to conduct multiple hypothesis tests under each 
 condition. 

\begin{table}[htbp]
\begin{center}
\caption{CPU time cost by FGL, GGL, separated $\psi$-learning, and  FBIA  
 for the datasets generated with the AR(2) structure.} 
\label{Time1}
\vspace{2mm}
\begin{tabular}{ccccc} \hline
 Sample size &FGL&GGL& $\psi$-Learning  & FBIA \\\hline
$n=100$&14.89 hrs & 28.66 mins & 11.48 mins & 8.95 mins \\\hline
$n=500$& 18.46 hrs & 68.89 mins &12.31 mins  &9.77 mins  \\\hline
\end{tabular}
\end{center}
\end{table}

\subsection{Scenario with Spatial Priors}

 As in the scenario with temporal priors, we considered three types of network structures:
 AR(2), scale-free and hub. 
 For each type of structures, we set $K=5$ and $p=200$, and tried 
 two sample sizes $n=100$ and $n=500$. 
 For AR(2), we first generated the precision matrix $\Omega^{(0)}$ according to (\ref{plugin}). 
 Conditioned on $\Omega^{(0)}$, we generated the precision matrices 
 $\Omega^{(k)}$, $k=1,2,\ldots,5$, independently using the random edge 
 deleting-adding procedure as described in the scenario of temporal priors.   
 For the other two types of structures, we generated the precision matrices 
 $\Omega^{(0)}$ using the R package {\it huge}, and then generated 
 $\Omega^{(k)}$, $k=1,2,\ldots,5$ independently using the random edge deleting-adding procedure. 
 Given the precision matrices, we then generated 10 independent datasets of size $n$ 
 by drawing from the multivariate Gaussian 
 distribution $N\left(0, \left(\Omega^{(k)}\right)^{-1} \right)$ for each condition $k$.

 The FBIA, FGL, GGL, separated $\psi$-learning and graphical EM (Xie, Liu and Valdar, 2016) methods 
 were applied to this example.  The graphical EM algorithm was specially designed for 
 jointly estimating multiple dependent Gaussian graphical models under this scenario. 
 It works by decomposing the problem into two graphical layers, namely, the systemic layer 
 and the category-specific layer.  
 The former induces cross-graph dependence and represents the underlying common structure, 
 and the latter represents the graph-specific variation.  
 By treating the systemic layer data as missing, the EM algorithm was applied to
 estimate the underlying precision matrices.
  
 Figure \ref{fig_2} shows the precision-recall curves produced 
 for two datasets by FBIA, FGL, GGL, separated $\psi$-learning and graphical EM. 
 Table \ref{Tab4} summarizes the performance of these methods for 
 all simulated datasets of this example. The comparison indicates that FBIA significantly 
 outperforms all other methods, especially when the sample size is small. 

\begin{figure}
\centering
\subfigure[AR(2) with $(n,p)=(100,200)$.]{
\label{fig7} 
\includegraphics[width=2.2in,angle=270]{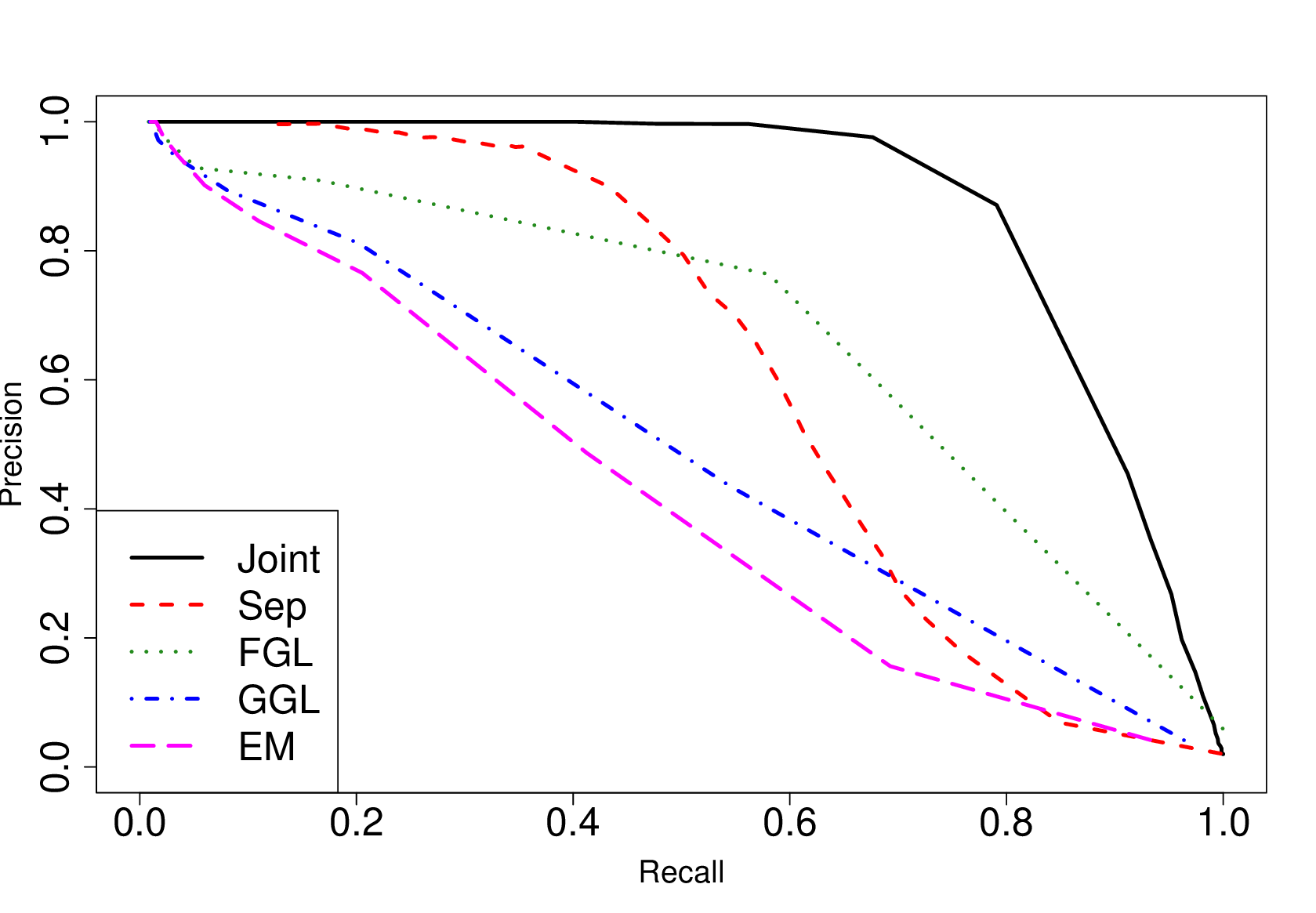}}
\hspace{0.5in}
\subfigure[AR(2) with $(n,p)=(500,200)$.]{
\label{fig8}
\includegraphics[width=2.2in,angle=270]{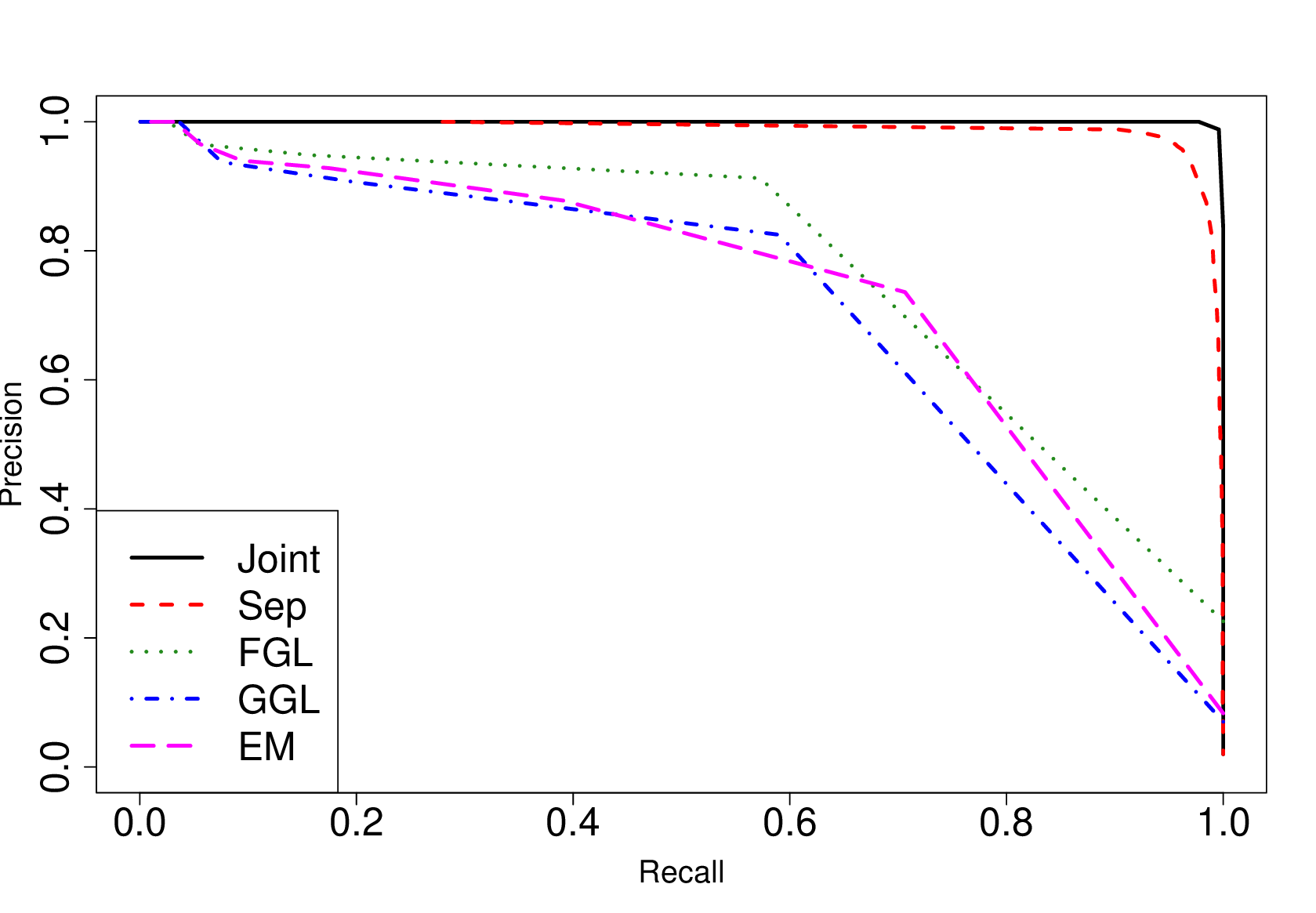}}
\\
\noindent
\subfigure[Scale-free with $(n,p)=(100,200)$.]{
\label{fig9} 
\includegraphics[width=2.2in,angle=270]{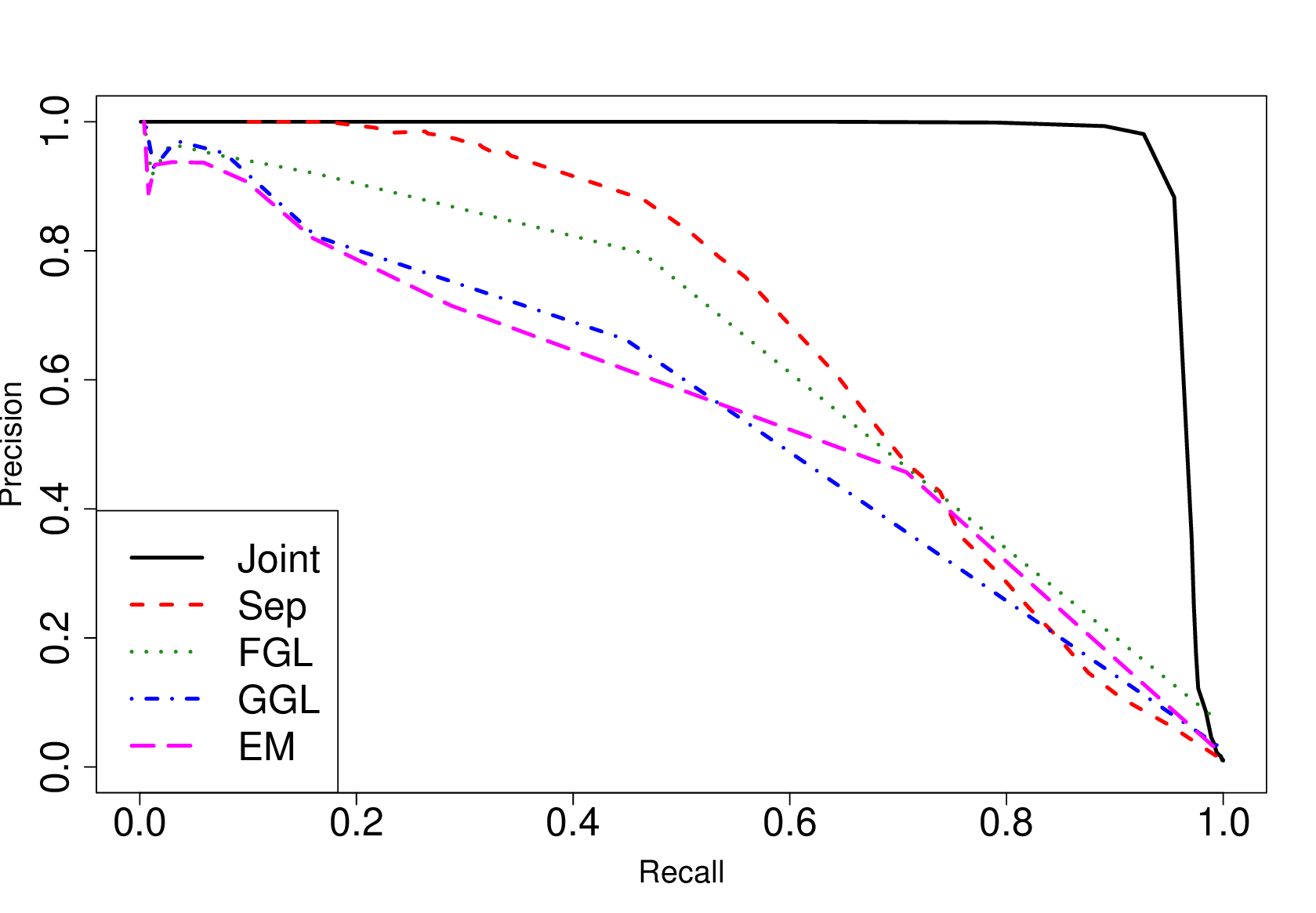}}
\hspace{0.5in}
\subfigure[Scale-free with $(n,p)=(500,200)$.]{
\label{fig10} 
\includegraphics[width=2.2in,angle=270]{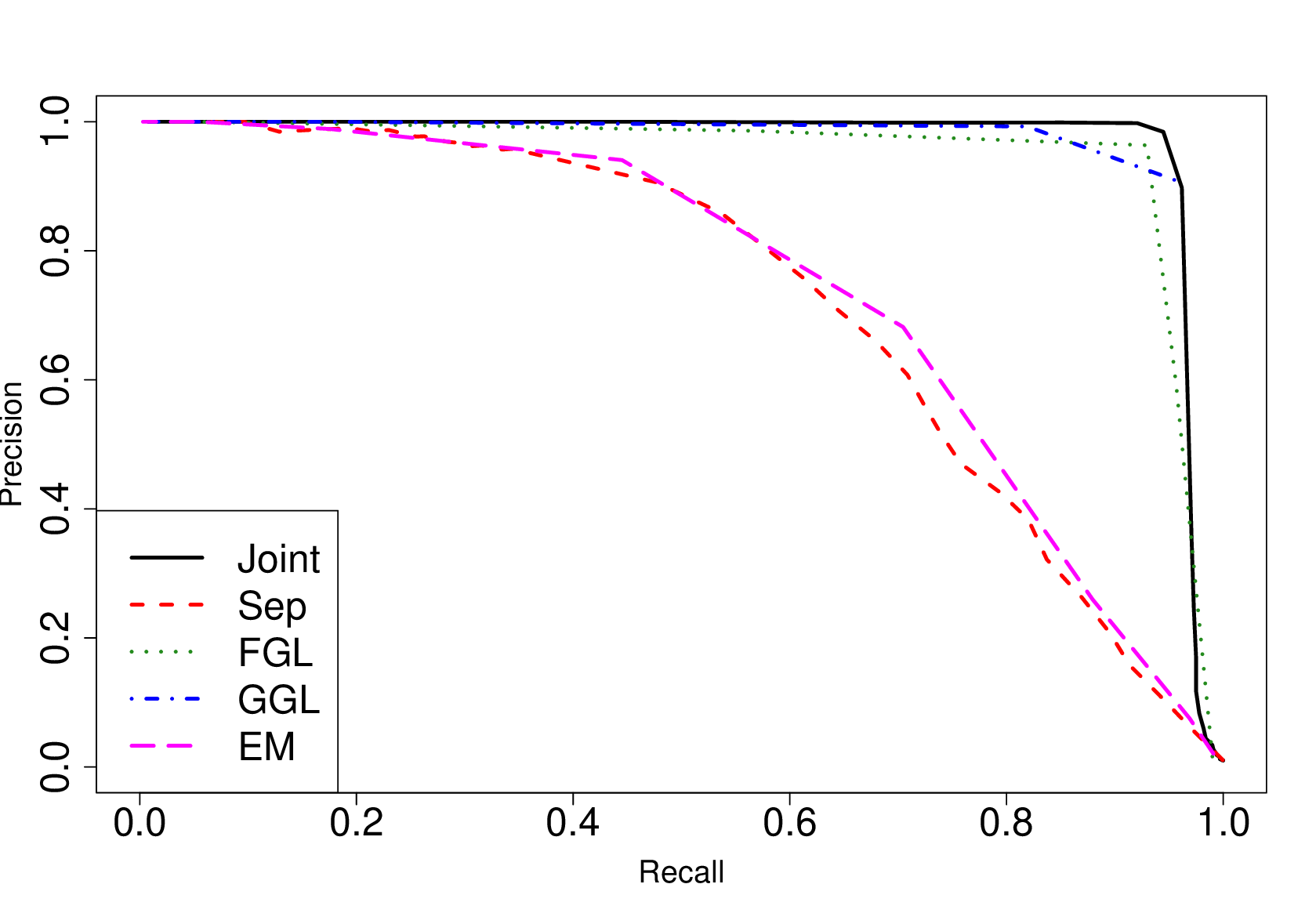}}
\noindent
\subfigure[Hub with $(n,p)=(100,200)$.]{
\label{fig11} 
\includegraphics[width=2.2in,angle=270]{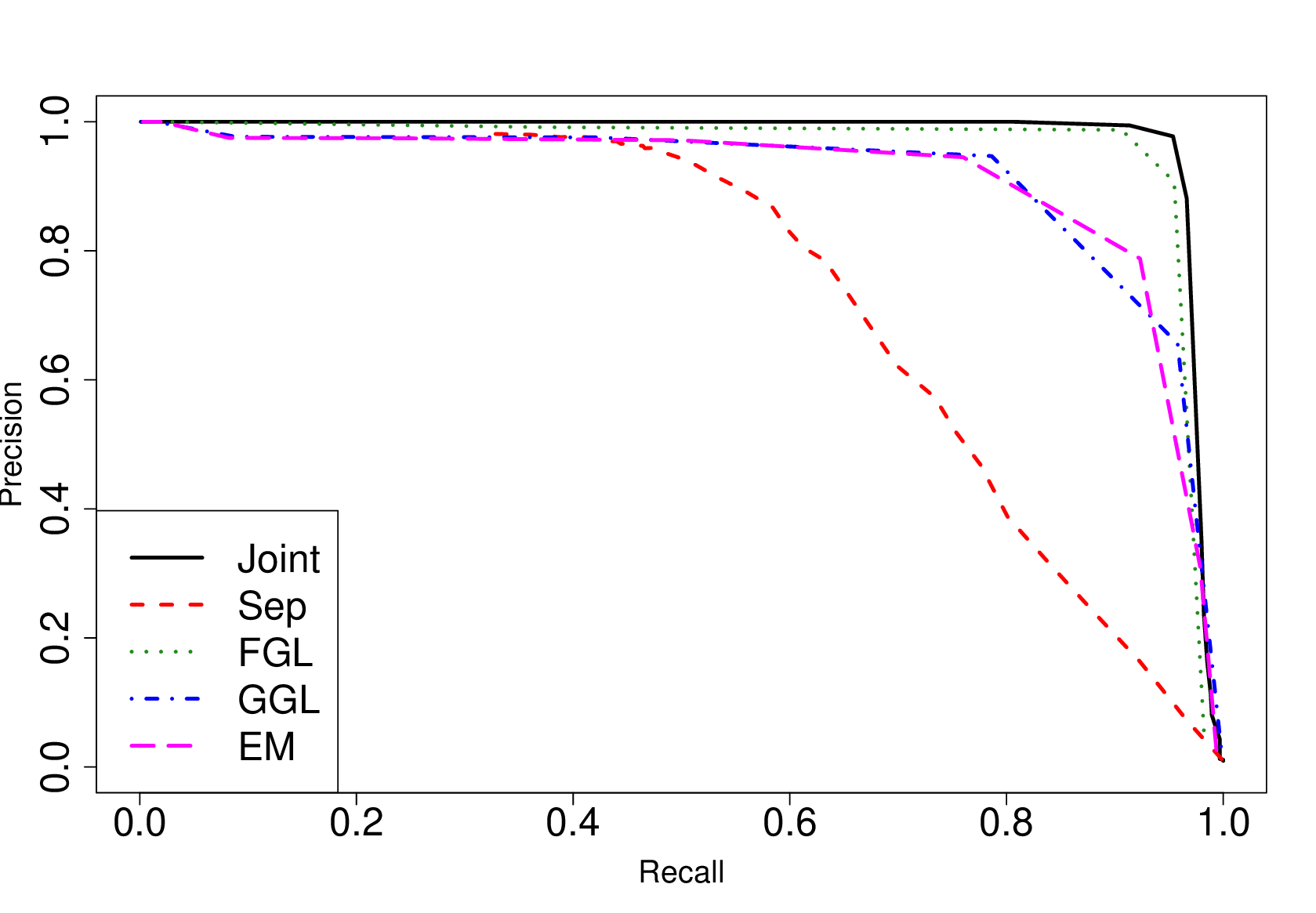}}
\hspace{0.5in}
\subfigure[Hub with $(n,p)=(500,200)$.]{
\label{fig12} 
\includegraphics[width=2.2in,angle=270]{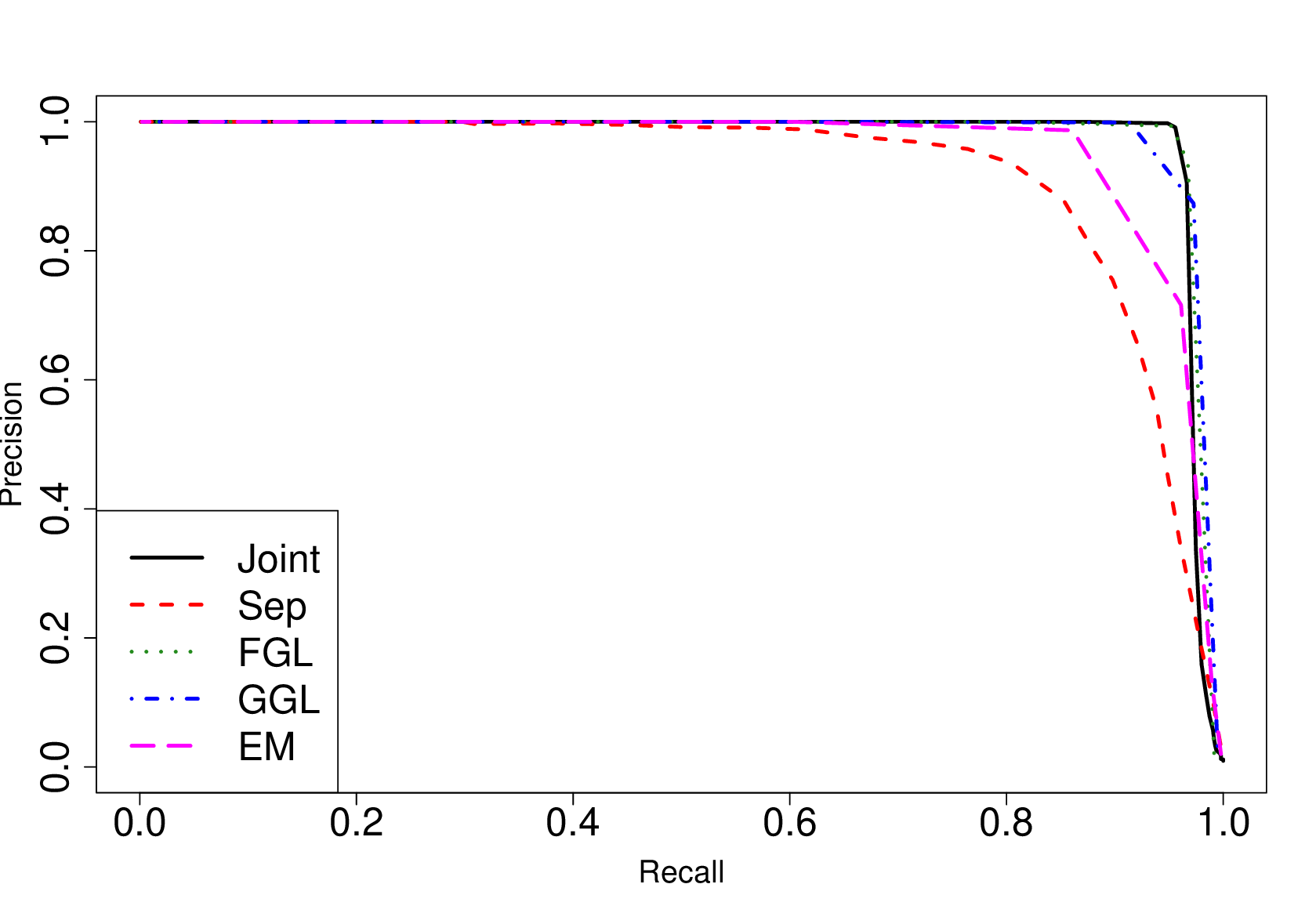}}
\\
\caption{Comparison of the FBIA method (labeled as ''Joint'' in the plots) with
 FGL, GGL and separated $\psi$-learning (labeled as ''Sep'' in the plots) for the
 simulated spatial data with the underlying structures: AR(2) (top row),
 scale free (middle row), and hub (bottom row).}
\label{fig_2} 
\end{figure}

\begin{table}[htbp]
\begin{center}
\caption{Averaged areas under the Precision-Recall curves produced by
 FBIA, FGL, GGL, separated $\psi$-learning, and graphical EM for 10 datasets 
 simulated under the scenario of spatial priors, where the
 number in the parentheses represents the standard deviation of the averaged area. }
\label{Tab4}
\vspace{2mm}
\begin{tabular}{ccccccc} \hline
 n& Structure & FGL & GGL & $\psi$-Learning &  EM &  FBIA \\ \hline
  \multirow{3}{2cm}{\centering $100$} & AR(2) & 0.681(0.003) & 0.508(0.004) & 0.616(0.006) & 0.427(0.005) & {\bf 0.878}(0.004) \\
  &scale-free &  0.631(0.007) & 0.556(0.006) & 0.657(0.006) & 0.555(0.006) &  {\bf 0.961}(0.001) \\
   &hub &  0.949(0.005) & 0.885(0.007) & 0.730(0.008) &  0.876(0.008)&  {\bf 0.971}(0.001)  \\\hline 
  \multirow{3}{2cm}{\centering $500$} & AR(2)&  0.797(0.002)& 0.719(0.002) & 0.989(0.001) &  0.753(0.003)& {\bf 0.999}(0.001)\\
 & scale-free & 0.949(0.001)& 0.963(0.001) & 0.736(0.004)&  0.749(0.004) &  {\bf 0.967}(0.001)\\
   &hub &  0.972(0.001)& 0.969(0.001) & 0.916(0.001) &  0.948(0.002)&  {\bf 0.975}(0.001)\\\hline
\end{tabular}
\end{center}
\end{table}

Table \ref{Time2} reports the CPU times cost by FGL, GGL, separated $\psi$-learning, graphical EM, and FBIA for
one dataset of AR(2) structure. The CPU times for the other two graph structures are about the same.
For FGL, this example is even more time consuming than the previous one, although it was run 
under exactly the same setting for the two examples. One reason is that $K$ has increased from 4 
to 5. For FBIA, the CPU time is not much increased  compared to the previous example. 

\begin{table}[htbp]
\begin{center}
\caption{CPU time cost by FGL, GGL, separated $\psi$-learning, graphical EM and FBIA for
 the datasets generated with AR(2) structure.} 
\label{Time2}
\vspace{2mm}
\begin{tabular}{cccccc} \hline
Sample Size&FGL&GGL&  $\psi$-Learning & EM &FBIA \\\hline
$n=100$&38.43 hrs & 27.21 mins & 14.27 mins & 20.71 mins & 11.47 mins \\\hline
$n=500$& 59.51 hrs & 29.80 mins &14.94 mins  &22.57 mins &13.95 mins \\\hline
\end{tabular}
\end{center}
\end{table}

\section{TEDDY Data Analysis}

 This section applied the FBIA method to the mRNA gene expression data collected 
 in the study of The Environmental Determinants of Diabetes in the Young (TEDDY). 
 In the study, to reduce potential bias and retain study power while reducing the costs by limiting 
 the numbers of samples requiring laboratory analyses, 
 the gene expression data were collected from the nested matched case-control cohort. 
 A subject who developed two primary outcomes, persistent confirmed islet autoimmunity 
 (i.e. the presence of one confirmed autoantibody, GADA65A, IA-2A or IAA, on two or more consecutive samples) 
 and/or T1D, was defined as a case. The controls are randomly selected among cohort members who 
 have not yet developed the disease at the time a case is diagnosed. 
 For each subject, the gene expression data were collected at multiple time points within  
 four years of age. Refer to Lee et al. (2014) for the detailed description for the study.
 Our goal is to integrate all the data to  construct one gene network under each distinct condition. 

 The dataset consists of 21285 genes and 742 samples collected at multiple time points from a 
 total of 313 subjects. Among the 742 samples, half of them are for the case and  half of them are for the control.  
 The dataset also contains some external variables for each patient, which include age (the time of data collected),
 gender, race, race ethnicity, season of birth, number of older siblings, and country.
 To simplify the analysis, we first filtered out some non-differentially expressed 
 genes across the case and control conditions. This was done by conducting a paired $t$-test for each gene 
 at each time point and then applied the multiple hypothesis test method by Liang and Zhang (2008) 
 to identify the set of genes that are significantly differentially expressed under the two conditions at least 
 at one time point.  With this filtering process, 572 genes were selected for further study. 
 Figure \ref{abs0} shows the histogram of the ages of the samples.  Based on this histogram, 
 we selected only the samples fallen into the first 9 groups for the further analysis, where 
 each mode of the histogram is treated as a group. 
 The respective group sizes are  29, 40, 49, 43, 32, 27, 27, 23, and 21, 
 which are the same for both the case and control.  
 Since the samples were grouped in ages, the index $k=1,2,\ldots, 9$ can be understood as 
 the time of experiments. In grouping the samples we have ensured that in each group, 
 each sample corresponds to a different patient and thus the samples within the same group 
 can be treated as mutually independent. Since the sample size of each group is 
 small, we set $\alpha_1=0.05$ and $\alpha_2=0.01$, which are smaller than the default values.  

\begin{figure}[htbp]
\includegraphics[scale=1.0]{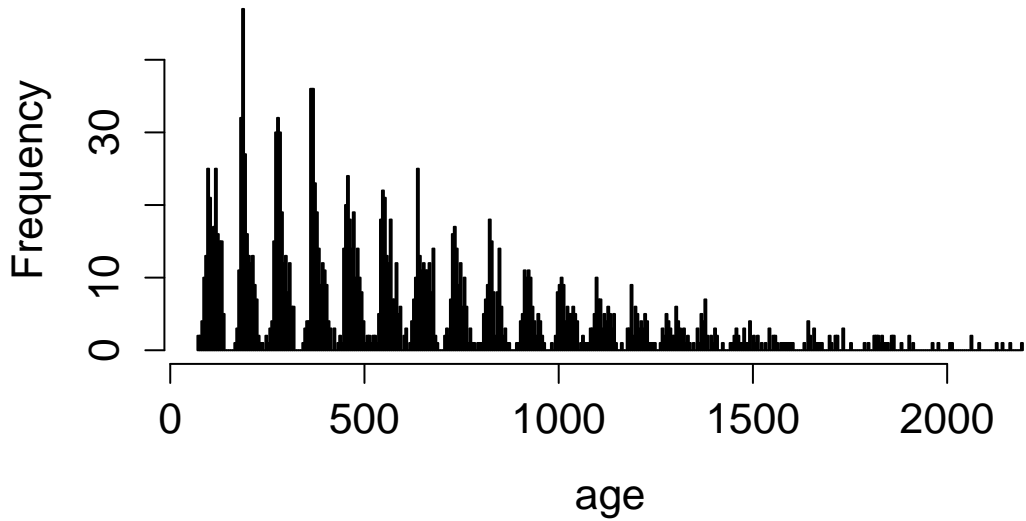}
\caption{Histogram of Ages (in days) for all samples in the case and control groups.}
\label{abs0}
\end{figure}

 To adjust the effect of external variables, we adopted the method proposed by Liang et al. (2015). 
 Let $W_1^{(k)},..., W_q^{(k)}$ denote the external variables observed at condition $k$. 
 To adjust for their effects, 
 we can replace the empirical correlation coefficient used in the $\psi$-score calculation step 
 by the p-value obtained in testing the hypotheses $H_0 :\beta_{q+1} =0 \leftrightarrow H_1 :\beta_{q+1}\neq 0$ 
 for the regression
\begin{equation}
X_{i}^{(k)}=\beta_0+\beta_1W_1^{(k)}+\cdots+ \beta_qW_q^{(k)}+\beta_{q+1}X_{j}^{(k)}+\epsilon,
\end{equation}
where $X_{i}^{(k)}$ denote the expression value of gene $i$ measured at condition $k$, 
 and $\epsilon$ denotes a vector of Gaussian random errors. 
Similarly, we can replace the $\psi$-partial correlation coefficient calculated in the $\psi$-score calculation step
by the p-value obtained in testing the hypotheses $H_0 :\beta_{q+1} =0 \leftrightarrow H_1 :\beta_{q+1}\neq 0$ 
 for the regression
\begin{equation}\label{p-reg}
X_{i}^{(k)}=\beta_0+\beta_1W_1^{(k)}+\cdots+ \beta_qW_q^{(k)}+\beta_{q+1}X_{j}^{(k)}
  +\sum_{s \in S_{ij}^{(k)}}\gamma_s X_{s}^{(k)}+\epsilon,
\end{equation}
where $S_{ij}^{(k)}$ is the separator of $X^{(i)}$ and $X^{(j)}$ under condition $k$. 
With the $p$-values, we can define the adjusted $\psi$-score as $\psi_{l}^{(k)}=\Phi^{-1}(1-p_{l}^{(k)})$, 
where $p_{l}^{(k)}$ is the p-value obtained from equation (\ref{p-reg}) 
for edge $l$ at condition $k$. 

For this dataset, the effect of all available demographical variables, including 
 age (the time of data collection), gender, race, race ethnicity, season of birth, number of older siblings, 
 and country, have been adjusted. 
With the adjusted $\psi$-scores, the FBIA method is ready to be applied 
to construct the gene networks. Given the complexity of the dataset, which 
contains case and control groups and multiple time points for each group, we 
calculated the integrated $\psi$-scores in two steps. First, we integrated the $\psi$-scores  
across 9 time points under the case and control, separately. 
Then, for each time point, we integrated the $\psi$-scores across the case and control 
conditions. In this way, all information of the data collected under the 18 conditions 
were integrated together.  
Figure \ref{dis1} shows a schematic diagram  for this two-step procedure. 
Finally, we applied the multiple hypothesis test to the Bayesian integrated $\psi$-scores  
to determine the structure of the gene networks under the 18 conditions. 
The total CPU time cost by FBIA was 19.2 hours, which is pretty long as 
$K=9$ is large. For a larger value of $K$, we might resort to MCMC for estimating 
the posterior probabilities $\pi(\be_l|\bpsi_l)$'s.

 Figure \ref{fig_case} shows the networks constructed by FBIA for the 
 case samples at 9 time points. The networks have identified quite a few hub genes, 
 which refer to the genes with high connectivity. Table \ref{tab_case} shows the 
 top 5 hub genes identified at each time point for the case samples.  
 The lists of hub genes are pretty stable. For example, 
 RPS26P11 and RPS26 consistently appear as top 2 genes at all time points,  
 the gene ADAM10 appeared at 5 out of 9 time points, and quite a few genes appeared 
 twice or more times, such as PRF1, POGZ, BCL11B, GGNBP2, and TMEM159.  Note that 
 RPS26P11 is a pseudo-gene, which represents a segment of the gene RPS26.

\begin{figure}[]
\centering
\subfigure[time 1]{
\label{fig13} 
\includegraphics[width=2in]{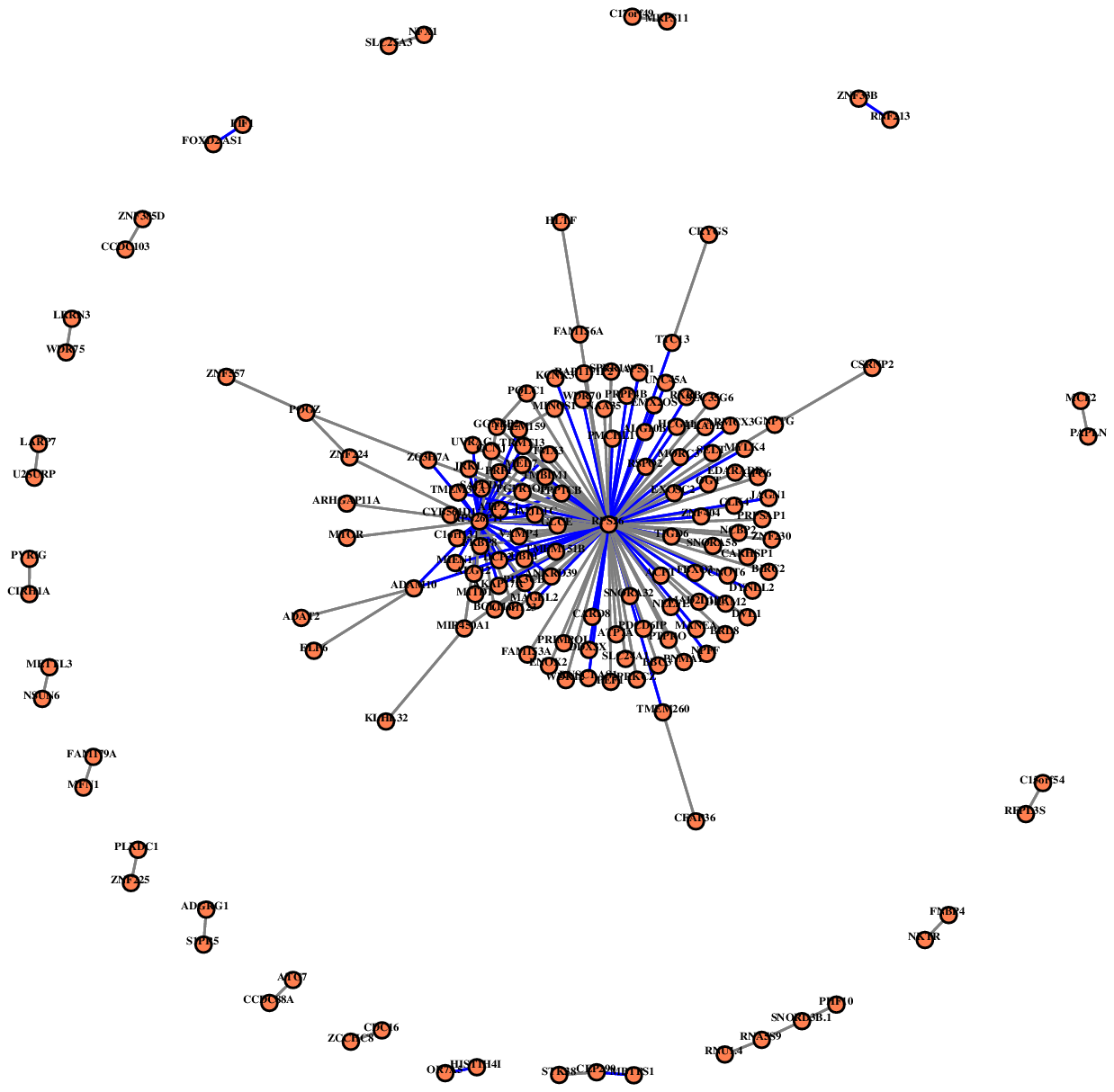}}
\subfigure[time 2]{
\label{fig14}
\includegraphics[width=2in]{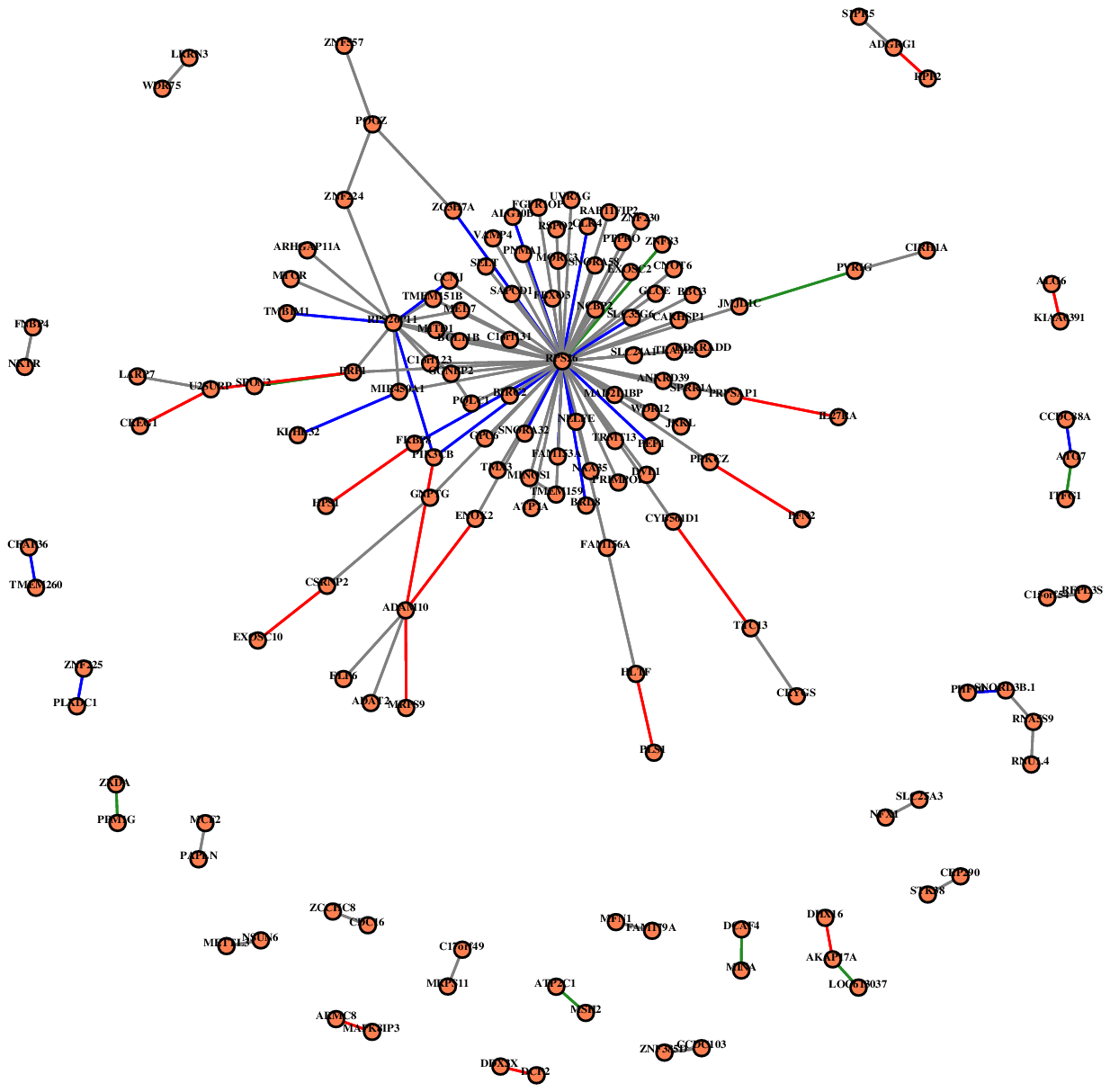}}
\subfigure[time 3]{
\label{fig15}
\includegraphics[width=2in]{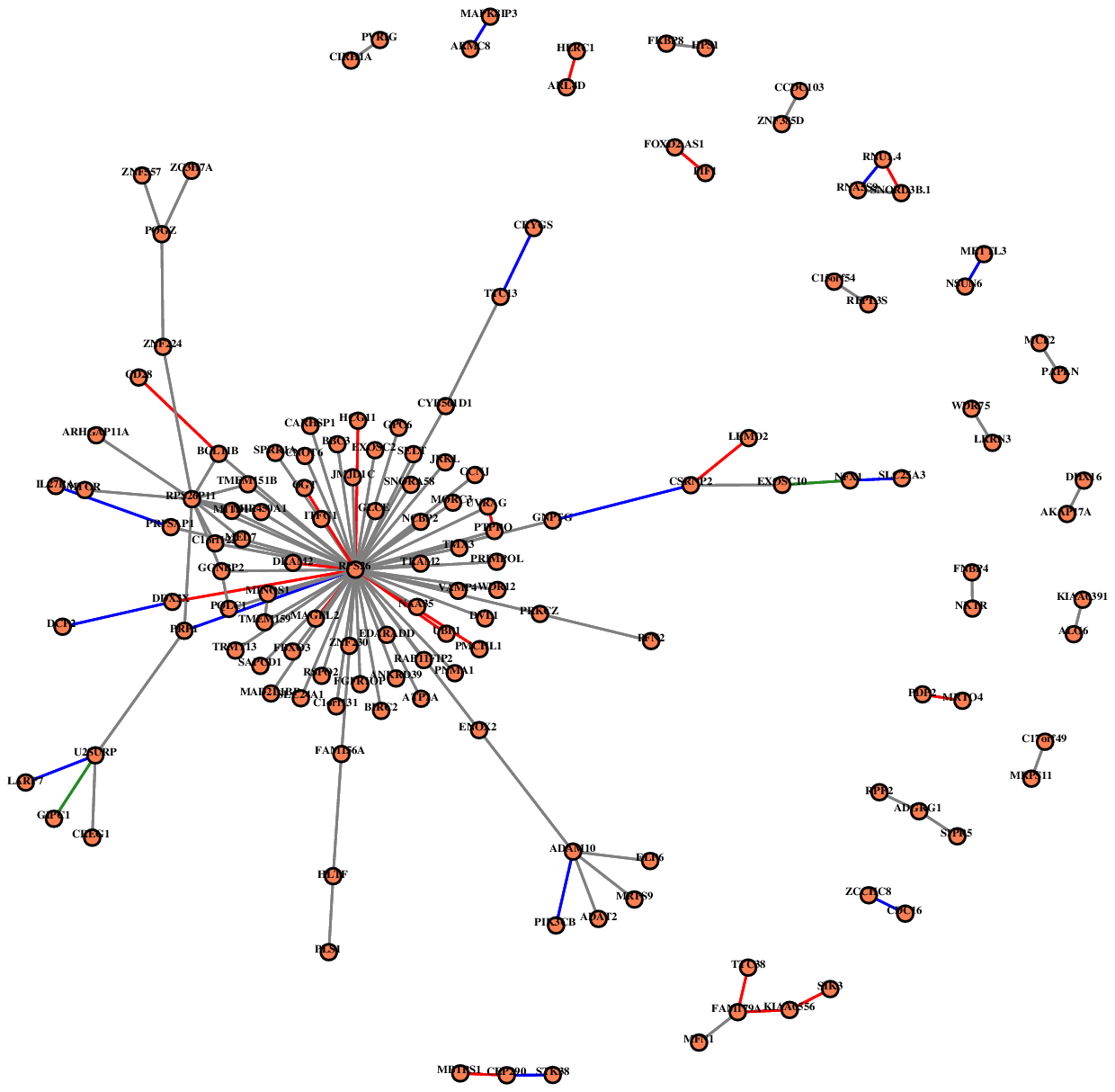}}
\\
\noindent
\subfigure[time 4]{
\label{fig16} 
\includegraphics[width=2in]{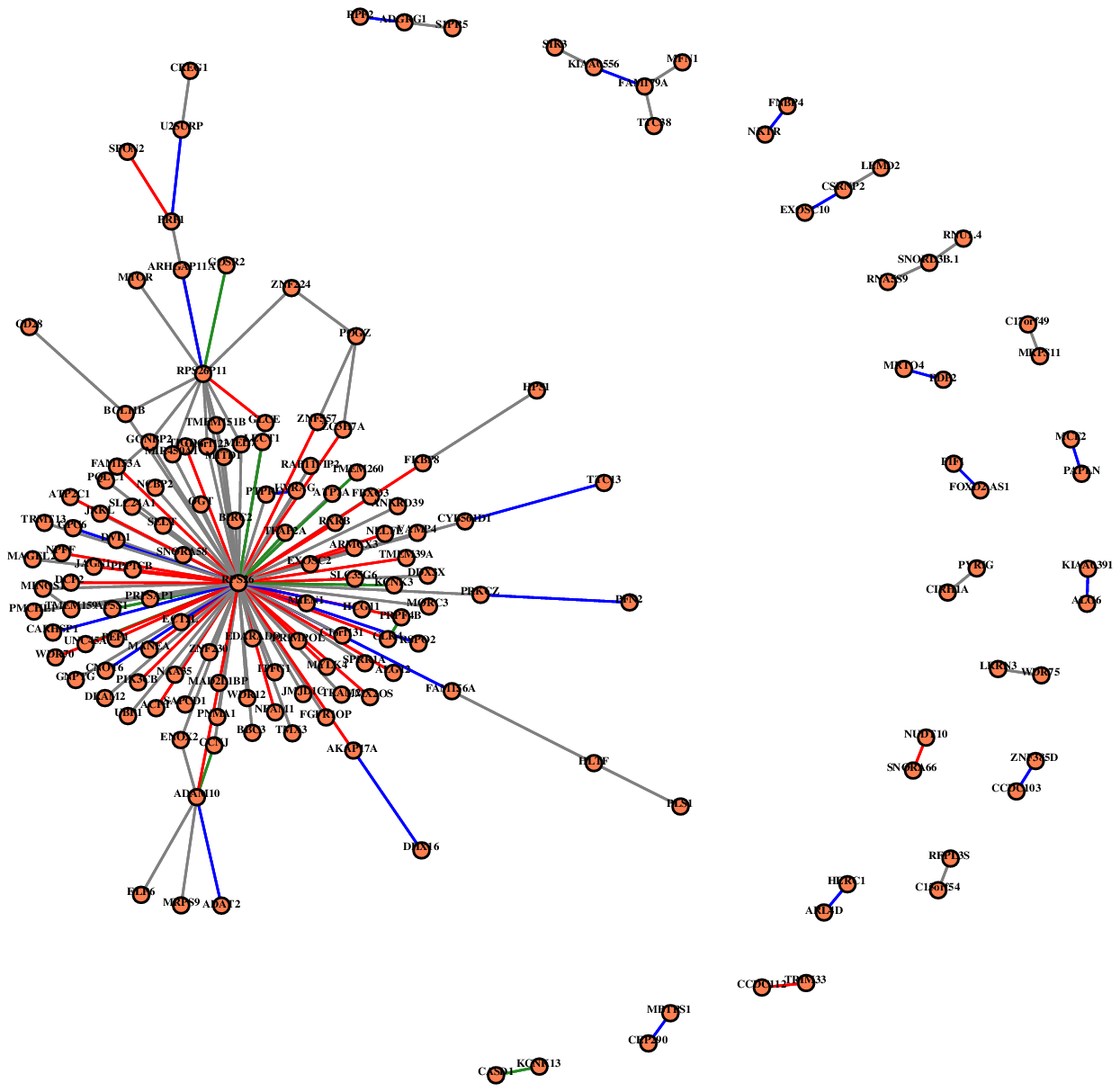}}
\subfigure[time 5]{
\label{fig17} 
\includegraphics[width=2in]{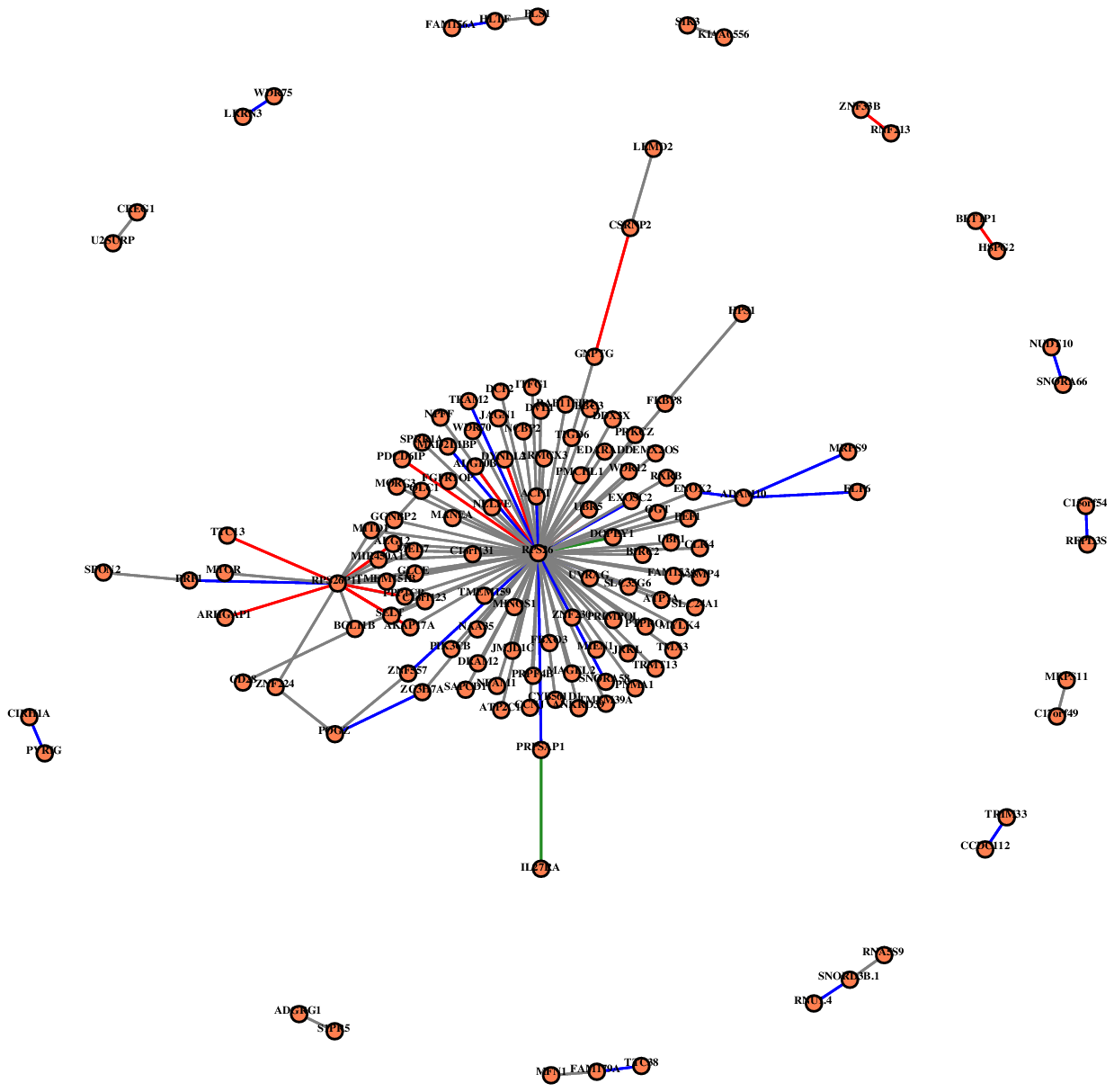}}
\subfigure[time 6]{
\label{fig18}
\includegraphics[width=2in]{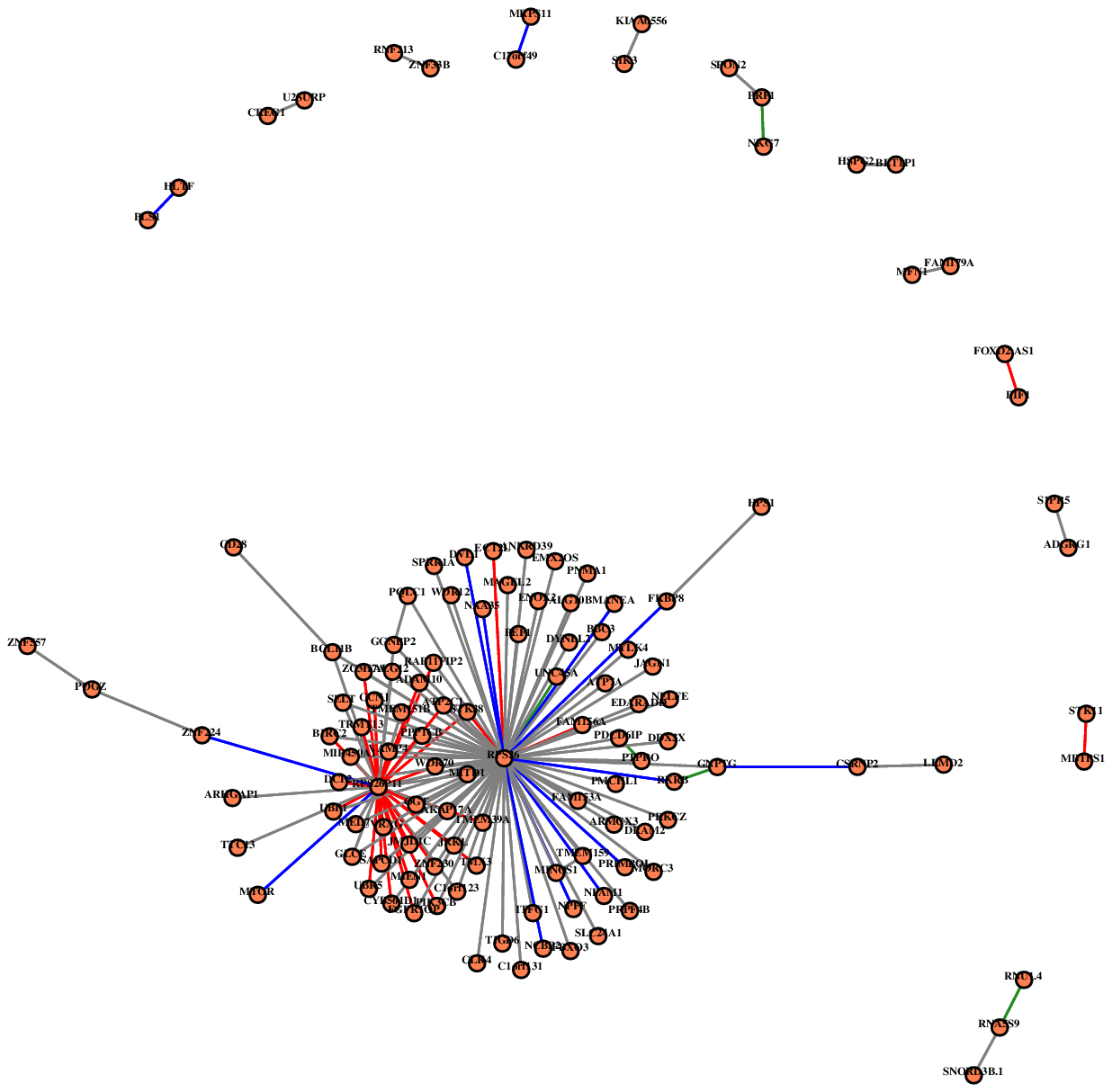}}
\\
\noindent
\subfigure[time 7]{
\label{fig19} 
\includegraphics[width=2in]{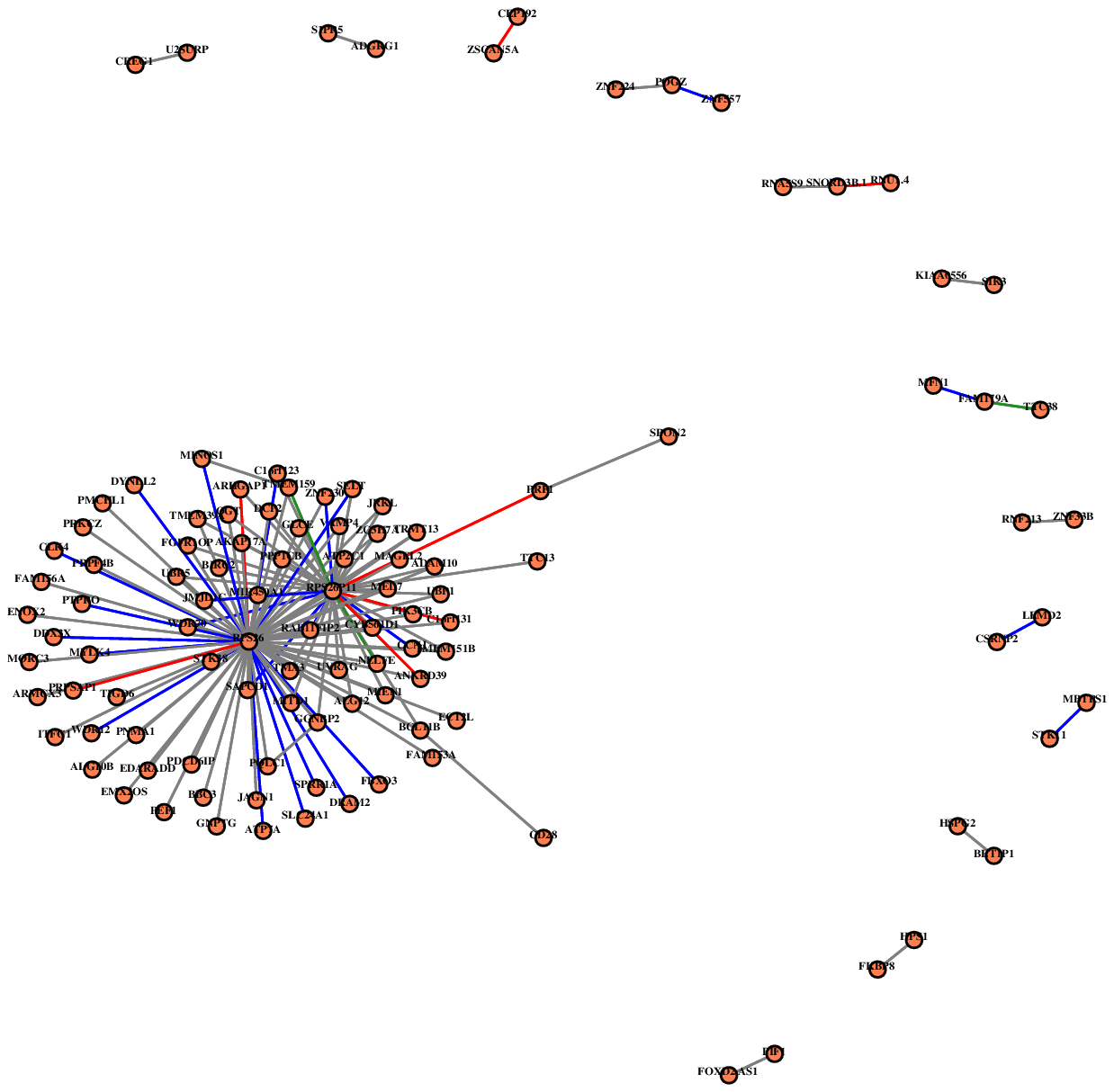}}
\subfigure[time 8]{
\label{fig20} 
\includegraphics[width=2in0]{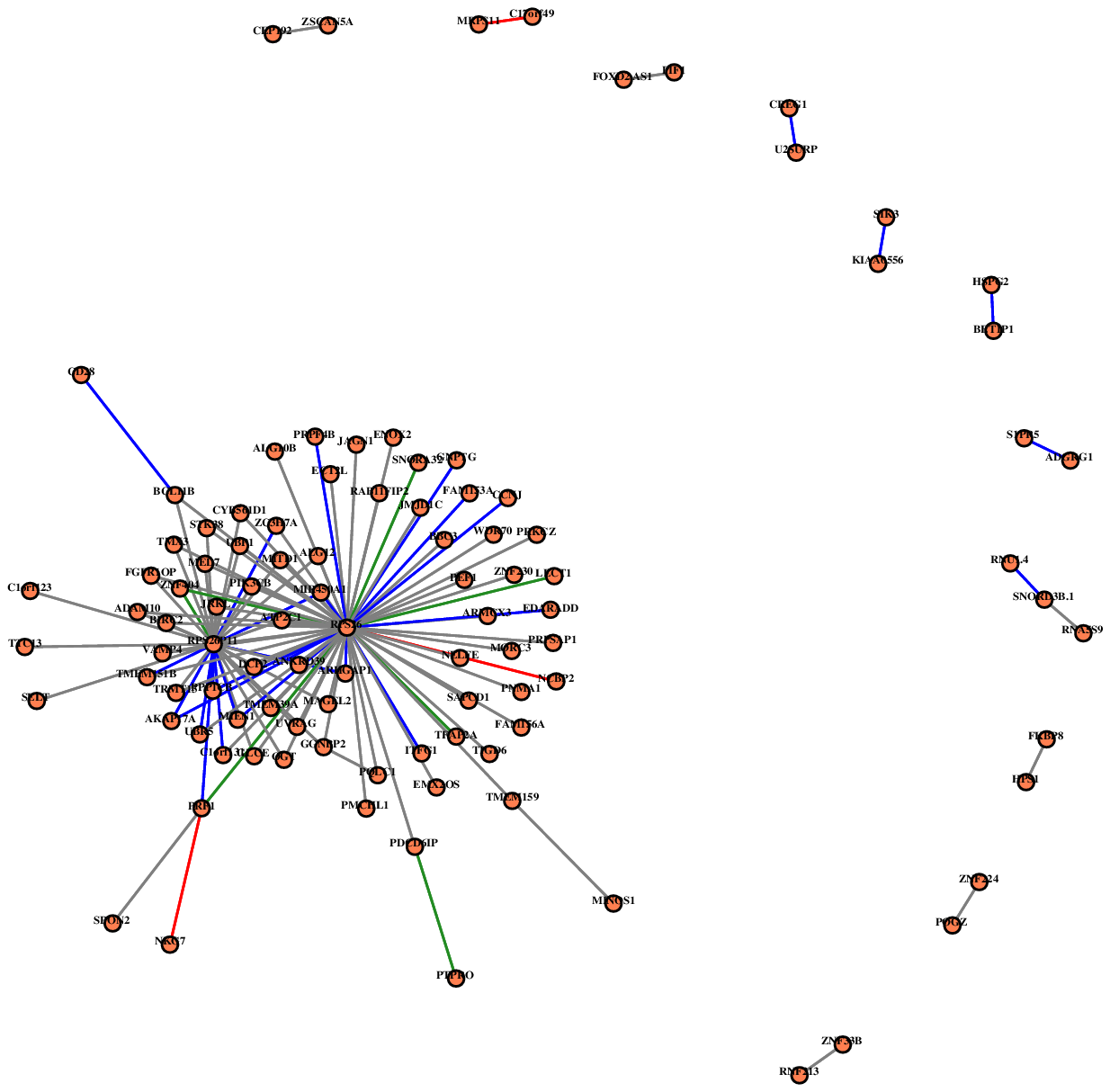}}
\subfigure[time 9]{
\label{fig21}
\includegraphics[width=2in]{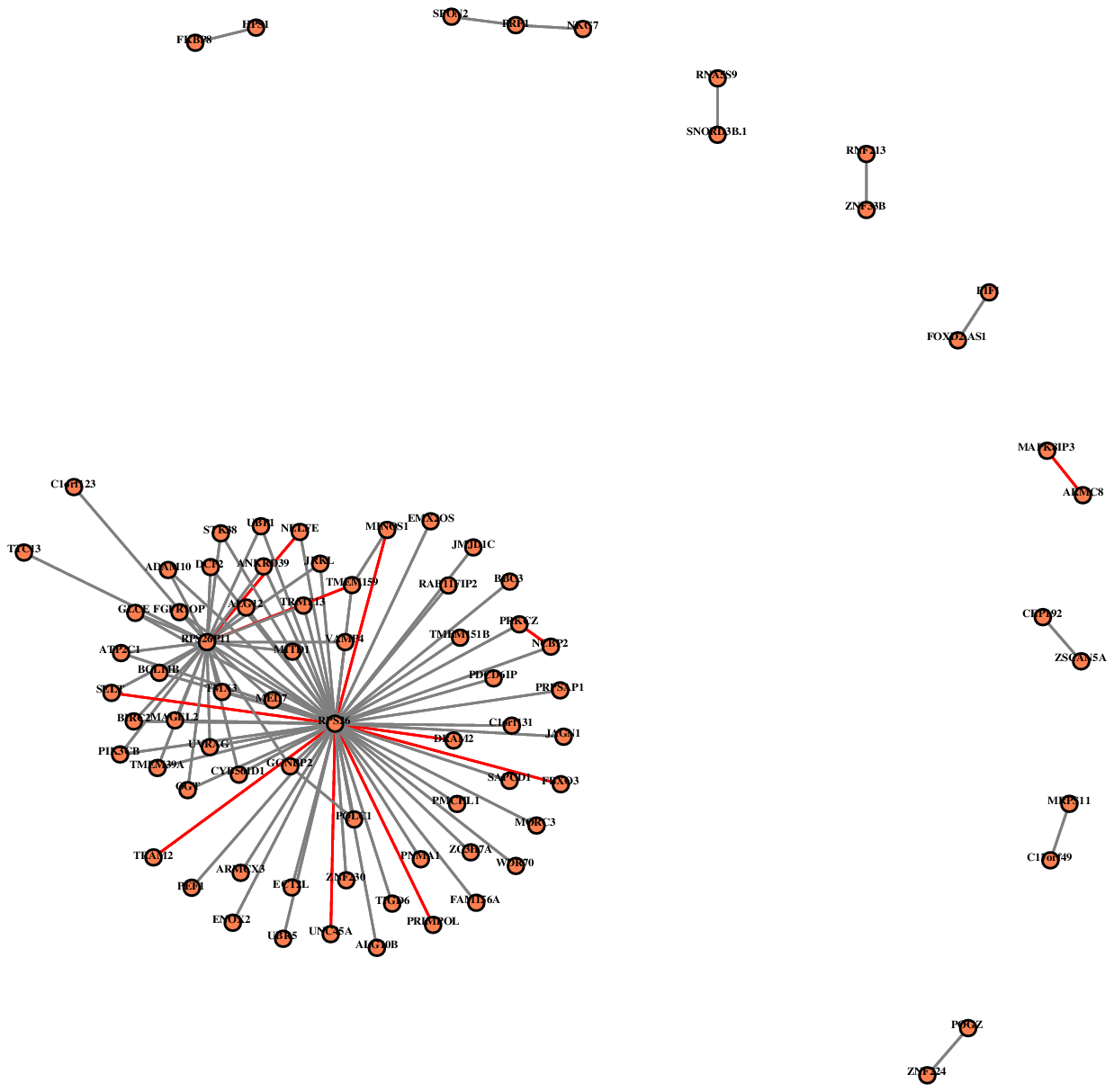}}
\\
\caption{Gene networks produced by FBIA for the case TEEDY samples at 9 time points. The red edge lines denote the new connections appearing in the current network compared with the network in the last one time point; the blue edge lines denote the disappearing connections in the network of next time point; the green edge lines denote that these lines are both new appearing connections and disappearing connections; the gray edge lines denote the unchanged connections in the current network and network in the last one time point.}
\label{fig_case} 
\end{figure}

Table \ref{tab_case} includes 11 different genes in total. Among the 11 genes, 9 genes have been verified in the literature 
to be T1D associated genes. For example, Schadt et al (2008) reported that RPS26 is a T1D causal gene, 
 and Ma and Hart (2013) reported that the gene 
 O-GlcNAc transferase (OGT) is directly linked to many metabolic diseases including diabetes. Other than identifying some verified T1D associated genes, we have also some new findings such as gene PRF1. Orilieri (2008) claimed that PRF1 variations are susceptibility factors for type 1 diabetes development. In Table \ref{tab_case},  PRF1 appeared as a hub gene twice, which suggests that the connection between PRF1 and  type 1 diabetes might be worth to be further explored. Moreover, we also identifies some connection changes in the networks. As showed in Figure \ref{fig_case}, the new appearing and disappearing connections are marked in different colors at each time point, which identify some evolvement patterns of the network.

\begin{table}[!h]
\tabcolsep=3pt\fontsize{10}{14}
\selectfont
\begin{center}
\caption{Top 5 hub genes identified by FBIA for the case TEDDY samples at 9 time points: 
 'Links' denotes the number of links of the gene to other genes, $k$ is the index of time points, 
 * indicates that there exist other genes which has the same number of links with this genes,
 $\Delta$ indicates that this gene has been verified as a T1D-related gene in the literature.}
\label{tab_case}
\vspace{2mm}
\begin{tabular}{ccccccccc}
\hline\hline
  \multicolumn{9}{c}{Case Group}\\\hline
    &Gene& Links  & &Gene& Links& &Gene& Links\\ \hline
   &  $^\Delta$RPS26   & 104   & & $^\Delta$RPS26   & 68 & & $^\Delta$RPS26 & 64\\ 
  &  $^\Delta$RPS26P11 & 40   &&  $^\Delta$RPS26P11 & 15 &&  $^\Delta$RPS26P11 & 12\\ 
     k=1& $^\Delta$ADAM10   & 4   &k=2 & $^\Delta$ADAM10   & 5 &k=3 & $^\Delta$ADAM10 & 5\\ 
 &  $^\Delta$ POGZ & 3   &&  $^\Delta$ PRF1 & 4 &&  U2SURP & 4\\ 
  &$^\Delta$TMEM159*   & 3   & & $^\Delta$ POGZ & 3 & & $^\Delta$BCL11B* & 3\\\hline\hline
  
     & $^\Delta$RPS26   & 99  & & $^\Delta$RPS26   & 91 & & $^\Delta$RPS26 & 86\\ 
  &  $^\Delta$RPS26P11 & 14   &&  $^\Delta$RPS26P11 & 18 &&  $^\Delta$RPS26P11 & 42\\ 
     k=4& $^\Delta$ADAM10   & 6   &k=5 & $^\Delta$ADAM10   & 4 &k=6 & $^\Delta$BCL11B & 3\\ 
 &  $^\Delta$BCL11B & 3   &&  $^\Delta$BCL11B & 3 &&  GNPTG & 3\\ 
  & $^\Delta$POGZ*   & 3   & & $^\Delta$POGZ*   & 3 & & $^\Delta$GGNBP2 & 3\\\hline \hline

 & $^\Delta$RPS26   & 78   & & $^\Delta$RPS26   & 70 & & $^\Delta$RPS26 & 61\\ 
  &  $^\Delta$RPS26P11 & 46   &&  $^\Delta$RPS26P11 & 39 &&  $^\Delta$RPS26P11 & 30\\ 
     k=7& $^\Delta$BCL11B & 3  &k=8 & $^\Delta$ PRF1   & 4 &k=9 & $^\Delta$TMEM159   & 3 \\ 
 & $^\Delta$TMEM159   & 3   &&  $^\Delta$BCL11B & 3 && $^\Delta$ GGNBP2 & 3 \\ 
  & $^\Delta$GGNBP2 & 3   & & $^\Delta$GGNBP2   & 3 & & $^\Delta$OGT*  & 2\\\hline\hline
\end{tabular}
\end{center}
\end{table}

For comparison, the GGL method was also applied to this example, 
for which the regularization parameters were chosen according
to the minimum AIC criterion. The total CPU time cost by the method was 20.2 hours. 
FGL was not applied to this example, as it would take extremely long CPU time.  
 Figure \ref{fig_GGL} shows the networks constructed by GGL for the
 case samples at all 9 time points. 
 Table \ref{tab_GGL} shows the top 5 hub genes identified by GGL at each time point for the case samples.
 The lists of hub genes are pretty stable, which consists of 7 different genes only. 
 Among the 7 genes, only 3 genes RPS26, OGT and JMJD1C have been verified in the literature as 
 T1D-associated genes.  
 Moreover, as showed in Figure \ref{fig_GGL}, the hub genes in networks are almost identical at each time point. 

\begin{figure}[]
\centering
\subfigure[time 1]{
\includegraphics[width=2in]{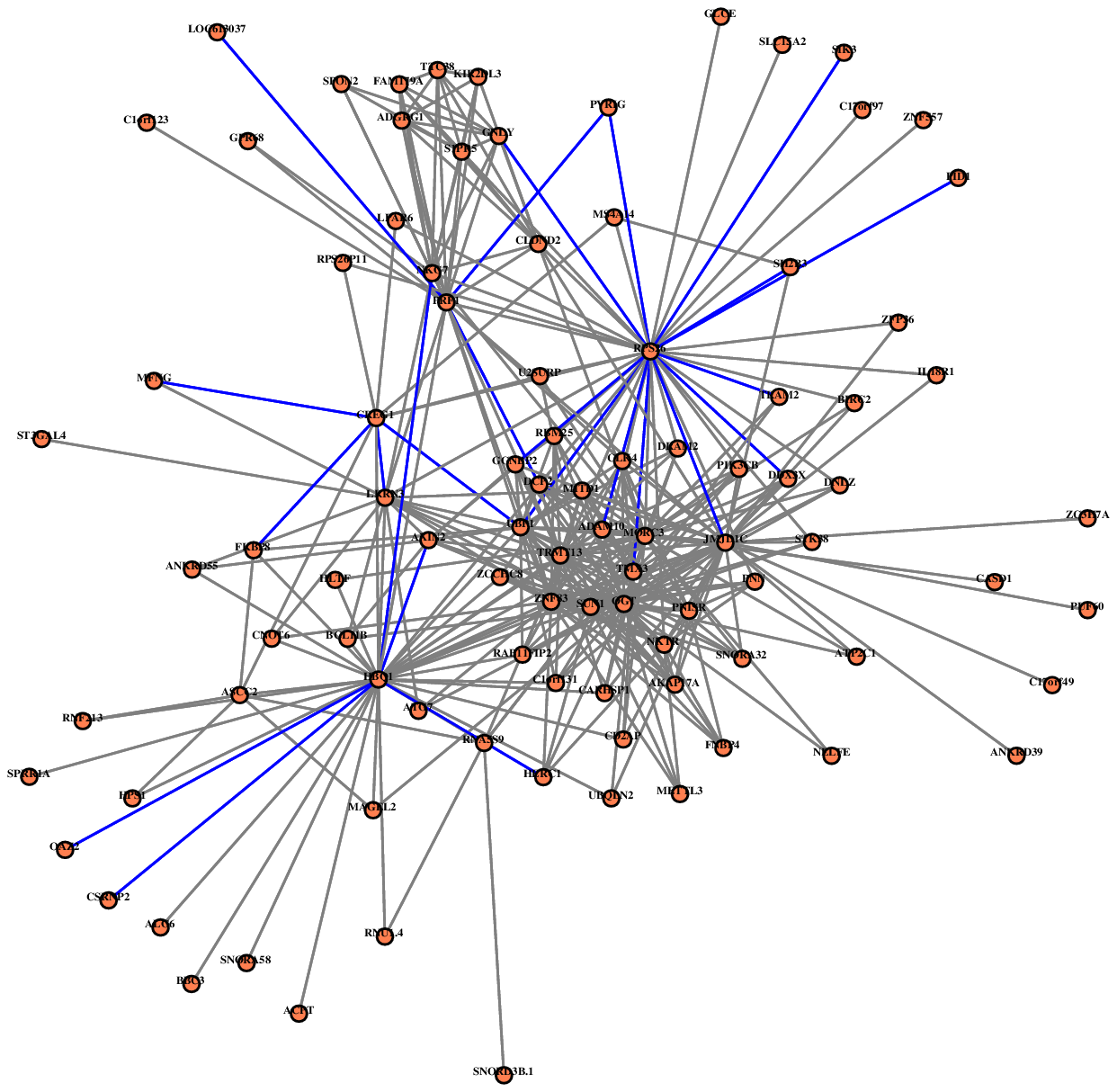}}
\subfigure[time 2]{
\includegraphics[width=2in]{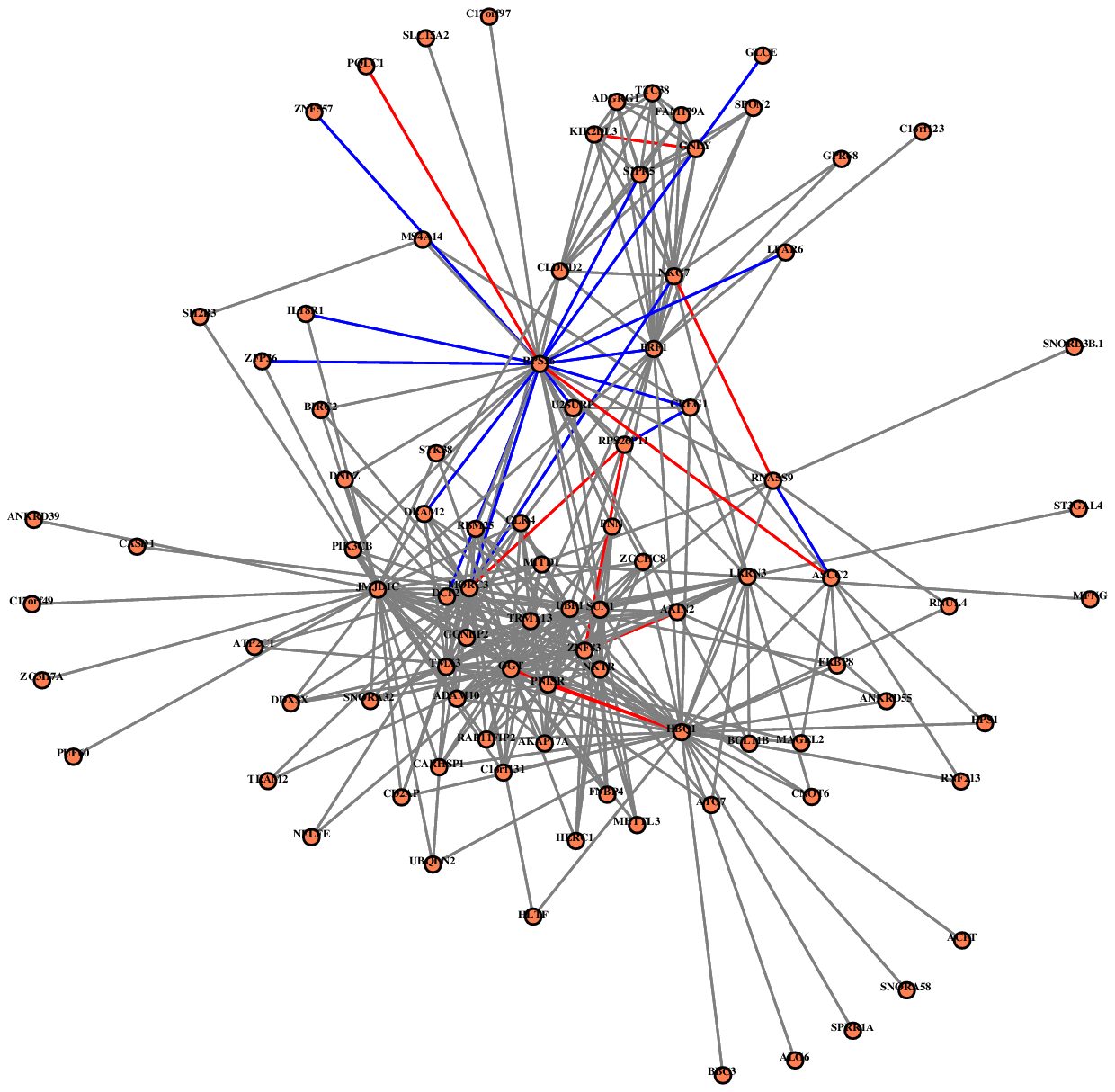}}
\subfigure[time 3]{
\includegraphics[width=2in]{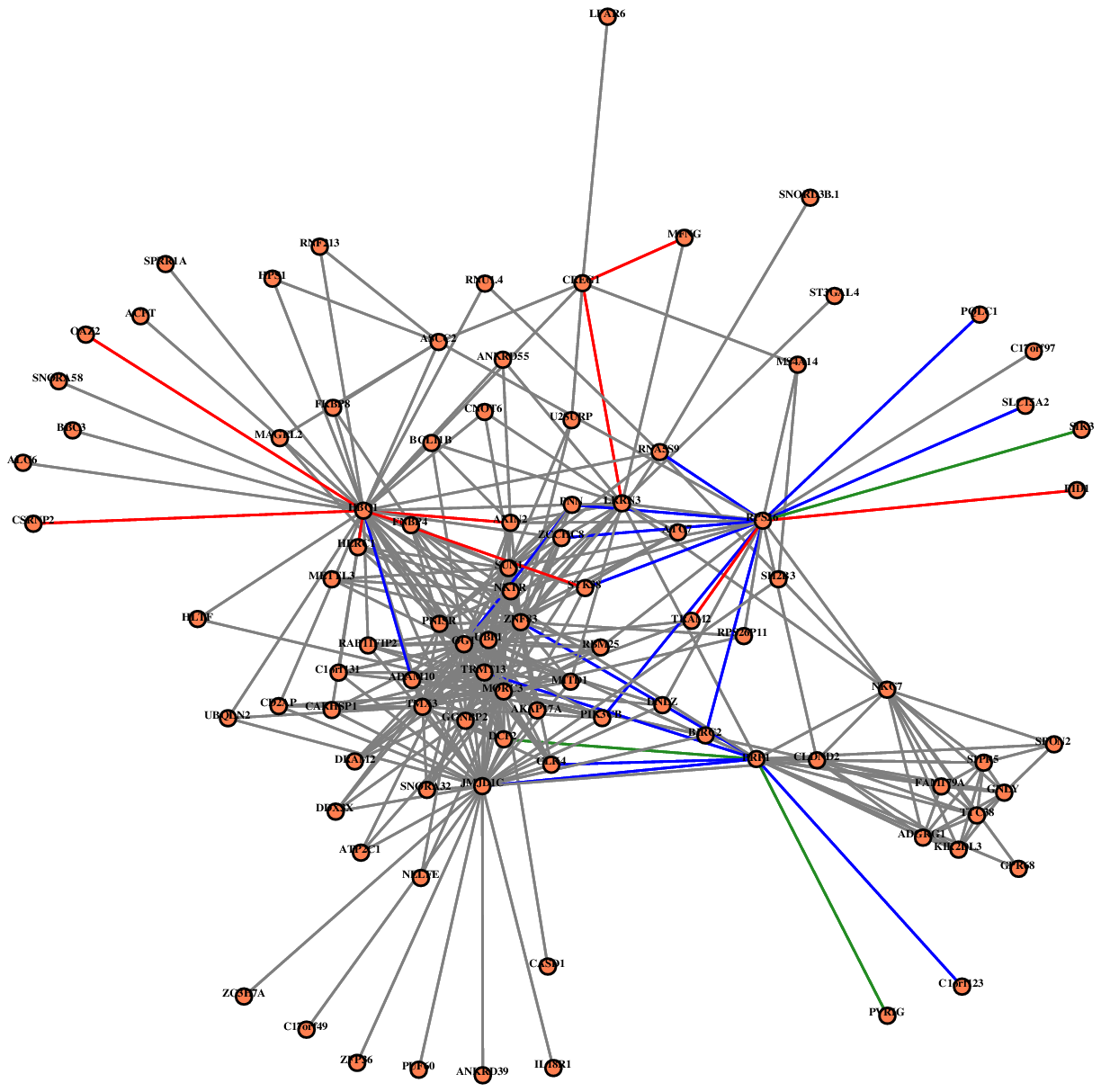}}
\\
\noindent
\subfigure[time 4]{
\includegraphics[width=2in]{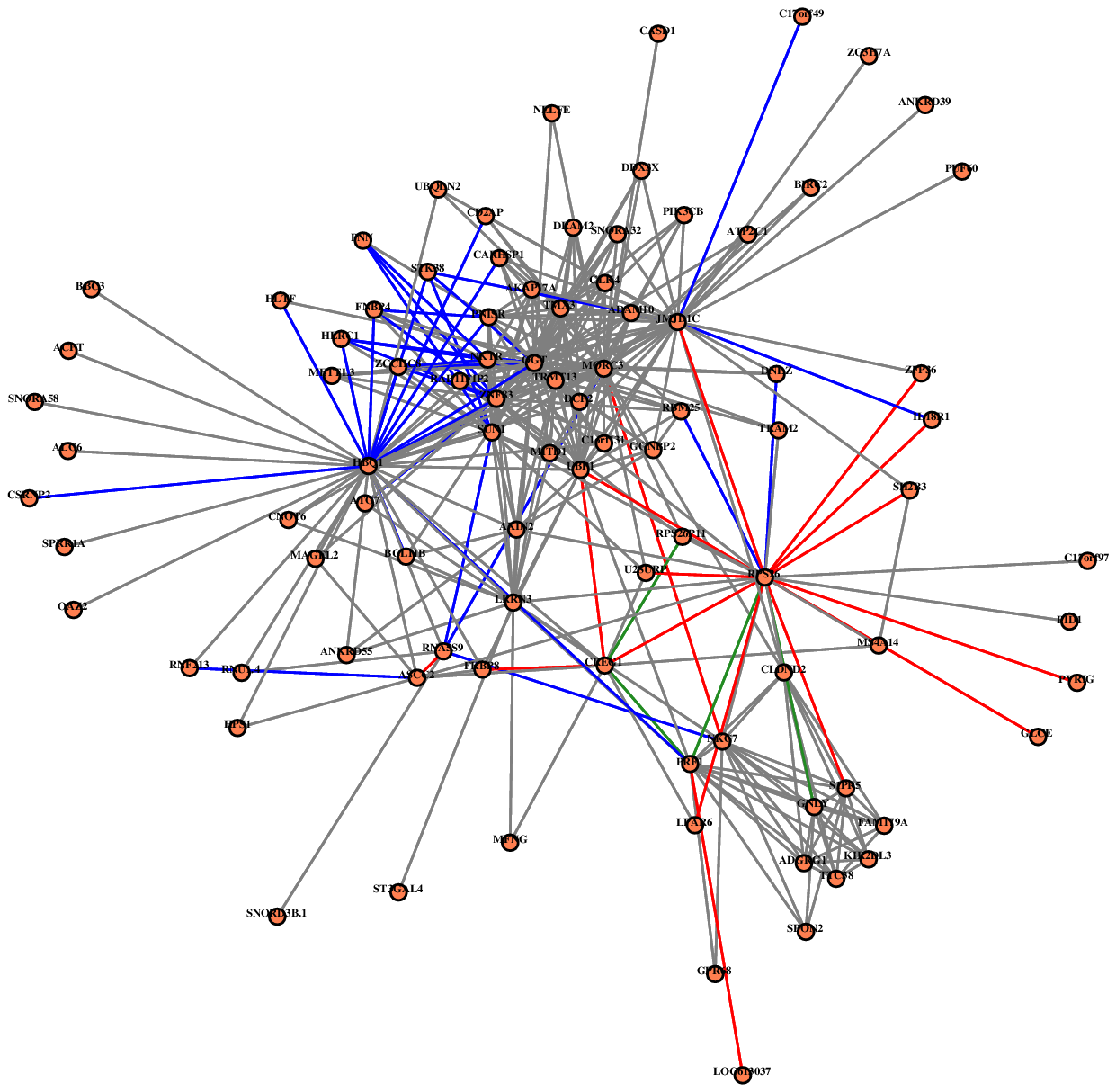}}
\subfigure[time 5]{
\includegraphics[width=2in]{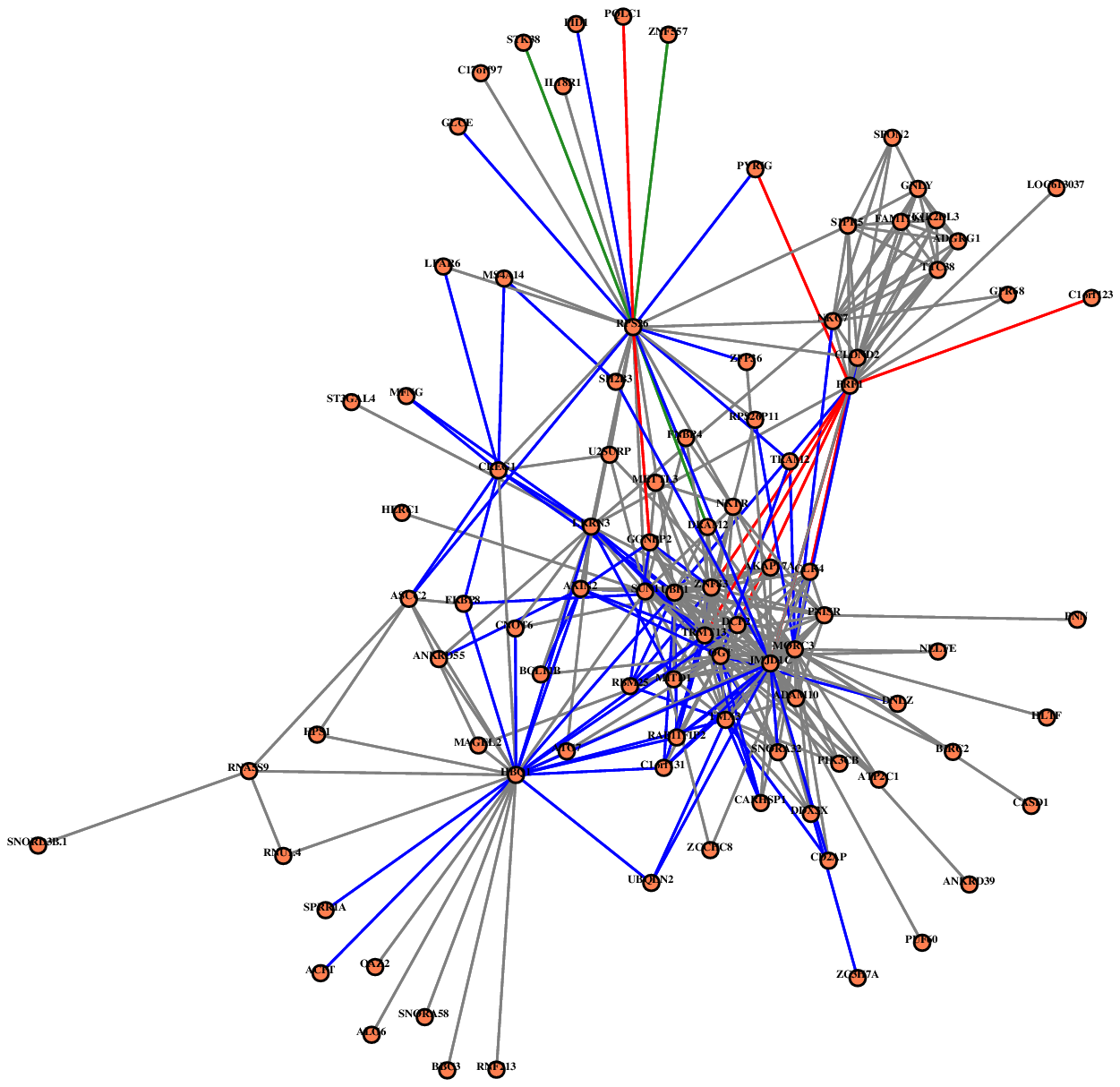}}
\subfigure[time 6]{
\includegraphics[width=2in]{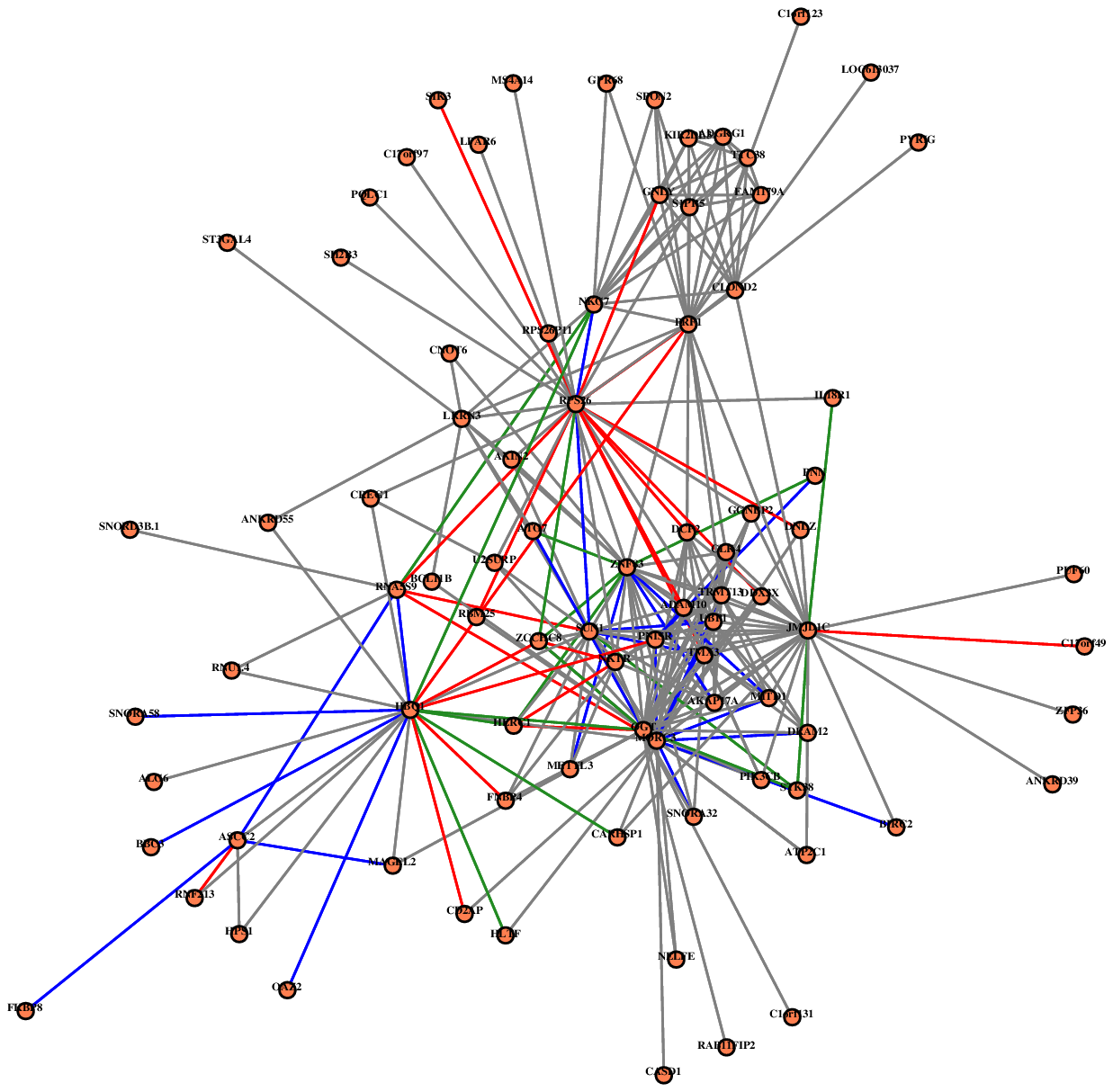}}
\\
\noindent
\subfigure[time 7]{
\includegraphics[width=2in]{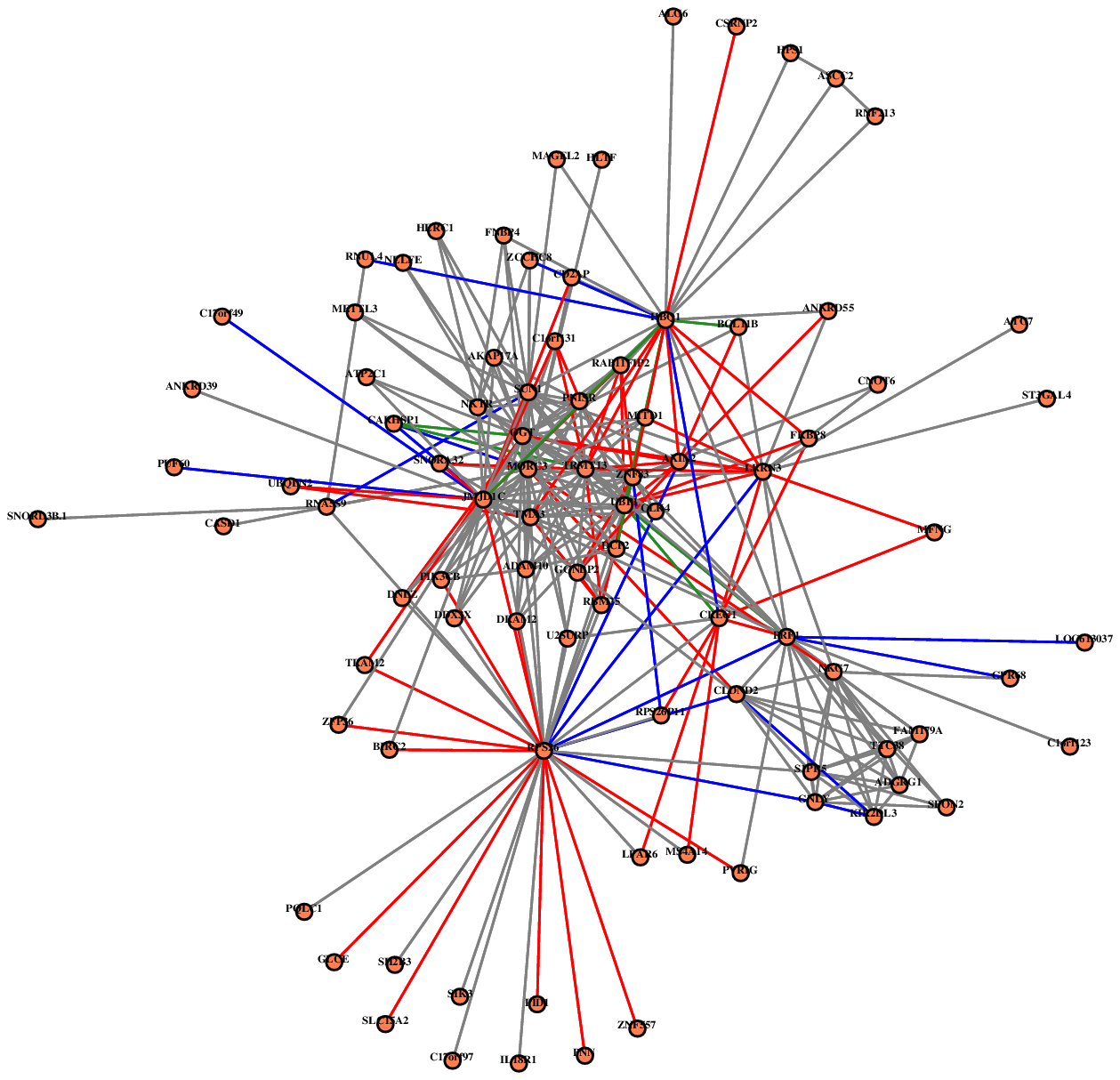}}
\subfigure[time 8]{
\includegraphics[width=2in]{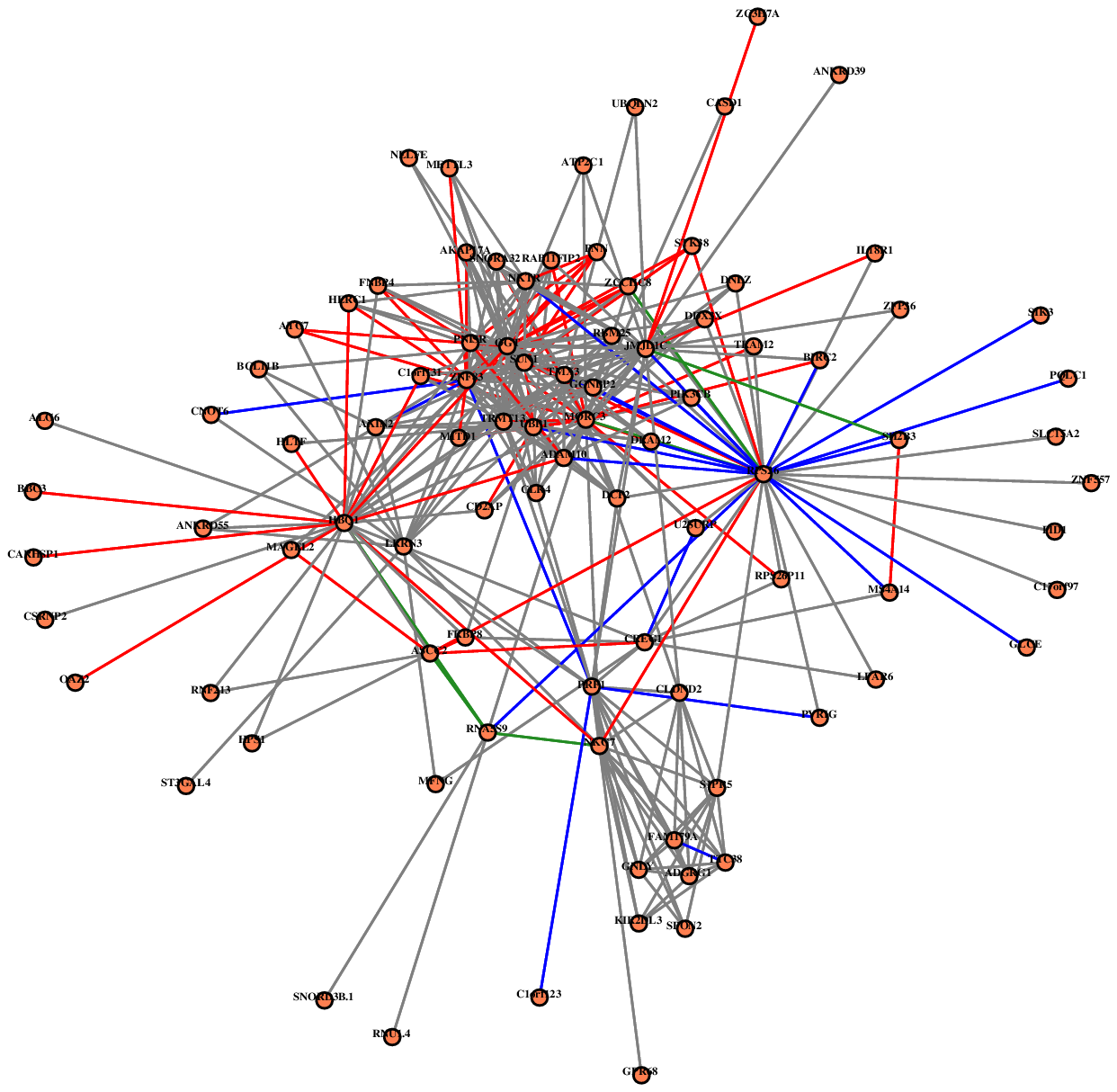}}
\subfigure[time 9]{
\includegraphics[width=2in]{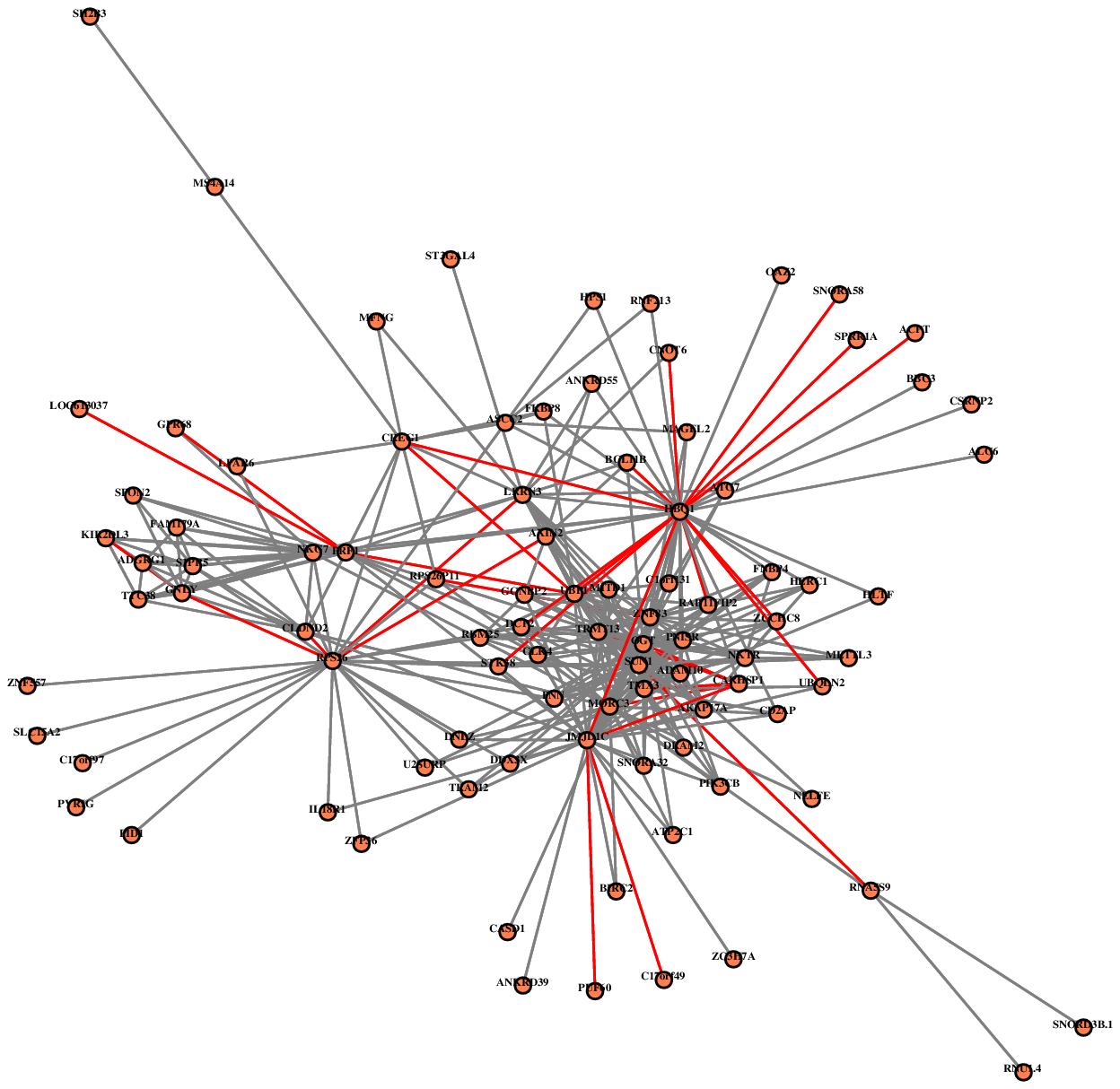}}
\\
\caption{Gene networks constructed by GGL for the case TEDDY samples at 9 time points. The red edge lines denote the new connections appearing in the current network compared with the network in the last one time point; the blue edge lines denote the disappearing connections in the network of next time point; the green edge lines denote that these lines are both new appearing connections and disappearing connections; the gray edge lines denote the unchanged connections in the current network and network in the last one time point.}
\label{fig_GGL} 
\end{figure}

\begin{table}[!h]
\tabcolsep=3pt\fontsize{10}{14}
\selectfont
\begin{center}
\caption{ Top 5 hub genes identified by GGA for the case TEDDY samples at 9 time points:
 'Links' denotes the number of links of the gene to other genes, $k$ is the index of time points,
 * indicates that there exist other genes which has the same number of links with this genes,
 $\Delta$ indicates that this gene has been verified as a T1D-related gene in the literature. }
\label{tab_GGL}
\vspace{2mm}
\begin{tabular}{ccccccccc}
\hline\hline
  \multicolumn{9}{c}{Case Group}\\\hline
    &Gene& Links  & &Gene& Links& &Gene& Links\\ \hline
   &  $^\Delta$RPS26   & 42   & & $^\Delta$ OGT   & 39 & & HBQ1 & 42\\ 
   
  &  HBQ1 & 39   &&  $^\Delta$JMJD1C   & 37 &&  $^\Delta$OGT & 39\\ 
  
     k=1& $^\Delta$OGT   & 38   &k=2 &HBQ1 & 36 &k=3 & $^\Delta$JMJD1C & 37\\ 
     
 &  $^\Delta$JMJD1C   & 38   &&  MORC3   & 35 &&  MORC3   & 33\\ 
 
  & MORC3   & 34  & & $^\Delta$RPS26   & 32 & & ZNF83 & 29\\\hline\hline

     & HBQ1   & 41   & & $^\Delta$JMJD1C    & 35 & & $^\Delta$OGT & 33\\ 
     
  &  $^\Delta$OGT & 38  &&  $^\Delta$OGT & 34 &&  $^\Delta$RPS26 & 29\\ 
  
     k=4& $^\Delta$JMJD1C     & 37  &k=5 & MORC3   & 33 &k=6 & $^\Delta$JMJD1C   & 29\\ 
     
 &  MORC3 & 34   &&  $^\Delta$RPS26 & 28 &&   MORC3  & 28\\ 
 
  & ZNF83  & 28   & &HBQ1   & 28 & &  ZNF83 & 24\\\hline \hline

 & $^\Delta$RPS26   & 38   & & $^\Delta$RPS26    & 39 & & HBQ1  & 41\\ 
 
  &  $^\Delta$OGT & 34   &&  $^\Delta$OGT & 38 &&  $^\Delta$OGT & 39\\ 
  
     k=7& $^\Delta$JMJD1C   & 34  &k=8 & MOR3 & 34 &k=9 & $^\Delta$JMJD1C  & 36\\ 
     
 & HBQ1 & 27   &&  HBQ1 & 29 && MOR3  & 34\\ 
 
  & MOR3   & 23   & & ZNF83   & 28 & &  ZNF83     & 26\\\hline\hline
\end{tabular}
\end{center}
\end{table}

To further assess the quality of the networks produced by FBIA and GGL, 
we fit them by the power law curve (see, e.g., Kolaczyk 2009, pp.80-85). 
A nonnegative random variable $X$ is said to have a power law distribution if
\begin{equation}
P(X=x)\propto x^{-\upsilon},
\end{equation}
for some positive constant $\upsilon$. 
The power law states that the majority of nodes are of very low degree, although some are of
much higher degree. A network whose degree distribution follows the power law is called a scale-free network 
and it has been verified that many biological networks, such as gene expression networks, 
protein-protein interaction networks, and metabolic networks (Barab\'{a}si and Albert 1999), follow
the power law. 
As shown in Figure \ref{power}, where the connections of all 9 networks are combined to 
 generate a single log-log plot for each method, the networks produced by 
 FBIA seem to be more fit to the power law than those generated by GGL.  
 GGL tends to identify too many high connectivity genes. 

\begin{figure}
\centering
\subfigure[FBIA]{
\label{power} 
\includegraphics[width=2in,angle=270]{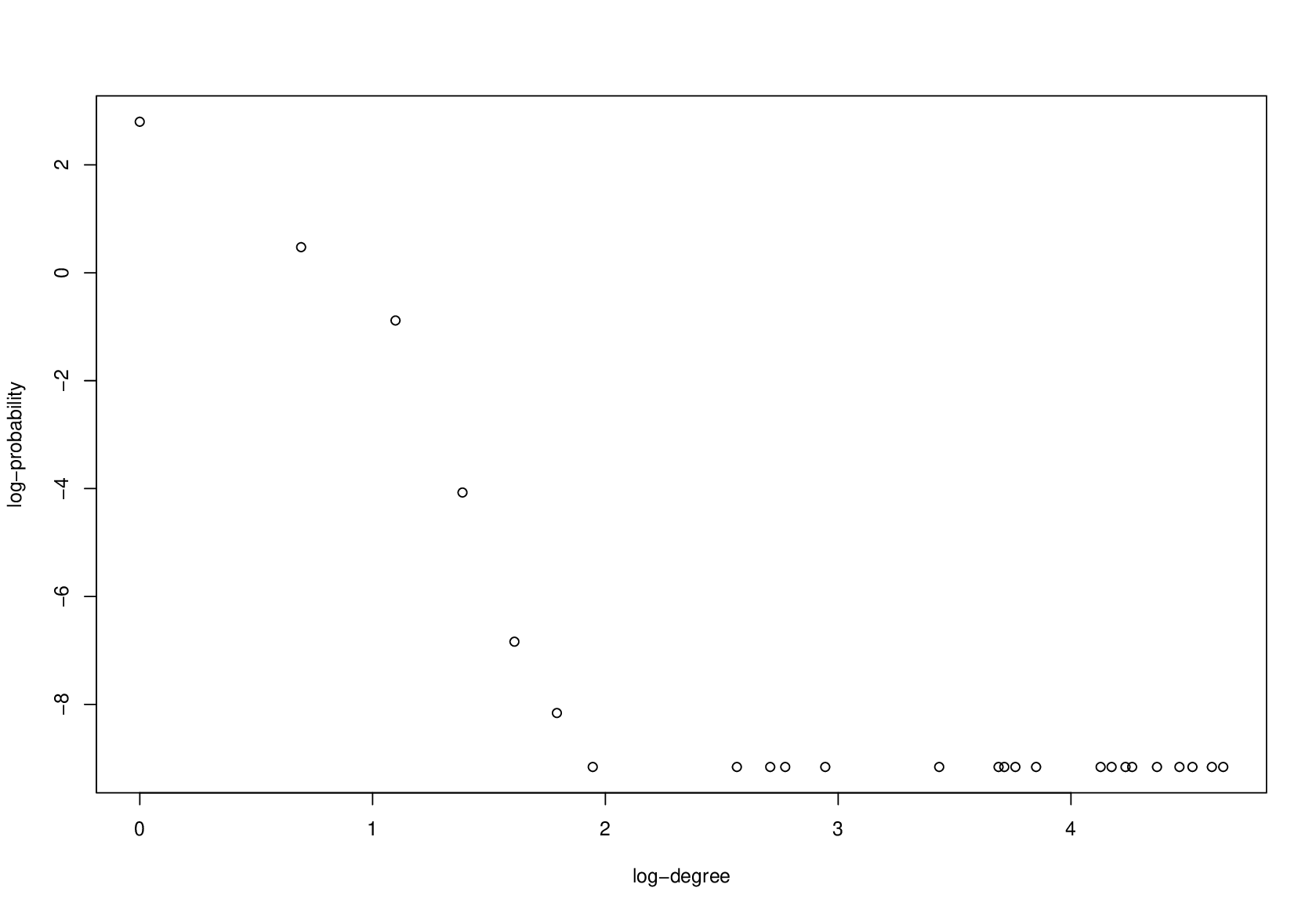}}
\hspace{0.5in}
\subfigure[GGL]{
\includegraphics[width=2in,angle=270]{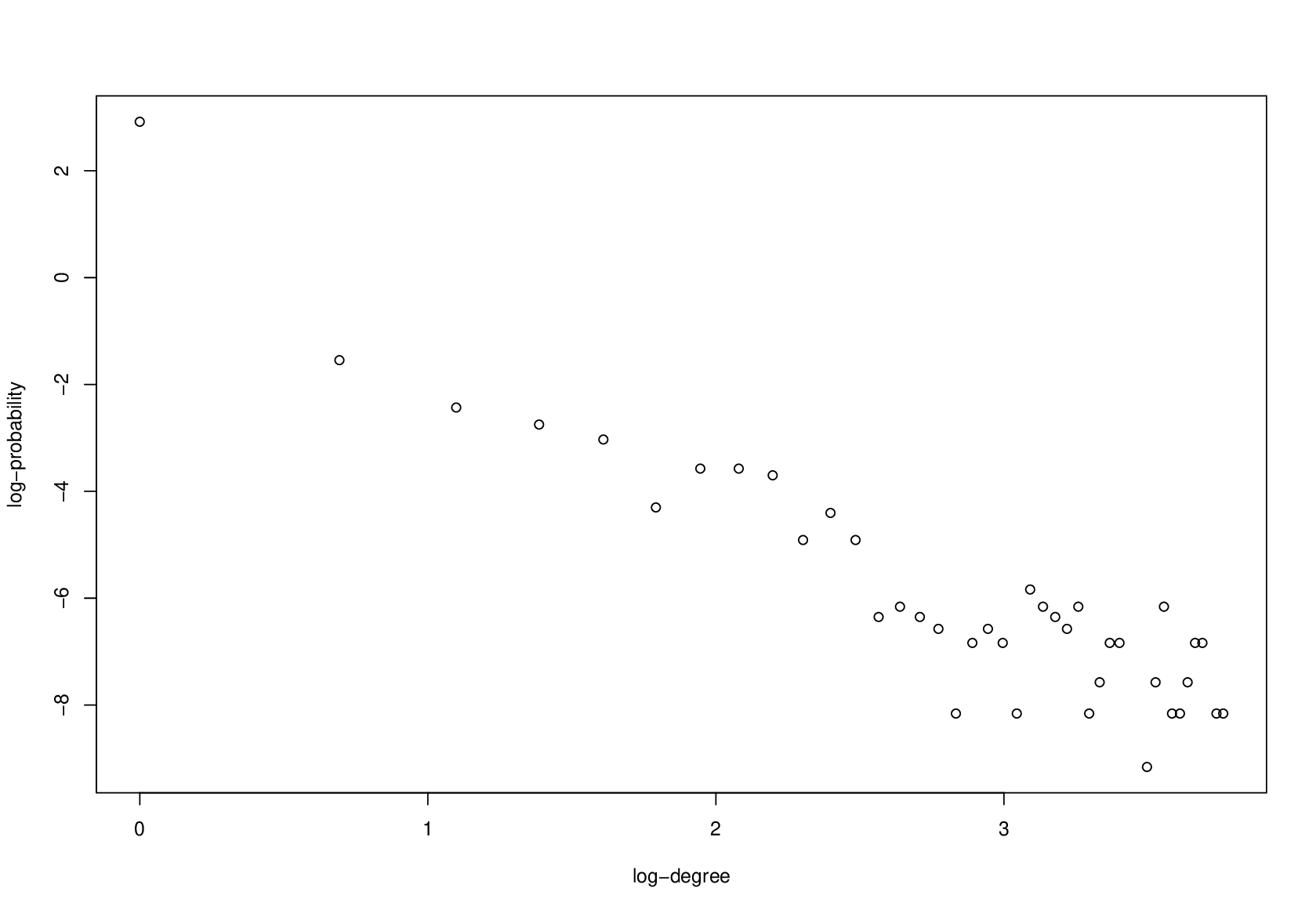}}
\\
\caption{Power law plots generated by FBIA (left) and GGL (right) for case TEDDY samples. } 
\end{figure}

 In summary, FBIA tends to outperform GGL for this real data example. 
 First, FBIA can identify more hub genes which are associated with T1D.  
 Second, the networks produced by FBIA are more fit to the power law than those
 generated by GGL. 

From the perspective of data analysis, one might also be interested in estimating the 
gene networks constructed from the controls, as well as the differences 
between the networks from the cases and controls. 
For comparing the networks from the cases and controls, we can adopt the method
described in Section 6 of Liang et al. (2015). However, since the method by Liang et al. (2015) requires that
the two networks under comparison are independent, the sample information from the cases and controls
should not be integrated in this case. We left this work to the future.

\section{Discussion}

In this paper, we have proposed FBIA as a promising method for jointly estimating multiple GGMs under 
 distinct conditions and applied FBIA to TEDDY data.  
 The FBIA method consists of a few important steps, which is to first summarize the graph structure
 information contained in the data using the $\psi$-learning algorithm (Liang, Song and Qiu, 2015),
 then integrate information via a meta-analysis procedure under the Bayesian framework, and
 finally determine the structures of multiple graphs via a multiple hypothesis test. Compared to the existing
 methods, FBIA has a few significant advantages.
 First, FBIA includes a meta-analysis procedure to explicitly integrate 
 information across distinct conditions. However, the existing methods 
 often integrate information through prior distributions or penalty function, 
 which is usually less efficient. 
 Second, FBIA can be run very fast, especially when $K$ is small. 
 The overall computational complexity of FBIA is $O(p^2 2^K)$, where the factor $2^K$ is  
 the total number of possible configurations of an edge across all $K$ conditions. 
 When $K$ is large, we need to resort to MCMC for an efficient estimation of 
 the posterior probabilities $\pi(\be_l|\bpsi_l)$'s for $l=1,2,\ldots, p(p-1)/2$.  
 Since $\pi(\be_l|\bpsi_l)$'s can be estimated for each $l \in \{1,2,\ldots,p(p-1)/2\}$ independently, 
 this step can be done in parallel. In addition, we note that the correlation coefficients and 
 $\psi$-scores can also be calculated in parallel. 
 Hence, the whole method can be executed very fast on a parallel architecture.   
 Moreover, instead of working on the original data, the Bayesian integration step chooses to work 
 on the edge-wise $\psi$-scores, which avoids to invert high-dimensional covariance matrices 
 and thus can be very fast. Note that, in calculation of $\psi$-scores,
 the $\psi$-learning algorithm (Liang et al., 2015) also successfully avoids to 
 invert high-dimensional covariance matrices through correlation screening.   
 Third, FBIA can provide an overall uncertainty measure for the edges detected in the 
 multiple graphical models. This has been beyond the ability of many of the existing methods, 
 especially when $p$ is large.  

 The FBIA method has a very flexible framework, which can be easily extended to 
 joint estimation of multiple mixed graphical models. For example, we consider the
 scenario that the data consists of only Gaussian and multinomial random variables, 
 for which the joint distribution is well defined (Lee and Hastie, 2015).  
 For such mixed data, the $\psi$-learning algorithm can be performed under the 
 framework of generalized linear models; that is, we can replace the correlation 
 coefficients and $\psi$-partial correlation coefficients used in the algorithm 
 by the corresponding $p$-values obtained in the marginal variable screening tests (Fan and Song, 2010) 
 and conditional independence tests. 
 Then we can replace the $\psi$-scores by the $Z$-scores corresponding to the $p$-values 
 of the conditional independence tests.  
 For other types of continuous random variables, we can apply the 
 nonparanormal transformation (Liu et al., 2009) 
 to Gaussianize them prior to the application of the FBIA method.

\section*{Acknowledgments} 

This study was supported by grant 2015PG-T1D050 provided by the Leona M. and Harry B. Helmsley Charitable Trust. 
Liang's research was support in part by the grants USF-ITN-15-11-MH, DMS-1612924, DMS/NIH R01-GM117597, and NIH R01-GM126089.  
The TEDDY Study is funded by U01 DK63829, U01 DK63861, U01 DK63821, U01 DK63865, U01 DK63863, U01 DK63836, U01 DK63790, UC4 DK63829, UC4 DK63861, UC4 DK63821, UC4 DK63865, UC4 DK63863, UC4 DK63836, UC4 DK95300, UC4 DK100238, UC4 DK106955, UC4 DK112243, UC4 DK117483, and Contract No. HHSN267200700014C from the National Institute of Diabetes and Digestive and Kidney Diseases (NIDDK), National Institute of Allergy and Infectious Diseases (NIAID), National Institute of Child Health and Human Development (NICHD), National Institute of Environmental Health Sciences (NIEHS), Centers for Disease Control and Prevention (CDC), and JDRF. This work supported in part by the NIH/NCATS Clinical and Translational Science Awards to the University of Florida (UL1 TR000064) and the University of Colorado (UL1 TR001082).  Members of the TEDDY Study Group are listed in the Supplementary File.
The authors thank Dr. George Tseng for his comments/suggestions on the FBIA method.  

\section*{Contributions of Authors} 

 Liang initiated the project, proposed the FBIA method, and participated the writing of the manuscript; 
 Jia conducted the simulation and data analysis, and participated the development of the FBIA method as well as  
 the writing of the manuscript. 
 TEDDY Study Group provided the real dataset as well as the grant support in part to the research. 
  
\appendix
\section{Appendix: Three Component Mixture Distribution} \label{App:AppendixA}

\subsection{Bayesian Clustering and Meta-Analysis.}

Considering the scores $(\psi_l^{(k)})$ follow a three-component Gaussian mixture distribution:

\begin{equation}\label{pluginA}
     p(\psi_{l}^{(k)}|e_{l}^{(k)})=\left\{\begin{array}{ll}
                  N(\mu_{l0},\sigma_{l0}^2),&\textrm{if $e_{l}^{(k)}=-1$},\\
                  N(\mu_{l1},\sigma_{l1}^2),&\textrm{if $e_{l}^{(k)}=0$},\\
                  N(\mu_{l2},\sigma_{l2}^2),&\textrm{if $e_{l}^{(k)}=1$},
                \end{array}\right.
\end{equation}
for $l=1,2,\ldots, N$ and $k=1,2,\ldots, K$. Each pair $(l,k)$ corresponds to one candidate edge in 
 graph $k$ and $e_l^{(k)}$ is the indicator for the status of 
 edge $l$ in graph $k$; $e_l^{(k)}=-1$ if the edge exists with a large negative $\psi$-score, $e_l^{(k)}=0$ 
 if the edge does not exist, and $e_l^{(k)}=1$ if the edge exists with a large positive $\psi$-score. 
 It is reasonable to assume that the components 
 $N(\mu_{l0},\sigma_{l0}^2)$ , $N(\mu_{l1},\sigma_{l1}^2)$ and $N(\mu_{l2},\sigma_{l2}^2)$ 
 are all independent of $k$.  
Let $\bpsi_l=(\psi_{l}^{(1)},...,\psi_{l}^{(K)})$ and 
 $\be_{l}=(e_{l}^{(1)},...,e_{l}^{(K)})$. Conditioned on $\be_l$, the joint likelihood function of $\bpsi_l$ is given by  
\begin{equation} \label{jointeq3}
 p(\bpsi_{l}|\be_{l},\mu_{l0},\sigma_{l0}^2,\mu_{l1},\sigma_{l1}^2) 
  = \prod_{\{k:e_{l}^{(k)}=-1\}}\phi(\psi_{l}^{(k)}|\mu_{l0},\sigma_{l0}^2)\prod_{\{k:e_{l}^{(k)}=0\}}
  \phi(\psi_{l}^{(k)}|\mu_{l1},\sigma_{l1}^2)\prod_{\{k:e_{l}^{(k)}=1\}}
  \phi(\psi_{l}^{(k)}|\mu_{l2},\sigma_{l2}^2), 
\end{equation}
where $\phi(.|\mu,\sigma^2)$ is the density function of the Gaussian distribution 
 with mean $\mu$ and variance $\sigma^2$. Then we still consider two types of priors 
 for $\be_l$'s, namely, temporal prior and spatial prior.

 \subsubsection{Temporal Prior}

 To enhance the similarity for the networks between adjacent conditions, we let $\be_l$ be subject to 
 the following prior distribution 
\begin{equation}\label{priorA2}
  p(\be_l|\bq) = \frac{(K-1)!}{N_{l0}! N_{l1}! N_{l2}!}\hspace{0.2cm} q_0^{N_{l0}}\cdot q_1^{N_{l1}}\cdot q_2^{N_{l2}},
\end{equation}
 where $\sum_{i=0}^{2}q_i=1$, and $N_{li}=\#\{k: |e_l^{(k+1)}-e_l^{(k)}|=i, k=1,2,\ldots, K-1\}$ 
 denotes the number of edges with the changing mode $i$,  and 
 $\bq=(q_0,q_1,q_2)$ are the prior probabilities for different changing modes. 
 In this paper, we assume that $\bq$ follows a Dirichlet distribution $Dir(\alpha_0,\alpha_1,\alpha_2)$, 
 where $\alpha_0$, $\alpha_1$ and $\alpha_2$ are pre-specified positive parameters.  
 Further, we let $\mu_{l0}$, $\mu_{l1}$ and $\mu_{l2}$ be subject to an improper uniform distribution, i.e., 
 $\pi(\mu_{l0}) \propto 1$, $\pi(\mu_{l1}) \propto 1$and $\pi(\mu_{l2}) \propto 1$, and let $\sigma_{l0}^2$, $\sigma_{l1}^2$ and $\sigma_{l2}^2$  
 be subject to an inverted-gamma distribution, i.e., 
 $\sigma_{l0}^2, \sigma_{l1}^2, \sigma_{l2}^2  \sim IG(a_2,b_2)$, where $a_2$ and $b_2$ are pre-specified constants. 
 Then the joint posterior distribution of $(\be_l,\mu_{l0},\sigma_{l0}^2, \mu_{l1}, \sigma_{l1}^2, \mu_{l2}, \sigma_{l2}^2, \bq)$ 
 is given by 
 \[
 \pi(\be_l,\mu_{l0},\sigma_{l0}^2, \mu_{l1}, \sigma_{l1}^2, \mu_{l2}, \sigma_{l2}^2, \bq|\bpsi_l)\propto 
 p(\bpsi_l|\be_l,\mu_{l0},\sigma_{l0}^2, \mu_{l1}, \sigma_{l1}^2, \mu_{l2}, \sigma_{l2}^2) \pi(\mu_{l0},\sigma_{l0}^2, \mu_{l1}, \sigma_{l1}^2, \mu_{l2}, \sigma_{l2}^2) 
 \pi(\be_l|\bq) \pi(\bq),
 \]
 where $\pi(\cdot)$'s denote the respective prior distributions.  
 After integrating out the parameters $\mu_{l0}$, $\sigma_{l0}^2$, $\mu_{l1}$, $\sigma_{l2}^2$ and $q$, we have 
 the marginal posterior distribution of $\be_l$ given by 
\begin{equation} \label{posteqA1}
\begin{split}
\pi(\be_l|\bpsi_l) & \propto \frac{\prod_{i=0}^2\Gamma(\alpha_i+N_{li})}{\Gamma\left(\sum_{i=0}^2(\alpha_i+N_{li})\right)} \\
&\times \frac{1}{\sqrt{n_0}}(\frac{1}{\sqrt{2\pi}})^{n_0}\Gamma(\frac{n_0-1}{2}+a_2)
\left[\frac{1}{2}\sum_{\{k:e_{l}^{(k)}=-1\}}(\psi_l^{(k)})^2-\frac{(\sum_{\{k:e_{l}^{(k)}=0\}}\psi_l^{(k)})^2}{2n_0}+b_2\right]^{-(\frac{n_0-1}{2}+a_2)}\\
&\times \frac{1}{\sqrt{n_1}}(\frac{1}{\sqrt{2\pi}})^{n_1}\Gamma(\frac{n_1-1}{2}+a_2)
\left[\frac{1}{2}\sum_{\{k:e_{l}^{(k)}=0\}}(\psi_l^{(k)})^2-\frac{(\sum_{\{k:e_{l}^{(k)}=1\}}\psi_l^{(k)})^2}{2n_1}+b_2\right]^{-(\frac{n_1-1}{2}+a_2)}\\
&\times \frac{1}{\sqrt{n_2}}(\frac{1}{\sqrt{2\pi}})^{n_2}\Gamma(\frac{n_2-1}{2}+a_2)
\left[\frac{1}{2}\sum_{\{k:e_{l}^{(k)}=1\}}(\psi_l^{(k)})^2-\frac{(\sum_{\{k:e_{l}^{(k)}=1\}}\psi_l^{(k)})^2}{2n_2}+b_2\right]^{-(\frac{n_2-1}{2}+a_2)}\\
& = (G) \times (H) \times (I) \times (J), \\
\end{split}
\end{equation}
when $n_0>0$, $n_1>0$ and $n_2>0$ hold, where $n_0=\#\{k:e_{l}^{(k)}=-1\}$, 
 $n_1=\#\{k:e_{l}^{(k)}=0\}$, and $n_2=\#\{k:e_{l}^{(k)}=1\}$. 
When any $n_i=0$ where $i=0,1,2$, we exclude the term (H), (I), (J) in the equation (\ref{posteqA1}),
 respectively. Given $K$ distinct conditions, the total number of possible configurations 
 of $\be_l$ is $3^K$. For each possible configuration of $\be_l$, we can 
 calculate its posterior probability and integrated $\psi$-scores. 
 We denote the corresponding posterior probabilities by $\pi_{ld}$,   
 and denote the corresponding integrated $\psi$-scores by 
 $\bar{\bpsi}_{ld}=(\bar{\psi}_{ld}^{(1)},...,\bar{\psi}_{ld}^{(K)})$ for $d=1,2,\ldots, 3^K$. 
 Here, according to Stouffer's meta-analysis method (Stouffer et al., 1949; Mosteller and Bush, 1954), we define 
 \begin{equation} \label{inteqA}
\bar{\psi}_{ld}^{(k)}=\begin{cases}
\sum_{\{i:e_{ld}^{(i)}=-1\}} w_i \psi_l^{(i)}/\sqrt{\sum_{\{i: e_{ld}^{(i)}=0\}} w_i^2},  & \mbox{if $e_{ld}^{(k)}=-1$}, \\
\sum_{\{i:e_{ld}^{(i)}=0\}}w_i \psi_l^{(i)}/\sqrt{\sum_{\{i: e_{ld}^{(i)}=1\}} w_i^2},  & \mbox{if $e_{ld}^{(k)}=0$}, \\
\sum_{\{i:e_{ld}^{(i)}=1\}}w_i \psi_l^{(i)}/\sqrt{\sum_{\{i: e_{ld}^{(i)}=1\}} w_i^2},  & \mbox{if $e_{ld}^{(k)}=1$}, \\
\end{cases}
\end{equation} 
for $k=1,\ldots, K$, 
where the weight $w_i$ might account for the size or quality of the samples collected under each condition. 
In this paper, we set $w_i=1$ for all $i=1,\ldots, K$. 
Then the Bayesian integrated $\psi$-scores are given by  
\begin{equation} \label{aveeqA}
\hat{\psi}_l^{(k)}=\sum_{d=1}^{3^K} \pi_{ld} \bar{\psi}_{ld}^{(k)}, \quad l=1,2,\ldots,N; \ k=1,2,\ldots,K, 
\end{equation}
which has integrated information across all conditions. 
When $K$ is large, the posterior probabilities 
 $\pi_{ld}$'s can be estimated with a short MCMC run. Since the MCMC can be run in parallel for different $l$'s, 
 the computation is not a big burden when $K$ is large.   
 
\subsubsection{Spatial Prior}

 To enhance our prior knowledge that there exits a common structure 
 for all the networks from which they evolve independently, we let $\be_l$'s 
 be subject to the following prior distribution  
\begin{equation}\label{priorA3}
  p(\be_l|\bq) = \frac{(K-1)!}{N^*_{l0}!\cdot N^*_{l1}!\cdot N^*_{l2}!}\hspace{0.2cm} q_0^{N^*_{l0}}\cdot q_1^{N^*_{l1}}\cdot q_2^{N^*_{l2}},
\end{equation}
where $\sum_{i=0}^{2}q_i=1$ and $N^*_{li}=\#\{k,|e_l^{(k)}-e_l^{mod}|=i\}$, indicates the number of different edge changes at 
 condition $k$ from $e_l^{mod}$, where $k=1,2,\ldots,K$ and $e_l^{mod}$ is the mode of $\be_l$ and  
 represents the common status of the edge $l$ across all networks. 
 With this prior distribution, the posterior distribution $\pi(\be_l|\bpsi_l)$ can also 
 be expressed in the form of (\ref{posteqA1}) but with 
$N_{li}$ changes to $N^*_{li}$, where $i=0,1,2$.

 \section{Appendix: Consistency of the FBIA method.} 

 Without loss of generality, we assume that the sample size is the same under each condition, i.e., 
  $n_1=n_2=\cdots=n_k=n$. 
 Let $\{X_1^{(k)}, \ldots, X_n^{(k)}\}$ denote the samples collected under condition $k$ for 
 $k=1,\ldots, K$, where each $X_i \in \mR^p$ has a probability distribution $P^{(k)}$.
 To indicate that the dimension $p$ can grow as a function of the sample size $n$, we will
 rewrite $p$ as $p_n$, rewrite $K$ as $K_n$, $P^{(k)}$ as $P_n^{(k)}$, and the true conditional
  independence graph $\bG^{(k)}$ as $\bG_n^{(k)}$.
 Let $\mG_n^{(k)}$ denote the true correlation graph
 under condition $k$, which has the same set of nodes as $\bG_n^{(k)}$.
 Let $\gamma_{nk}$ denote a threshold value of the empirical correlation coefficient,
 let $\hat{\mE}_{\gamma_{nk}}^{(k)}$ denote the edge set of the
 network obtained through correlation thresholding at $\gamma_{nk}$, and
 let $\hat{\mE}_{\gamma_{nk},i}^{(k)}$ denote the neighborhood of node $i$ in $\hat{\mE}_{\gamma_{nk}}^{(k)}$.
 That is, we define
 \begin{equation} \label{noteq1}
  \hat{\mE}_{\gamma_{nk}}^{(k)}=\{(i,j): |\hat{r}_{ij}^{(k)}|> \gamma_{nk}\}, \quad \mbox{and} \quad 
  \hat{\mE}_{\gamma_{nk},i}^{(k)}=\{j: j \ne i, |\hat{r}_{ij}^{(k)}|> \gamma_{nk}\}.
 \end{equation}
 For convenience, we call the network with the edge set $\hat{\mE}_{\gamma_{nk}}^{(k)}$ the
 thresholding correlation network under condition $k$. Similar to (\ref{noteq1}), we define
 \begin{equation} \label{noteq2}
  \tbE_n^{(k)}=\{ (i,j): \rho_{ij|\bV\setminus \{i,j\}^{(k)} } \ne 0, \ i, j=1,\ldots,p_n\},  \ \ 
  \tmE_n^{(k)}=\{ (i,j):  r_{ij}^{(k)} \ne 0, \ i, j=1,\ldots,p_n\},
 \end{equation}
 as the edge sets of $\bG_n^{(k)}$ and $\mG_n^{(k)}$, respectively.

 To establish the consistency of the FBIA 
 method, we assume the following conditions. 

 \begin{itemize}
 \item[$(A_1)$] The distribution $P_n^{(k)}$ satisfies the conditions:
                \begin{itemize}
                \item[(i)] $P_n^{(k)}$ is multivariate Gaussian;
                \item[(ii)] $P_n^{(k)}$ satisfies the Markov property
                and faithfulness condition with respect to the
                undirected graph $\bG_n^{(k)}$ for each $k=1,2,\ldots, K_n$ and $n \in \mN$.
               \end{itemize}

 \item[$(A_2)$] The dimension $p_n=O(\exp(n^{\delta}))$ for some constant $ 0\leq \delta < 1$. 
                Note that $p_n$ is the same under all conditions.

 \item[$(A_3)$] The correlation coefficients satisfy
                \begin{equation} \label{correq1}
                  \min\{ |r_{ij}^{(k)}|; r_{ij}^{(k)} \ne 0, \ i,j=1,2,\ldots,p_n, \ i \ne j, k=1,2 \ldots, K_n \} 
                   \geq c_0 n^{-\kappa},  
                 \end{equation}
                for some constants $c_0>0$ and $0<\kappa<(1-\delta)/2$, and
                \begin{equation} \label{correq2}
                 \max\{ |r_{ij}^{(k)}|; i, j=1,\ldots, p_n, i\ne j, k=1,2,\ldots, K_n \} \leq M_r <1,
                \end{equation}
                for some constant $0<M_r<1$.
 \end{itemize}
 
 Following from the faithfulness property, we have $\tbE_n^{(k)} \subseteq \tmE_n^{(k)}$ for $k=1,2,\ldots,K_n$,
  see Liang et al. (2015) for the detail. Therefore,  there exist
  constants $c_1>0$ and $0<\kappa' \leq \kappa$ such that
  \begin{equation} \label{correq3}
       \min \{ |r_{ij}^{(n)}|; (i,j) \in \tbE_n{(k)}, i,j=1,\ldots,p_n, \ k=1,2,\ldots,K_n \} \geq c_1 n^{-\kappa'}.
  \end{equation}
  This result is quite understandable, as the directly dependent
  variables, i.e., those connected by edges in $\bG^{(n)}$, tend to have higher correlations
  than the indirectly dependent variables. 

 Lemma \ref{lem1} concerns the sure screening property of the thresholding correlation network,
 which is modified from Luo, Song and Witten (2015).
 \begin{lemma} \label{lem1}
  Assume $(A_1)$, $(A_2)$, and $(A_3)$ hold.
  Let $\gamma_{nk}=2/3 c_1 n^{-\kappa'}$.
  Then there exist constants $c_2$ and $c_3$ such that
 \[
   P(\tbE_n^{(k)} \subseteq \hat{\mE}_{\gamma_{nk}}^{(k)}) \geq 1- c_2 \exp(-c_3 n^{1-2\kappa'}),
  \]
  \[
    P(b_{\bG_n^{(k)}}(i) \subseteq \hat{\mE}_{\gamma_{nk},i}^{(k)} )  \geq 1- c_2 \exp(-c_3 n^{1-2\kappa'}),
  \]
  where $b_{\bG_n^{(k)}}(i)$ denotes the neighborhood of node $i$ in the graph $\bG_n^{(k)}$. 
 \end{lemma}
 
  Lemma \ref{lem1} implies that the $\psi$-partial correlation coefficient can be evaluated based on the
  thresholding correlation network, while ensuring
  its equivalence to the full conditional partial correlation coefficient.
  Lemma \ref{lem2} concerns the sparsity of the thresholding correlation network, which is modified 
  from Theorem 2 of Luo, Song and Witten (2015). 
 
 \begin{itemize}

 \item[$(A_4)$] There exist constants $c_4 >0$ and $0\leq \tau < 1-2 \kappa'$ such that
                $\max_k \lambda_{\max}(\Sigma_k) \leq c_4 n^{\tau}$, where $\Sigma_k$ denotes the
                covariance matrix of $P_n^{(k)}$, and $\lambda_{\max}(\Sigma_k)$ is
                the largest eigenvalue of $\Sigma_k$.  
 \end{itemize}

  \begin{lemma} \label{lem2}
  Assume $(A_1)$, $(A_2)$, $(A_3)$, and $(A_4)$ hold.
  Let $\gamma_{nk}=2/3 c_1 n^{-\kappa'}$. Then for each node $i$,
  \[
  P\left[ |\hat{\mE}_{\gamma_{nk},i}^{(k)}| \leq O(n^{2 \kappa'+\tau}) \right] \geq 1- c_2 \exp(-c_3 n^{1-2 \kappa'}), 
  \quad k=1,2,\ldots,K_n,
  \]
   where $c_2$ and $c_3$ are as given in Lemma \ref{lem1}.
 \end{lemma}

 Lemma \ref{lem1.1} concerns uniform consistency of the estimated correlation coefficient, 
 which is modified from Lemma 13.1 of B\"uhlmann and van de Geer (2011). 

 \begin{lemma} \label{lem1.1}
  Assume $(A_1)$-(i) and condition (\ref{correq2}) in $(A_3)$.
  Then, for any $0<\gamma<2$,
 \[
 \sup_{k \in \{1,2,\ldots, K_n\}} 
  \sup_{i,j \in \{1,\ldots,p_n\}} P[ |\hat{r}_{ij}^{(k)}-r_{ij}^{(k)}| > \gamma ] \leq 
  c_5 (n-2) \exp\left\{ (n-4) \log( \frac{4-\gamma^2}{4+\gamma^2}) \right\},
 \]
 for some constant $0< c_5 < \infty$ depending on $M_r$ in $(A_3)$ only.
 \end{lemma}

 \begin{itemize}
 \item[$(A_5)$] The $\psi$-partial correlation coefficients  satisfy
  \[
  \inf\{ \tilde{\psi}_{ij}^{(k)}; \tilde{\psi}_{ij} \ne 0, \ 0<|S_{ij}^{(k)}| \leq q_n, \ 
   i,j=1,\ldots, p_n, \ i\ne j, \ k=1,2,\ldots, K_n \} \geq c_6 n^{-d},
  \]
  where $q_n=O(n^{2\kappa'+\tau})$, $0<c_6<\infty$, $0<d< (1-\delta)/2$ are some constants, 
  and $S_{ij}^{(k)}$ is as defined in the $\psi$-score calculation step. In addition,
  \[
  \sup\{ \tilde{\psi}_{ij}; \ 0<|S_{ij}^{(k)}| \leq q_n, \ i,j=1,\ldots, p_n, \ i\ne j, \ k=1,2,\ldots,K_n \} 
  \leq M_{\tilde{\psi}} < 1,
  \]
  for some constant $0< M_{\tilde{\psi}}<1$. 

 \item[$(A_6)$] The number of distinct conditions $K_n=O(n^{\delta+2d+\epsilon-1})$ for some constant $\epsilon>0$ 
                such that  $\delta+2d+\epsilon-1 \geq 0$, 
                where $\delta$ is as defined in $(A_2)$ and $d$ is as defined in $(A_5)$. 
 \end{itemize} 

 Note that combining $(A_3)$ and $(A_5)$, we will get condition $(A_4)$ used by  
 Kalisch and B\"uhlmann (2007) in studying the convergence of the PC algorithm (Spirtes, Glymour and Scheines, 2000). 
 Since we used different notations for correlation coefficient and $\psi$-partial correlation coefficients, 
 we wrote them as two conditions. Condition $(A_6)$ is rather weak. As $\delta+2d<1$, we can choose 
 $\epsilon$ such that $K_n=O(1)$. This is consistent with our numerical results; the method 
 can perform very well even with a small value of $K_n$.

 \begin{lemma} \label{lem1.2} Assume $(A_1)$-(i), $(A_2)$,
   $(A_3)$ and $(A_6)$. If $\eta_{nk}=1/2 c_0 n^{-\kappa}$, then
  \[
   P[\hat{\mE}_{\eta_{nk}}^{(k)}=\tmE_n^{(k)}, k=1,\ldots, K_n] =1-o(1), \quad \mbox{as $n \to \infty$}.
   \]
 \end{lemma} 
   \begin{proof} Let $A_{ij}^{(k)}$ denote that an error event occurs
  when testing the hypotheses $H_0: r_{ij}^{(k)}=0$ versus $H_1: r_{ij}^{(k)} \ne 0$
  for variables $i$ and $j$ under condition $k$. Thus
\begin{equation} \label{Jan25eq1}
 P[\mbox{an error occurs in $\hat{\mE}_{\eta_{nk}}^{(k)}$ for $k=1,\ldots, K_n$}]  = 
 P\left[ \cup_k \cup_{i \ne j} A_{ij}^{(k)} \right] \leq O(p_n^2 K_n) \sup_k \sup_{i \ne j} P(A_{ij}^{(k)}).
\end{equation}
 Let $A_{ij}^{(k,1)}$ and $A_{ij}^{(k,2)}$ denote the false positive and
 false negative errors, respectively. Then
 \begin{equation} \label{Jan25eq2}
 A_{ij}^{(k)}=A_{ij}^{(k,1)} \cup A_{ij}^{(k,2)},
 \end{equation}
 where,
 \begin{equation} \label{Jan25eq20}
 \begin{cases} 
   \mbox{False positive error $A_{ij}^{(k,1)}$}: \ |\hat{r}_{ij}^{(k)}|> 
     \frac{c_0}{2} n^{-\kappa} \quad \mbox{and $r_{ij}^{(k)}=0$}, &  \\
   \mbox{False negative error $A_{ij}^{(k,2)}$}: \ |\hat{r}_{ij}^{(k)}|\leq  \frac{c_0}{2}n^{-\kappa}  
    \quad \mbox{and $r_{ij}^{(k)} \ne 0$}. & \\
 \end{cases}
 \end{equation}
 Then there exists  some constant $0<C<\infty$,
 \begin{equation} \label{Jan25eq3}
 \sup_k \sup_{ij} P(A_{ij}^{(k,1)}) =\sup_k \sup_{ij} P\left(|\hat{r}_{ij}^{(k)}-r_{ij}^{(k)}|> \frac{c_0}{2} n^{-\kappa} \right) 
 \leq O(n) \exp(-C  n^{1-2\kappa}),
 \end{equation}
 using Lemma \ref{lem1.1} and the fact that
 $\log((4-a^2)/(4+a^2)) \sim -a^2/2$ as $a \to 0$.
 Furthermore,
  \begin{equation} \label{Jan25eq21}
 \sup_k \sup_{ij} P(A_{ij}^{(k,2)}) =\sup_k \sup_{ij} P\left(|\hat{r}_{ij}^{(k)}| \leq \frac{c_0}{2} n^{-\kappa} \right) 
 \leq \sup_k \sup_{ij} P\left( |\hat{r}_{ij}^{(k)}-r_{ij}^{(k)}| > \frac{c_0}{2} n^{-\kappa} \right),
 \end{equation}
 since, by $(A_3)$, $\min_k \min_{ij} |r_{ij}^{(k)}| \geq c_0 n^{-\kappa}$ in this case.
 By Lemma \ref{lem1.1}, we have
 \begin{equation} \label{Jan25eq4}
 \sup_k \sup_{ij} P(A_{ij}^{(k,2)}) \leq O(n) \exp(-C n^{1-2\kappa} ), 
 \end{equation}
 for some $0<C<\infty$.
 As a summary of (\ref{Jan25eq1})--(\ref{Jan25eq4}), we have
 \begin{equation} \label{Jan25eq5}
   P[\mbox{an error occurs in $\hat{\mE}_{\eta_{nk}}^{(k)}$ for $k=1,2,\ldots,K_n$}] \leq O(p_n^2 K_n n) \exp(-C n^{1-2\kappa}) =o(1), 
  \end{equation}
  because  $0< \kappa < (1-\delta)/2$ by $(A_3)$, $K_n=O(n^{\delta+2d+\epsilon-1})$ by $(A_6)$, 
  and $\log(p_n)=n^{\delta}$ by $(A_2)$.
  This concludes the proof.
\end{proof}  

  As explained before, we have $\tbE_n^{(k)} \subseteq \tmE_n^{(k)}$ for $k=1,2,\ldots,K_n$. 
  Further, it follows from Lemma \ref{lem1.2} that
 \begin{equation} \label{correq4}
  P[ \tbE_n^{(k)} \subseteq \hat{\mE}_{\eta_{nk}}^{(k)},k=1,\ldots,K_n] =1-o(1). 
 \end{equation}
  Therefore, based on Lemma \ref{lem1}, Lemma \ref{lem2} and (\ref{correq4}), we
  propose to restrict the neighborhood size of each node  to be
 \begin{equation} \label{neisize}
  \min\left\{ |\hat{\mE}_{\eta_{nk}, i}^{(k)}|,  \frac{n}{\xi_n \log(n)} \right \},
 \end{equation}
  where $\xi_n$ is a constant. 
  The value of $\eta_{nk}$ can be determined
  through a simultaneous test for
  the hypotheses $H_0: r_{ij}^{(k)}=0 \leftrightarrow H_1: r_{ij}^{(k)}\ne 0$, $1\leq i <j \leq p_n$,
  at a significance level of $\alpha_1$.
  Our experience shows that the rule (\ref{neisize}) can perform much better than the rule $n/[\xi_n \log(n)]$,
  especially when $n$ is large.
 
 Under condition $(A_5)$, we have that 
 the minimum $\psi$-score for the edges with $\tilde{\psi}_{ij}^{(k)} \ne 0$ is given by 
 \[
 \min_k \min_{i\ne j} \psi_{ij}^{(k)} = c_7 n^{1/2-d},
 \]
 for some constant $c_7$. This can be obtained by plugging  
 the lower bound of $\tilde{\psi}_{ij}^{(k)}$ into (\ref{score}).    
 In what follows, for convenience, we will re-denote $\mu_{l0}$ by $\mu_{ij,0}$, 
 re-denote $\mu_{l1}$ by $\mu_{ij,1}$, and re-denote 
 $\hat{\psi}_l^{(k)}$ by  $\hat{\psi}_{ij}^{(k)}$ for the corresponding pair $(i,j)$. 
 Let $\hat{\psi}_{B,ij}^{(k)}$ denote the Bayesian estimator of $\mu_{ij}^{(k)}$ 
 (with the Bayesian method described in Section 2.2), where 
 $\mu_{ij}^{(k)}=\mu_{ij,0}$ or $\mu_{ij,1}$ as defined in (\ref{plugin}).   
 Theoretically we have $\mu_{ij,0}=0$ and $\mu_{ij,1}> c_7 n^{1/2-d}$. Following the 
 property of Bayesian estimator, we have that $\hat{\psi}_{B,ij}^{(k)}$  
 is consistent and has a variance of order $O(1/K_n)$. 
 Note that the integrated $\psi$-score $\hat{\psi}_{ij}^{(k)}$ is a boosted 
 version of $\hat{\psi}_{B,ij}^{(k)}$; which amplifies $\hat{\psi}_{B,ij}^{(k)}$ 
 by a factor between 1 and $\sqrt{K_n}$. 
 Such amplification helps to 
 improve the power of the proposed method by reducing the false negative errors.   
  
  Let $\zeta_n$ denote the threshold value of the integrated $\psi$-scores used in
  the joint edge detection step.
  Let $\hat{\bE}_{\zeta_n}^{(k)}$ denote the partial correlation network obtained through thresholding
  integrated $\psi$-scores. That is, we define
 \[
  \hat{\bE}_{\zeta_n}^{(k)}=\{(i,j): |\hat{\psi}_{ij}^{(n)}|> \zeta_n, \ i, j=1,2,\ldots,p_n \}, \quad k=1,2,\ldots,K_n.
 \] 

 Let $\hat{\mE}_*^{(k)}$ denote the edge set of a correlation network for which each node has a degree of $O(n/\log(n))$,
 adjacent with $O(n/\log(n))$ highest correlated nodes.
 It follows from Lemma \ref{lem2}, $(A_2)$ and $(A_6)$ that
 \begin{equation} \label{neweq2}
  P[ \tbE_k^{(k)} \subseteq \hat{\mE}_*^{(k)},k=1,2,\ldots,K_n ] \geq K_n \left[1- c_2 p_n \exp(-c_3 n^{1-2 \kappa'})\right]-(K_n-1)
  =1-o(1).
 \end{equation}
 Lemma \ref{lem4} establishes the consistency of $\hat{\bE}_{\zeta_n}^{(k)}$ conditioned on
 $\tbE_n^{(k)} \subseteq \hat{\mE}_{*}^{(k)} \cap \hat{\mE}_{\eta_{nk}}^{(k)}$ for
 all $k=1,2,\ldots,K_n$.  Essentially, it shows that for all graphs there exists a common threshold $\zeta_n$ 
 with respect to which the Bayesian integrated $\psi$-scores are separable in probability 
 for the linked and non-linked pairs of nodes. 

 \begin{lemma} \label{lem4} Assume $(A_1)$--$(A_6)$  hold and
   $\tbE_{n}^{(k)} \subseteq \hat{\mE}_{*}^{(k)} \cap \hat{\mE}_{\eta_{nk}}^{(k)}$ is true for all $k=1,2,\ldots,K_n$.  
   Let $\zeta_n=\frac{1}{2} c_7 n^{1/2-d}$, then
 \[
  P\left[ \hat{\bE}_{\zeta_n}^{(k)}= \tbE_n^{(k)}, k=1,\ldots, K_n| \tbE_n^{(k)} \subseteq \hat{\mE}_{*}^{(k)} 
  \cap \hat{\mE}_{\eta_{nk}}^{(k)}, k=1,\ldots,K_n \right] = 1- o(1),  \quad 
 \mbox{as $n\to 1$.}
 \]
 \end{lemma} 
 \begin{proof} Let $A_{ij}^{(k)}$ denote that an error event occurs
  when testing the hypotheses $H_0: \mu_{ij}^{(k)}=0$ versus $H_1: \mu_{ij}^{(k)} \ne 0$
  for variables $i$ and $j$ under condition $k$. Thus
\begin{equation} \label{Jan25eq1.2}
 P[\mbox{an error occurs in $\hat{\bE}_{\zeta_n}^{(k)}$ for $k=1,\ldots, K_n$}]  = 
 P\left[ \cup_k \cup_{i \ne j} A_{ij}^{(k)} \right] \leq O(p_n^2 K_n) \sup_k \sup_{i \ne j} P(A_{ij}^{(k)}).
\end{equation}
 Let $A_{ij}^{(k,1)}$ and $A_{ij}^{(k,2)}$ denote the false positive and
 false negative errors, respectively. Then
 \begin{equation} \label{Jan25eq2.2}
 A_{ij}^{(k)}=A_{ij}^{(k,1)} \cup A_{ij}^{(k,2)},
 \end{equation}
 where,
 \begin{equation} \label{Jan25eq20.2}
 \begin{cases} 
   \mbox{False positive error $A_{ij}^{(k,1)}$}: \ |\hat{\psi}_{ij}^{(k)}|> 
     \frac{c_7}{2} n^{1/2-d} \quad \mbox{and $\mu_{ij}^{(k)}=0$}, &  \\
   \mbox{False negative error $A_{ij}^{(k,2)}$}: \ |\hat{\psi}_{ij}^{(k)}|\leq  \frac{c_7}{2}n^{1/2-d}  
    \quad \mbox{and $\mu_{ij}^{(k)} \ne 0$}. & \\
 \end{cases}
 \end{equation}
 Then we have 
 \begin{equation} \label{Jan25eq3.2}
 \begin{split}
 \sup_k \sup_{ij} P(A_{ij}^{(k,1)}) & =\sup_k \sup_{ij} P\left(|\hat{\psi}_{ij}^{(k)}-\mu_{ij,0}^{(k)}|> \frac{c_7}{2} n^{1/2-d} \right) \\
  & \leq \sup_k \sup_{ij} P\left(K_n^{1/2} |\hat{\psi}_{B,ij}^{(k)}-\mu_{ij,0}^{(k)}|> \frac{c_7}{2} n^{1/2-d} \right) \\
  & \leq  2 \exp\{-C n^{1-2d} \}, 
 \end{split}
 \end{equation}
 where $\mu_{ij,0}^{(k)}=0$, 
  $C$ denotes a constant, and the last inequality follows from the concentration inequality of the normal distribution, 
 i.e., $P(|Z|\geq z) \leq 2 e^{-z^2/2}$ for all $z>0$, where $Z$ denotes a standard normal random variable. 
 Furthermore, we have
  \begin{equation} \label{Jan25eq21.2}
 \begin{split}
 \sup_k \sup_{ij} P(A_{ij}^{(k,2)}) & =\sup_k \sup_{ij} P\left(|\hat{\psi}_{ij}^{(k)}| \leq \frac{c_7}{2} n^{1/2-d} \right)  
   \leq \sup_k \sup_{ij} P\left( |\hat{\psi}_{B,ij}^{(k)}| \leq \frac{c_7}{2} n^{1/2-d} \right) \\
  & \leq \sup_k \sup_{ij} P\left( |\hat{\psi}_{B,ij}^{(k)}-\mu_{ij,1}^{(k)}| \geq \frac{c_7}{2} n^{1/2-d} \right),  
 \end{split}
 \end{equation}
 following from $\min_k \min_{ij} |\mu_{ij,1}^{(k)}| \geq c_7 n^{1/2-d}$. Note that the first 
 inequality of (\ref{Jan25eq21.2}) implies that the meta-analysis step indeed reduces the false negative error. 
  Further, by the concentration inequality of the normal distribution, we have 
 \begin{equation} \label{Jan25eq4.2}
 \sup_k \sup_{ij} P(A_{ij}^{(k,2)}) \leq 2 \exp(-C K_n n^{1-2d} ), 
 \end{equation}
 for some constant $0<C<\infty$.  Note that the variance of $\hat{\psi}_{B,ij}^{(k)}$ is of order $O(1/K_n)$. 
 
 As a summary of (\ref{Jan25eq1.2})--(\ref{Jan25eq4.2}), we have
 \begin{equation} \label{Jan25eq5.2}
   P[\mbox{an error occurs in $\hat{\mE}_{\eta_{nk}}^{(k)}$ for $k=1,2,\ldots,K_n$}] \leq O(p_n^2 K_n) 
   \left(e^{-C n^{1-2d}}+ e^{-C K_n n^{1-2d}}\right) =o(1), 
  \end{equation}
  because $0< d < (1-\delta)/2$ by $(A_5)$,  $K_n=O(n^{\delta+2d-1+\epsilon})$ by $(A_6)$, and 
  $\log(p_n)=n^{\delta}$ by $(A_2)$.
  This concludes the proof.
\end{proof}

\paragraph{Proof of Theorem \ref{conthem}} 
 By invoking (\ref{correq4}), (\ref{neweq2}) and Lemma \ref{lem4}, we have
 \[ 
 \begin{split}
 P\left[ \hat{\bE}_{\zeta_n}^{(k)} =\tbE_n^{(k)}, k=1,2,\ldots,K_n \right] & \geq  
 P\left[ \hat{\bE}_{\zeta_n}^{(k)}= \tbE_n^{(k)}, k=1,2,\ldots,K_n | \tbE_n^{(k)} \subseteq \hat{\mE}_*^{(k)} 
  \cap \hat{\mE}_{\eta_{nk}}^{(k)},  k=1,2,\ldots,K_n \right]  \\
 & \times  P\left[ \tbE_n^{(k)} \subseteq \hat{\mE}_*^{(k)} \cap \hat{\mE}_{\eta_{nk}}^{(k)}, k=1,2,\ldots,K_n \right]   \\
 & \geq  \left[  1- o(1) \right] \left[ 1- o(1)+ 1-o(1)-1 \right] = 1-o(1), \\
\end{split}
 \]
 while concludes the proof.

\section*{References} 

\begin{description}

\item[] Barab\'{a}si, A. and Albert, R. (1999). Emergence of scaling in random networks. {\it Science} {\bf 286}, 509.

\item[] Benjamini, Y. and Yekutieli, D. (2001). The control of the false discovery rate in multiple testing 
        under dependency. {\it Annals of Statistics}, {\bf 29}(4), 1165-1188.
 
\item[] Bonifacio, E. (2015). Predicting type 1 diabetes using biomarkers. {\it Diabetes Care}, 38(6), 989-996.

\item[] B\"uhlmann, P. and van de Geer, S. (2011). {\it Statistics for High-Dimensional Data: Methods, Theory and 
        Applications}. Berlin: Springer-Verlag. 
 
\item[] Cai, T., Liu, W., and Luo, X. (2011). A constrained $l_1$ minimization approach to sparse precision matrix estimation. {\it Journal of the American Statistical Association}, 106(494), 594-607.

\item[] Danaher, P. (2012). JGL: Performs the joint graphical lasso for sparse inverse covariance estimation on 
        multiple classes. R package downloadable at https://cran.r-project.org/web/packages/JGL /index.html. 

\item[] Danaher, P., Wang, P., and Witten, D. M. (2014). The joint graphical lasso for inverse covariance estimation across multiple classes. {\it Journal of the Royal Statistical Society: Series B (Statistical Methodology)}, 76(2), 373-397.

\item[] Davis, J. and Goadrich, M. (2006). The relationship between Precision-Recall and ROC curves. 
  In {\it Proceedings of the 23rd international conference on Machine learning}, pp. 233-240.

\item[]
Fan, J., and Song, R. (2010). Sure independence screening in generalized linear models with NP-dimensionality. \ANNALS, 38(6), 3567-3604.

\item[] Friedman, J., Hastie, T. and Tibshirani, R. (2008). Sparse inverse covariance estimation
 with the graphical lasso. {\it Biostatistics}, {\bf 9}, 432-441.

\item[] Guo, J., Levina, E., Michailidis, G., and Zhu, J. (2011). Joint estimation of multiple graphical models. {\it Biometrika}, 98, 1-15.

\item[] Kalisch, M. and B\"uhlmann, P. (2007). Estimating high-dimensional directed acyclic graphs 
   with the PC algorithm. \JMLR, 8, 613-636. 

\item[] Kolaczyk, E. D. (2009). {\it Statistical Analysis of Network Data: Methods and Models.} New York, NY, USA: Springer.

\item[] Lee, J. and Hastie, T.J. (2015). Learning the structure of mixed graphical models. 
        \JCGS, 24, 230-253. 

\item[] Lee, H.S., Burkhardt, B.R., McLeod, W., Smith, S., Eberhard, C., Lynch, K., Hadley, D., 
 Rewers, M., Simell, O., She, J.X., Hagopian, B., Lernmark, A., Akolkar, B., Ziegler, A.G., Krischer, J.P.,
 TEDDY study group (2014). Biomarker discovery study design for type 1 diabetes in  The 
 Environmental Determinants of Diabetes in the Young (TEDDY) study. 
 {\it Diabetes Metab Res Rev}, 30(5), 424-434.

\item[] Liang, F., Liu, C., and Carroll, R. (2010). {\it Advanced Markov Chain Monte Carlo Methods: Learning from Past Samples}. 
 John Wiley \& Sons.

\item[] Liang, F., Song, Q. and Qiu, P. (2015). An Equivalent Measure of Partial Correlation Coefficients for
   High Dimensional  Gaussian Graphical Models.  \JASA, {\bf 110}, 1248-1265.

\item[] Liang, F. and Zhang, J. (2008). Estimating the false discovery rate using the
  stochastic approximation algorithm.  {\it Biometrika}, {\bf 95}, 961-977.

\item[] Lin, Z., Wang, T., Yang, C., and Zhao, H. (2017). 
   On Joint Estimation of Gaussian Graphical Models for Spatial and Temporal Data. {\it Biometrics}, 73, 769-779.

\item[]
Liu, H., Lafferty, J. and Wasserman, L. (2009). The nonparanormal: Semiparametric estimation of high dimensional undirected graphs. {\it Journal of Machine Learning Research}, 10(Oct), 2295-2328.

\item[] Luo, S., Song, R., and Witten, D. (2015). Sure screening for Gaussian graphical models. 
        Available at arXiv:1407.7819v1.

\item[] Ma, J., and Hart, G. W. (2013). Protein O-GlcNAcylation in diabetes and diabetic complications. {\it Expert review of proteomics}, 10(4), 365-380.

\item[] Meinshausen, N. and B\"uhlmann, P. (2006). High-dimensional graphs and variable selection
     with the Lasso. {\it Annals of Statistics}, {\bf 34}, 1436-1462.

\item[]
Mosteller, F. and Bush, R.R. (1954). Selected quantitative techniques. In: {\it Handbook of Social Psychology}, Vol. 1 (G.Lindzey, ed.), pp. 289-334.
  Addison-Wesley, Cambridge, Mass.

\item[]
Orilieri, E., Cappellano, G., Clementi, R., Cometa, A., Ferretti, M., Cerutti, E. et al. (2008). Variations of the perforin gene in patients with type 1 diabetes. Diabetes.

\item[] Peterson, C., Stingo, F. C., and Vannucci, M. (2015). Bayesian inference of multiple Gaussian graphical models. 
   {\it Journal of the American Statistical Association}, 110(509), 159-174.

\item[] Qiu, H., Han, F., Liu, H., and Caffo, B. (2016). Joint estimation of multiple graphical models from high dimensional time series.   {\it Journal of the Royal Statistical Society: Series B (Statistical Methodology)}, 78(2), 487-504.

\item[] Saito, T. and Rehmsmeier, M. (2015). The precision-recall plot is more informative than the ROC plot 
  when evaluating binary classifiers on imbalanced datasets. {\it PloS One}, 10(3), e0118432.

\item[] Schadt, Eric E., et al. (2008). Mapping the genetic architecture of gene expression in human liver. 
        {\it PLoS Biol}, 6(5):e107.

\item[] Shaddox, E., Stingo, F.C., Peterson, C.B., Jacobson, S., Cruickshank-Quinn, C., Kechris, K., 
        Bowler, R., and Vannucci, M. (2016). A Bayesian approach for learning gene networks underlying 
        disease severity in COPD. {\it Stat. Biosci.}, DOI 10.1007/s12561-016-9176-6.

\item[] Spirtes, P., Glymour, C., and Scheines, R. (2000). {\it Causation, Prediction, and Search} (2nd Edition). 
        The MIT Press.  

\item[] Stouffer S. et al . (1949). {\it The American Soldier: Adjustment during army life.}  Vol. 1. Princeton University Press, Princeton.

\item[] Wang, H. (2015). Scaling it up: Stochastic search structure learning in graphical models. {\it Bayesian Analysis}, 
        10(2), 351-377.

\item[]
Xie, Y., Liu, Y., and Valdar, W. (2016). Joint estimation of multiple dependent Gaussian graphical models with applications to mouse genomics. 
  {\it Biometrika}, 103, 493-511.

\item[] Yuan, M. and Lin, Y. (2007). Model selection and estimation in the Gaussian graphical model.
       {\it Biometrika}, {\bf 94}, 19-35.

\item[] Zhou, S., Lafferty, J., and Wasserman, L. (2010). Time varying undirected graphs. {\it Machine Learning}, 80(2-3), 295-319.
\end{description}

\end{document}